\documentclass[12pt]{article}
\usepackage[margin=1in]{geometry}
\usepackage{times}
\usepackage{color}
\usepackage{graphicx}
\usepackage[colorlinks,citecolor=blue,urlcolor=blue,linkcolor=blue,bookmarks=false,hypertexnames=true]{hyperref} 

\usepackage{multirow}
\usepackage{amsmath}
\usepackage{amsthm}
\usepackage{amssymb}
\usepackage{natbib}
\usepackage{rotating}
\usepackage{appendix}
\usepackage{titling}  

\newtheorem{cor}{Corollary}

\usepackage{authblk}

\DeclareMathOperator*{\argmin}{arg\,min}

\newtheorem{prop}{Proposition}
\newtheorem{lem}{Lemma}

\begin{document}
\title{Two-Stage Trigonometric Regression for Modeling Circadian Rhythms}
\author[1]{Michael T. Gorczyca\footnote{Corresponding author: mtg62@cornell.edu}}
\author[2]{Jenna D. Li}
\author[3]{Charissa M. Newkirk}
\author[4]{Arjun S. Srivatsa}
\author[5]{Hugo F. M. Milan}

\affil[1]{MTG Research Consulting}
\affil[2]{Allegheny-Singer Research Institute}
\affil[3]{The Ohio State University}
\affil[4]{Stanford University}
\affil[5]{Federal University of Rondonia}

\date{\today}
\maketitle

\begin{abstract}
    \noindent Gene expression levels, hormone secretion, and internal body temperature each oscillate over an approximately 24-hour cycle, or display circadian rhythms. Many circadian biology studies have investigated how these rhythms vary across cohorts, uncovering associations between atypical rhythms and diseases such as cancer, metabolic syndrome, and sleep disorders. A challenge in analyzing circadian biology data is that the oscillation peak and trough times for a measured phenomenon differ across individuals. If these individual-level differences are not accounted for in trigonometric regression, which is prevalent in circadian biology studies, then estimates of the population-level amplitude parameters can suffer from attenuation bias, or a decrease in magnitude towards zero. This attenuation bias could lead to inaccurate study conclusions. To address attenuation bias, we propose a refined two-stage (RTS) method for trigonometric regression given longitudinal data obtained from each individual participating in a study. In the first stage, the parameters of individual-level models are estimated. In the second stage, transformations of these individual-level parameter estimates are aggregated to produce population-level parameter estimates for inference. Simulation studies show that our RTS method mitigates bias in parameter estimation, obtains greater statistical power, and maintains appropriate type I error control when compared to the standard two-stage (STS) method, which ignores individual-level differences in peak and trough times. The only exception for parameter estimation and statistical power occurs when the oscillation amplitudes are weak relative to random variability in the data and the sample size is small. Illustrations with cortisol level data and heart rate data show that our RTS method obtains larger population-level amplitude parameter estimates and smaller $p$-values for multiple hypothesis tests when compared to the STS method.
\end{abstract}

\section{Introduction}\label{sec1}

A recurring theme in studying biological phenomena such as gene expression levels, hormone secretion, and internal body temperature is that they display non-random oscillations over an approximately 24-hour cycle \citep{Andreani2015, Beersma2007, Marcheva2013}. These oscillations, known as ``circadian rhythms,'' are influenced in part by daily changes in environmental exposures, such as variations in light exposure over the 24-hour day-night cycle \citep{Mistlberger2007}. A notable feature of these rhythms is their association with an individual's health. For example, some studies have found that atypical rhythms are associated with diseases such as cancer \citep{Altman2016, Truong2016}, metabolic syndrome \citep{Morris2016}, and sleep disorders \citep{Walker2020}. In addition, the efficacy of treatments such as chemotherapy \citep{Dallmann2014, Haus2009}, heart surgery \citep{Montaigne2018, Young2023}, and vaccines \citep{Long2016} appear to vary based on the time of day that treatment is administered. These findings have contributed to an increase in circadian biology research over the past two decades \citep{Hughes2017, Zong2023}, with many studies conducted to inform the development of improved treatment strategies for diseases \citep{Chan2017, Chauhan2017, Halberg2013, Haus2009, Li2013}.

When conducting a circadian biology study, an investigator will often use trigonometric regression to model how a biological phenomenon oscillates over time and make inferences about the modeled oscillations \citep{Bingham1982, Cornelissen2014, Tong1976}. However, a challenge with performing trigonometric regression on circadian biology study data is that the times at which a biological phenomenon peaks and troughs can be distinct across individuals. Specifically, the peak and trough times could be distinct across individuals due to differences in genetic makeup \citep{Hsu2015}, differences in age \citep{Kennaway2023}, and differences in environmental exposures \citep{Khodasevich2021, Phillips2019}. If these individual-level differences in oscillations are ignored during model fitting, then the resulting parameter estimates would be attenuated, or biased toward zero, which could lead to incorrect study conclusions \citep{Gorczycaa2024, Gorczycab2024, Sollberger1962, Weaver1995}.

This article is motivated by the attenuation bias caused by individual-level differences in how a biological phenomenon oscillates, and considers a scenario where data are collected longitudinally from each individual participating in a study. For this scenario, we initially assess the suitability of the standard two-stage (STS) method for trigonometric regression, which involves first estimating individual-level parameters and then averaging them to obtain population-level estimates, for addressing individual-level differences in oscillations (\citealp[Chapter 5]{Davidian1995}; \citealp{Sheiner1980}). This assessment quantifies how population-level parameter estimates produced by the STS method for trigonometric regression suffer from attenuation bias when the peak and trough times are different for each individual. To address this attenuation bias, we propose a refined two-stage (RTS) method for trigonometric regression. This RTS method instead averages transformations of these individual-level parameter estimates.

The remainder of this article is organized as follows. In Section \ref{sec:2}, an overview of the STS method, the limitations of the STS method for trigonometric regression, and our RTS method for trigonometric regression are presented. In Section \ref{sec:3}, Monte Carlo simulation studies are performed to assess the utility of our RTS method relative to the STS method. In Section \ref{sec:4}, our RTS method is applied on cortisol level data and heart rate data. Finally, in Section \ref{sec:5}, our RTS method and directions for future work are discussed.

\section{Methodology} \label{sec:2}

\subsection{Background and Notation for the Standard Two-Stage Method} \label{sec:2.1}

\subsubsection{Parameter Estimation} \label{sec:2.1.1}

Suppose a longitudinal circadian biology experiment is conducted on two cohorts: a case cohort consisting of $M^{(1)}$ individuals (cohort $c = 1$) and a control cohort consisting of $M^{(0)}$ individuals (cohort $c = 0$). For each $i$-th individual from the $c$-th cohort, the investigator obtains $n^{(c)}_i$ measurements of a biological phenomenon. We denote the time of the $j$-th measurement for the $i$-th individual from the $c$-th cohort as $X^{(c)}_{i,j}$. The corresponding measurement recorded at $X^{(c)}_{i,j}$ is denoted as $Y^{(c)}_{i,j}$. 

A common assumption in circadian biology studies is that each measurement $Y_{i,j}^{(c)}$ is correctly represented by an amplitude-phase trigonometric regression model of order $K$ \citep{Bingham1982, Cornelissen2014, Tong1976}. In this article, we adopt this assumption while allowing the oscillations to vary across individuals. Specifically, we assume that $Y_{i,j}^{(c)}$ is correctly represented by an individual-level amplitude-phase trigonometric regression model of order $K$, or
\begin{align}
    Y^{(c)}_{i,j} &= f(X^{(c)}_{i,j}, \theta^{(c)}_i) + \epsilon^{(c)}_{i,j} \nonumber \\
    &= \left\{\theta^{(c)}_{i,0} + \sum_{k=1}^K\theta^{(c)}_{i,2k-1}\cos\left(\frac{k \pi X^{(c)}_{i,j}}{12}+\theta^{(c)}_{i,2k}\right)\right\} + \epsilon^{(c)}_{i,j}. \label{eq:trig_mod}
\end{align}
In this model, the parameter vector $\theta_i^{(c)}$ characterizes how a biological phenomenon oscillates for the $i$-th individual in the $c$-th cohort. To be precise, the parameter $\theta^{(c)}_{i,0}$ in (\ref{eq:trig_mod}) represents the individual-level midline of the modeled oscillation, and each term $\theta^{(c)}_{i,2k-1} \cos\{(k\pi X^{(c)}_{i,j}/12) + \theta^{(c)}_{i,2k} \}$ for $k \in \{1,\ldots,K\}$ corresponds to an individual-level oscillation harmonic. For the $k$-th harmonic, $\theta^{(c)}_{i,2k-1}$ denotes the individual-level amplitude parameter, or the deviation from the midline to the $k$-th harmonic's peak; and $\theta^{(c)}_{i,2k}$ denotes the individual-level phase-shift parameter, which determines the times at which the $k$-th harmonic peaks in a 24-hour oscillation cycle \citep{Cornelissen2014, Gorczyca2025}. The term $\epsilon^{(c)}_{i,j}$ in (\ref{eq:trig_mod}) represents independent random noise associated with the $j$-th measurement for the $i$-th individual from the $c$-th cohort. The expectation of this random noise $\mathbb{E}(\epsilon^{(c)}_{i,j}) = 0$ and the variance $\mathrm{Var}(\epsilon^{(c)}_{i,j}) = (\sigma_i^{(c)})^2$. 

While the amplitude-phase model in (\ref{eq:trig_mod}) is biologically interpretable, the model is nonlinear with respect to its parameters, which would complicate parameter estimation without prior knowledge about the true parameter estimands \citep[Theorem 6.7]{Boos2013}. To simplify parameter estimation, many investigators instead estimate the parameters of a linear trigonometric regression model \citep{Archer2014, delolmo2022, Fontana2012, Hou2021, MllerLevet2013}. The individual-level linear model corresponding to (\ref{eq:trig_mod}) is defined as
\begin{align}
    Y_{i,j}^{(c)} &= h(X^{(c)}_{i,j}, \gamma^{(c)}_i) \nonumber \\
    &= \gamma^{(c)}_{i,0} + \sum_{k=1}^K\left\{\gamma^{(c)}_{i,2k-1}\sin\left(\frac{k \pi X^{(c)}_{i,j}}{12}\right) + \gamma^{(c)}_{i,2k}\cos\left(\frac{k \pi X^{(c)}_{i,j}}{12}\right)\right\}, \label{eq:lin_mod}
\end{align}
where $\gamma^{(c)}_{i}$ is an alternative parameter vector for the $i$-th individual in the $c$-th cohort. For the two individual-level models presented in (\ref{eq:trig_mod}) and (\ref{eq:lin_mod}), the intercept terms $\theta^{(c)}_{i,0}$ and $\gamma^{(c)}_{i,0}$ are equal to each other, and the following identities can be used to convert the remaining parameters of one model to the parameters of the other model:
\begin{equation}
\begin{aligned}
   \gamma^{(c)}_{i, 2k-1} &= -\theta^{(c)}_{i, 2k-1}\sin(\theta^{(c)}_{i, 2k}), \quad \quad \quad \ \ \gamma^{(c)}_{i, 2k} = \theta^{(c)}_{i, 2k-1}\cos(\theta^{(c)}_{i, 2k}), \\
   \theta^{(c)}_{i, 2k-1} &= \sqrt{(\gamma^{(c)}_{i, 2k-1})^2+(\gamma^{(c)}_{i, 2k})^2}, \quad \ \ \ \theta^{(c)}_{i, 2k} = \mathrm{atan2}(-\gamma^{(c)}_{i, 2k-1}, \gamma^{(c)}_{i, 2k}). \label{eq:alt_to_orig}
\end{aligned}
\end{equation}
It is noted that $\mathrm{atan2}(\kappa, \xi)$ in (\ref{eq:alt_to_orig}) denotes the two-argument arctangent function with arguments $\kappa$ and $\xi$ \citep{Bingham1982, Cornelissen2014, Tong1976}.

This article adopts the standard two-stage (STS) method to obtain population-level parameter estimates from these individual-level parameters and make population-level inferences (\citealp[Chapter 5]{Davidian1995}; \citealp{Sheiner1980}; \citealp{Steimer1984}). To obtain population-level parameter estimates, the STS method requires the assumption that the individual-level parameter vector $\gamma^{(c)}_i$ can be modeled as
\begin{align}
\gamma^{(c)}_i = \alpha^{(c)} + a^{(c)}_i. \label{eq:mixed_lin}
\end{align}
Here, $\alpha^{(c)}$ is interpreted as a non-random population-level parameter vector, and $a^{(c)}_i$ is interpreted as an individual-level random vector with $\mathbb{E}(a^{(c)}_i) = 0$ and a constant covariance matrix $\mathrm{Var}(a^{(c)}_i)=D^{(c)}$. It is noted that these assumptions about $\alpha^{(c)}$ and $a^{(c)}_i$ are common in mixed-effects modeling, where $\alpha^{(c)}$ is referred to as the fixed effects and $a_i^{(c)}$ is referred to as the random effects \citep{Davidian1995, Hedeker2006, McCulloch2000}. The STS method involves first estimating individual-level parameters by minimizing squared loss, or computing
\begin{align*}
\hat{\gamma}_i^{(c)} = \argmin_{\delta \in \mathbb{R}^{(2K+1)\times 1}} \ \ \sum_{j=1}^{n_i^{(c)}} \left\{ Y_{i,j}^{(c)} - h(X^{(c)}_{i,j}, \delta) \right\}^2,
\end{align*}
with $\hat{\gamma}_i^{(c)}$ denoting the estimate of $\gamma_i^{(c)}$. The population-level parameter vector $\alpha^{(c)}$ is then estimated by averaging over all individual-level estimates, or computing
\begin{align}
\hat{\alpha}^{(c)} = \frac{1}{M^{(c)}}\sum_{i=1}^{M^{(c)}} \hat{\gamma}_i^{(c)}. \label{eq:sts_pop}
\end{align}
The empirical covariance matrix for $\hat{\alpha}^{(c)}$ is defined as
\begin{align*}
\mathrm{Var}(\hat{\alpha}^{(c)}) = \frac{1}{M^{(c)}}\left(\hat{D}^{(c)} + \frac{1}{M^{(c)}}\sum_{i=1}^{M^{(c)}}\hat{\Sigma}^{(c)}_i\right), 
\end{align*}
which consists of two components. The first component is the matrix
\begin{align*}
\hat{D}^{(c)} &= \frac{1}{M^{(c)}-1}\sum_{i=1}^{M^{(c)}}(\hat{\gamma}_i^{(c)}-\hat{\alpha}^{(c)})(\hat{\gamma}_i^{(c)}-\hat{\alpha}^{(c)})^T,
\end{align*} 
which is a between-individual covariance estimate, or an estimate of $\mathrm{Var}(a_i^{(c)})=D^{(c)}$. The second component is the average of within-individual covariance estimates $\hat{\Sigma}^{(c)}_i$ \citep[Chapter 5]{Davidian1995}. Specifically, $\hat{\Sigma}^{(c)}_i$ is defined as
\begin{align*}
\hat{\Sigma}^{(c)}_i = \frac{(\hat{\sigma}^{(c)}_i)^2}{n^{(c)}_i}\left\{(W^{(c)}_i)^T (W^{(c)}_i)\right\}^{-1},
\end{align*}
where
\begin{align*}
W^{(c)}_i = \begin{bmatrix}
1 & \sin\left(\frac{\pi X^{(c)}_{i,1}}{12}\right) & \cos\left(\frac{\pi X^{(c)}_{i,1}}{12}\right) & \cdots & \sin\left(\frac{K\pi X^{(c)}_{i,1}}{12}\right) & \cos\left(\frac{K\pi X^{(c)}_{i,1}}{12}\right) \\
\vdots & \vdots & \vdots & \ddots & \vdots & \vdots \\
1 & \sin\left(\frac{\pi X^{(c)}_{i,n^{(c)}_i}}{12}\right) & \cos\left(\frac{\pi X^{(c)}_{i,n^{(c)}_i}}{12}\right) & \cdots & \sin\left(\frac{K\pi X^{(c)}_{i,n^{(c)}_i}}{12}\right) & \cos\left(\frac{K\pi X^{(c)}_{i,n^{(c)}_i}}{12}\right)
\end{bmatrix}
\end{align*}
represents an individual-level design matrix for regression, and
\begin{align*}
(\hat{\sigma}^{(c)}_i)^2 = \frac{1}{n_i^{(c)} - 2K - 1} \sum_{j=1}^{n_i^{(c)}} \left\{ Y_{i,j}^{(c)} - h(X^{(c)}_{i,j}, \hat{\gamma}_i^{(c)}) \right\}^2
\end{align*}
represents an individual-level estimate of the variance for the random noise \citep[Section 7.5.1]{Boos2013}.

\subsubsection{Inference for Population-Level Parameter Estimates} \label{sec:2.1.2}

Once $\hat{\alpha}^{(c)}$ and $\mathrm{Var}(\hat{\alpha}^{(c)})$ are obtained, an investigator can perform hypothesis tests to assess how a biological phenomenon oscillates for a cohort. Specifically, an investigator would assess null hypotheses of the form $H_0: g(\alpha^{(c)}) = 0$, where $g(\alpha^{(c)})$ is a function that maps the $(2K+1) \times 1$ vector $\alpha^{(c)}$ to a $q \times 1$ vector. A Wald-type test statistic for assessing this null hypothesis is defined as
\begin{align}
\tau = g(\hat{\alpha}^{(c)})^T\mathrm{Var}\left\{g(\hat{\alpha}^{(c)})\right\}^{-1}g(\hat{\alpha}^{(c)}),   \label{eq:wald_form}
\end{align}
where 
\begin{align}
\mathrm{Var}\{g(\hat{\alpha}^{(c)})\} &= \frac{1}{M^{(c)}}\left\{G(\hat{\alpha}^{(c)})^T\hat{D}^{(c)}G(\hat{\alpha}^{(c)}) + \frac{1}{M^{(c)}}\sum_{i=1}^{M^{(c)}}G(\hat{\gamma}^{(c)}_i)^T\hat{\Sigma}^{(c)}_iG(\hat{\gamma}^{(c)}_i)\right\}. \label{eq:delt_var}
\end{align}
To clarify, (\ref{eq:delt_var}) is an estimate of the variance for $g(\hat{\alpha}^{(c)})$ that is obtained with the Delta method, where $G(\alpha)$ is a $q \times (2K+1)$ matrix that represents the derivative of $g(\alpha)$ with respect to $\alpha$ \citep[Theorem 5.19]{Boos2013}. If the rank of $G(\alpha)$ is equal to $q$, then $\tau$ follows a central chi-squared distribution with $q$ degrees of freedom if the null hypothesis $g(\alpha^{(c)}) = 0$ is true. The $p$-value for this test is defined as $1-F_q(\tau)$, where $F_q(Z)$ is the cumulative distribution function of the central chi-squared distribution with $q$ degrees of freedom and argument $Z$. The null hypothesis is rejected if this $p$-value is less than the pre-determined significance level $\rho$ (\citealp[Section 3.2]{Boos2013}; \citealp[Section 6.2.2]{Davidian1995}). An investigator could also assess null hypotheses of the form $H_0: g(\alpha^{(1)}) - g(\alpha^{(0)}) = 0$ to test for differences in oscillations across cohorts. When the cohorts are independent, the corresponding Wald-type test statistic is instead defined as
\begin{align}
\tau &= \left\{g(\hat{\alpha}^{(1)}) - g(\hat{\alpha}^{(0)})\right\}^T \nonumber \\
& \quad \times \left[\mathrm{Var}\left\{g(\hat{\alpha}^{(1)})\right\} + \mathrm{Var}\left\{g(\hat{\alpha}^{(0)})\right\}\right]^{-1} \label{eq:test2} \\
& \quad \times \left\{g(\hat{\alpha}^{(1)}) - g(\hat{\alpha}^{(0)})\right\}. \nonumber
\end{align}
This test statistic would also follow a central chi-squared distribution with $q$ degrees of freedom when the null hypothesis is true.

The central chi-squared distributions used to compute $p$-values from the test statistics in (\ref{eq:wald_form}) and (\ref{eq:test2}) are the asymptotic null distributions that arise from large-sample approximations. In practice, however, the number of individuals in the $c$-th cohort ($M^{(c)}$) as well as the number of samples obtained from each $i$-th individual in the $c$-th cohort ($n^{(c)}_i$) can be small, which could make this large-sample approximation imprecise. To account for this potential lack of precision, we will consider the ``nonparametric random effects and individual residual bootstrap'' procedure for hypothesis testing, which instead estimates the distribution to compare a test statistic against and has empirically performed well for linear and nonlinear mixed-effects models \citep{Thaia2013, Thaib2013}. 

To provide an example of this bootstrap procedure when assessing the null hypothesis $H_0: g(\alpha^{(c)})=0$, we first note that $\theta_i^{(c)}$ from (\ref{eq:trig_mod}) can be modeled as
\begin{align}
\theta^{(c)}_{i} = \beta^{(c)} + b_i^{(c)}. \label{eq:mixed_ap}
\end{align}
Here, $\beta^{(c)}$ represents a non-random population-level vector and $b_i^{(c)}$ is a random individual-level vector with $\mathbb{E}(b_i^{(c)})=0$ and a constant covariance matrix $\mathrm{Var}(b_i^{(c)})$. In subsequent analyses, we will assess the null hypothesis $H_0: \beta_{2k-1}^{(c)}=0$ for all $k \in \{1,\ldots,K\}$, which is referred to as the ``zero amplitudes test'' and assesses whether a biological phenomenon oscillates \citep{Bingham1982}. The bootstrap procedure for the zero amplitudes test would first produce $R$ different bootstrap replicate test statistics, where the $r$-th bootstrap replicate test statistic $\tau^{(r)}$ would be obtained as follows:
\begin{enumerate}
 \item Draw a sample of vectors $\{\hat{\gamma}^{(c,r)}_{1},\ldots, \hat{\gamma}^{(c,r)}_{M^{(c)}}\}$ from the set $\{\hat{\gamma}^{(c)}_1,\ldots, \hat{\gamma}^{(c)}_{M^{(c)}}\}$ with replacement $M^{(c)}$ times. Here, the superscript $(c,r)$ denotes quantities obtained from the $c$-th cohort for the $r$-th bootstrap replicate.
 \item For each $i$-th individual, draw a sample of residual estimates $\{\hat{\epsilon}^{(c,r)}_{i,1},\ldots, \hat{\epsilon}^{(c,r)}_{i,n_i^{(c)}}\}$ from the set $\{\hat{\epsilon}^{(c)}_{i,1},\ldots, \hat{\epsilon}^{(c)}_{i,n_i^{(c)}}\}$ with replacement $n_i^{(c)}$ times. The residual estimate $\hat{\epsilon}^{(c)}_{i,j}$ is defined as $$\hat{\epsilon}^{(c)}_{i,j} = Y^{(c)}_{i,j} - h(X_{i,j}^{(c)}, \hat{\gamma}_i^{(c)}).$$
 \item Generate bootstrap replicate measurements $Y_{i,j}^{(c,r)} = h(X_{i,j}^{(c)}, \tilde{\gamma}_{i}^{(c,r)})+\hat{\epsilon}^{(c,r)}_{i,j}$. Here, each $\tilde{\gamma}_{i,2j-1}^{(c,r)}$ and $\tilde{\gamma}_{i,2j}^{(c,r)}$ for $j \in \{1,\ldots,K\}$ would be defined as 
 \begin{align*}
 \tilde{\gamma}_{i,2j-1}^{(c,r)} &= -\left\{\sqrt{(\hat{\gamma}_{i,2j-1}^{(c,r)})^2+(\hat{\gamma}_{i,2j}^{(c,r)})^2}-\sqrt{(\hat{\alpha}_{i,2j-1}^{(c)})^2+(\hat{\alpha}_{i,2j}^{(c)})^2}\right\}\sin\left\{\mathrm{atan2}(-\hat{\alpha}_{i,2j-1}^{(c)}, \hat{\alpha}_{i,2j}^{(c)})\right\}, \\
  \tilde{\gamma}_{i,2j}^{(c,r)} &= \left\{\sqrt{(\hat{\gamma}_{i,2j-1}^{(c,r)})^2+(\hat{\gamma}_{i,2j}^{(c,r)})^2}-\sqrt{(\hat{\alpha}_{i,2j-1}^{(c)})^2+(\hat{\alpha}_{i,2j}^{(c)})^2}\right\}\cos\left\{\mathrm{atan2}(-\hat{\alpha}_{i,2j-1}^{(c)}, \hat{\alpha}_{i,2j}^{(c)})\right\}.
 \end{align*}
 It is noted that the elements $\tilde{\gamma}_{i,2j-1}^{(c,r)}$ and $\tilde{\gamma}_{i,2j}^{(c,r)}$ are adjusted using the identities in (\ref{eq:alt_to_orig}) to generate data where the null hypothesis is true.
 \item Apply the STS method to fit a trigonometric regression model to the time points ($X_{i,j}^{(c)}$) and the generated bootstrap replicate measurements ($Y_{i,j}^{(c,r)}$).
 \item Given the parameter estimates obtained in Step 4, compute the corresponding Wald-type test statistic $\tau^{(r)}$ using (\ref{eq:wald_form}).
\end{enumerate}
After generating $\{\tau^{(1)}, \ldots, \tau^{(R)}\}$, a bootstrapped $p$-value would be computed as
\begin{align}
p = \frac{1}{R} \sum_{r=1}^{R} \mathbf{1}(\tau \leq \tau^{(r)}), \label{eq:emp_pval}
\end{align} 
where $\tau$ is the test statistic computed with the original dataset in (\ref{eq:wald_form}), and $\mathbf{1}(\tau \leq \tau^{(r)})$ denotes an indicator function that equals one when $\tau \leq \tau^{(r)}$ and zero otherwise. 

To compare oscillations between cohorts, or assess the null hypothesis $H_0: g(\alpha^{(1)})-g(\alpha^{(0)})=0$, the $r$-th bootstrap replicate test statistic $\tau^{(r)}$ would instead be obtained as follows:
\begin{enumerate}
 \item For each $c$-th cohort, draw a sample of vectors $\{\hat{\gamma}_{1}^{(c,r)},\ldots,\hat{\gamma}_{M^{(c)}}^{(c,r)}\}$ from the set \newline $\{\hat{\gamma}_{1}^{(0)},\ldots,\hat{\gamma}_{M^{(0)}}^{(0)}, \hat{\gamma}_{1}^{(1)},\ldots,\hat{\gamma}_{M^{(1)}}^{(1)}\}$ with replacement.
 \item For each $i$-th individual from the $c$-th cohort, draw a sample of residual estimates $\{\hat{\epsilon}^{(c,r)}_{i,1},\ldots, \hat{\epsilon}^{(c,r)}_{i,n_i^{(c)}}\}$ from the set $\{\hat{\epsilon}^{(c)}_{i,1},\ldots, \hat{\epsilon}^{(c)}_{i,n_i^{(c)}}\}$ with replacement $n_i^{(c)}$ times.
 \item Generate bootstrap replicate measurements $Y_{i,j}^{(c,r)} = h(X_{i,j}^{(c)}, \hat{\gamma}_{i}^{(c,r)})+\hat{\epsilon}^{(c,r)}_{i,j}$. 
\item Apply the STS method to fit a trigonometric regression model to the bootstrap sample for each cohort.
 \item Given the parameter estimates obtained in Step 4, compute the corresponding Wald-type test statistic $\tau^{(r)}$ using (\ref{eq:test2}).
\end{enumerate}
Calculation of a bootstrapped $p$-value for this hypothesis test is the same as (\ref{eq:emp_pval}). It is noted that the motivation for sampling individuals from both cohorts in Step 1 is to represent a scenario where the null hypothesis is true \citep[Section 11.6]{Boos2013}. 

\subsection{Limitations of the STS Method for Trigonometric Regression} \label{sec:2.2}

\subsubsection{Motivating Results for Parameter Estimation} \label{sec:2.2.1}

A challenge with modeling biological phenomena over time is that individual-level differences in oscillations can bias statistical analyses and study conclusions when these differences are not taken into account (\citealt{Crainiceanu2010}; \citealt{Marron2015}; \citealt[Chapters 4 and 8]{Srivastava2016}). In the context of trigonometric regression, prior work has also qualitatively \citep{Sollberger1962, Weaver1995} and quantitatively \citep{Gorczycaa2024, Gorczycab2024} described how mis-measurement of $X_{i,j}^{(c)}$ can attenuate the parameter estimates for the linear model from (\ref{eq:lin_mod}). This article is motivated by these results, and first assesses whether or not population-level parameter estimates produced by the STS method of Section \ref{sec:2.1} are biased when the individual-level phase-shift parameters are distinct across individuals. Specifically, the following proposition leverages the expression in (\ref{eq:mixed_ap}) to derive the expected population-level parameter estimates when each individual has distinct phase-shift parameters.
\begin{prop} \label{prop:1}
Suppose the following assumptions are valid:
\begin{enumerate}
    \item $n_i^{(c)}=n$ for all $i$, with $n > 2K$, and $X_{i,j}^{(c)} = 24(j-1)/n$ for all $i$ and $j$.
    \item Each $b^{(c)}_{i,j}$ is independent of $b^{(c)}_{i,k}$ for all $j\neq k$.
\end{enumerate}
If the STS method presented in Section \ref{sec:2.1} is used to obtain population-level parameter estimates of a $K$-th order model in (\ref{eq:lin_mod}), then the expectation of these parameter estimates can be expressed as
\begin{align*}
\mathbb{E}(\hat{\alpha}^{(c)}_{0}) &= \alpha^{(c)}_{0}, \\
\mathbb{E}(\hat{\alpha}^{(c)}_{2k-1}) &= \alpha^{(c)}_{2k-1}\mathbb{E}\left\{\cos(b^{(c)}_{i,2k})\right\} -\alpha^{(c)}_{2k}\mathbb{E}\left\{\sin(b^{(c)}_{i,2k})\right\}, \\
\mathbb{E}(\hat{\alpha}^{(c)}_{2k}) &= \alpha^{(c)}_{2k}\mathbb{E}\left\{\cos(b^{(c)}_{i,2k})\right\} +\alpha^{(c)}_{2k-1}\mathbb{E}\left\{\sin(b^{(c)}_{i,2k})\right\}.
\end{align*}
Here, $\alpha_{2k-1}^{(c)} = -\beta^{(c)}_{2k-1}\sin(\beta^{(c)}_{2k})$ and $\alpha_{2k}^{(c)} = \beta^{(c)}_{2k-1}\cos(\beta^{(c)}_{2k})$, where $\beta^{(c)}_{2k-1}$ denotes the true population-level amplitude parameter for the $k$-th harmonic of the amplitude-phase model in (\ref{eq:mixed_ap}), and $\beta^{(c)}_{2k}$ the corresponding true phase-shift parameter.
\end{prop}
A derivation for Proposition \ref{prop:1} is provided in Appendix \ref{App:A1}. To clarify the setup for this result, the first assumption concerns the sample collection protocol. Specifically, this sample collection protocol is an example of an equispaced experimental design, which is optimal for trigonometric regression under multiple statistical criteria (\citealp[Pages 94-97]{Federov1972}; \citealp[Pages 241-243]{Pukelsheim2006}). Equispaced experimental designs have also been recommended for the design of circadian biology experiments \citep{Hughes2017, Zong2023}. This experimental design is setup such that the number of distinct measurement times for each individual is greater than $2K$, which ensures consistent parameter estimation \citep[Pages 21-23]{Bloomfield2000}. The second assumption that $b^{(c)}_{i,j}$ is independent of $b^{(c)}_{i,k}$ for all $j\neq k$ is a simplifying assumption that is made to obtain closed-form expressions for the expected population-level parameter estimates. 

Proposition \ref{prop:1} shows that when the individual-level phase-shift parameters are distinct across individuals, the STS method produced biased estimates of the population-level parameters $\alpha^{(c)}$. Specifically, while the intercept term $\hat{\alpha}^{(c)}_0$ is an unbiased estimate of $\alpha^{(c)}_0$, the estimates $\hat{\alpha}^{(c)}_{2k-1}$ and $\hat{\alpha}^{(c)}_{2k}$ for each $k \in \{1,\ldots,K\}$ are biased. This bias depends on the expectation of the sine and cosine transforms of the individual-level phase-shift parameters. For example, when the probability distribution generating $b^{(c)}_{i,2k}$ becomes more uniform over the interval $[-\pi, \pi)$, the quantities $\mathbb{E}\{\cos(b^{(c)}_{i,2k})\}$ and $\mathbb{E}\{\sin(b^{(c)}_{i,2k})\}$ attenuate further toward zero in magnitude. This bias would also affect the numeric value of the amplitude and phase-shift parameters obtained using the identities in (\ref{eq:alt_to_orig}), which we summarize in the following corollary.
\begin{cor} \label{cor:1}
If the assumptions of Proposition \ref{prop:1} are valid, then the $k$-th population-level amplitude parameter obtained from the identities in (\ref{eq:alt_to_orig}) would be expressed as
\begin{align*}
    \sqrt{\mathbb{E}(\hat{\alpha}^{(c)}_{2k-1})^2+\mathbb{E}(\hat{\alpha}^{(c)}_{2k})^2} &= \beta^{(c)}_{2k-1}|\phi_{b^{(c)}_{i,2k}}(1)|,
\end{align*}
and the $k$-th population-level phase-shift parameter would be expressed as
\begin{align*}
    \mathrm{atan2}\left\{-\mathbb{E}(\hat{\alpha}^{(c)}_{2k-1}), \mathbb{E}(\hat{\alpha}^{(c)}_{2k})\right\} &= \mathrm{atan2}\bigg[\sin(\beta^{(c)}_{2k})\mathbb{E}\left\{\cos( b^{(c)}_{i,2k})\right\} + \cos(\beta^{(c)}_{2k})\mathbb{E}\left\{\sin(b^{(c)}_{i,2k})\right\}, \\
    & \quad \quad \quad \quad \ \cos(\beta^{(c)}_{2k}) \mathbb{E}\left\{\cos(b^{(c)}_{i,2k})\right\} - \sin(\beta^{(c)}_{2k}) \mathbb{E}\left\{\sin(b^{(c)}_{i,2k})\right\} \bigg].
\end{align*}
Here, $|\phi_{b^{(c)}_{i,2k}}(t)|$ denotes the magnitude of the characteristic function for $b^{(c)}_{i, 2k}$ evaluated at $t$.
\end{cor}
A derivation for this result is provided in Appendix \ref{App:A2}. Corollary \ref{cor:1} shows that the population-level amplitude parameter estimates are attenuated, as the magnitude of a characteristic function $|\phi_{b^{(c)}_{i,2k}}(t)|$ is bounded by one for any argument $t$. Additionally, Corollary \ref{cor:1} indicates that the population-level phase-shift estimates could be unbiased depending on the distribution of individual-level phase-shift parameters. For example, if this distribution is symmetric around zero, then $\mathbb{E}\{\sin(tb^{(c)}_{i,2k})\}=0$ given any real-valued constant $t$. As a result,
\begin{align*}
    \mathrm{atan2}\left\{-\mathbb{E}(\hat{\alpha}^{(c)}_{2k-1}), \mathbb{E}(\hat{\alpha}^{(c)}_{2k})\right\} &= \mathrm{atan2}\bigg[\sin(\beta^{(c)}_{2k})\mathbb{E}\left\{\cos( b^{(c)}_{i,2k})\right\}, \cos(\beta^{(c)}_{2k})\mathbb{E}\left\{\cos( b^{(c)}_{i,2k})\right\} \bigg] \\
    &= \beta^{(c)}_{2k}.
\end{align*}
It is noted that this result for unbiased phase-shift estimation has been identified in other research efforts \citep{Gorczycad2024, Gorczycab2024}.

\subsubsection{Numeric Example for Parameter Estimation and Inference} \label{sec:2.2.2}

Proposition \ref{prop:1} and Corollary \ref{cor:1} imply that inferences with population-level parameter estimates obtained from the STS method would also be biased. To illustrate this bias for inference, suppose an investigator performs a study on how a treatment affects biomarker levels. The investigator recruits an even number of individuals ($M^{(1)}$) whose biomarker levels follow a first-order amplitude-phase trigonometric regression model in (\ref{eq:amp}) post-treatment (there is only one case cohort, $c=1$). For the experimental protocol, the investigator measures biomarker levels from each individual once every four hours over a 24-hour period, which results in $n_i^{(1)}=n=6$ longitudinal biomarker level measurements from each individual.

Once the investigator obtains data, the STS method of Section \ref{sec:2.1} is used to obtain population-level parameter estimates. To isolate the effect of bias from distinct individual-level phase-shifts, we assume that the individual-level parameter estimates equal their true values. Specifically, each individual-level intercept parameter and each individual-level amplitude parameter equal their corresponding population-level parameters, with $\theta^{(1)}_{i,0} = \beta^{(1)}_0 = 6$ and $\theta^{(1)}_{i,1} = \beta^{(1)}_1 = 1/2$. However, half of the participating individuals have an individual-level phase-shift parameter of $\theta^{(1)}_{i,2} = -\pi/4$, while the other half have $\theta^{(1)}_{i,2} = \pi/4$. The variance estimate of the individual-level random noise is $(\hat{\sigma}^{(1)}_i)^2=(\sigma^{(1)}_i)^2=1$ for all $M^{(1)}$ individuals.

When the STS method is used to estimate the population-level parameter vector $\hat{\alpha}^{(1)}$, application of (\ref{eq:sts_pop}), or averaging over the individual-level parameter estimates, would yield
\begin{align*}
    \hat{\alpha}^{(1)} &= \frac{1}{M^{(1)}}\sum_{i=1}^{M^{(1)}/2}\left[6 \ \ -\frac{\sin(-\pi/4)}{2} \ \ \frac{\cos(-\pi/4)}{2}\right] + \frac{1}{M^{(1)}}\sum_{i=1}^{M^{(1)}/2}\left[6 \ \ -\frac{\sin(\pi/4)}{2} \ \ \frac{\cos(\pi/4)}{2}\right] \\
    &=\left[6 \ \ 0 \ \ \frac{1}{2\sqrt{2}}\right].
\end{align*}
When the identities from (\ref{eq:alt_to_orig}) are then used to obtain population-level amplitude and phase-shift parameter estimates, the investigator would obtain 
\begin{align*}
    \hat{\beta}^{(1)}_1 &= \sqrt{(0^2) + \left(\frac{1}{2\sqrt{2}}\right)^2} = \frac{1}{2\sqrt{2}}, \quad \quad \hat{\beta}^{(1)}_2 = \mathrm{atan2}\left(0, \frac{1}{2\sqrt{2}} \right) = 0.
\end{align*}
To clarify, the amplitude parameter estimate $\hat{\beta}^{(1)}_1$ is biased, while the phase-shift parameter estimate $\hat{\beta}^{(1)}_2$ is unbiased because the distribution for $\theta^{(1)}_{i,2}$ is symmetric around zero. This attenuation bias aligns with Corollary \ref{cor:1}, as 
\begin{align*}
    \beta_1^{(1)}|\phi_{b^{(1)}_{i,2}}(1)| &= \frac{\left|\mathbb{E}\left\{\exp(zb^{(1)}_{i,2})\right\}\right|}{2} = \frac{\left|\mathbb{E}\left\{\cos(b^{(1)}_{i,2})\right\}\right|}{2}=\frac{|\cos(-\pi/4)+\cos(\pi/4)|}{4}=\frac{1}{2\sqrt{2}}, 
\end{align*}
where $z = \sqrt{-1}$. Figure \ref{fig:example} provides a graphical comparison of population-level biomarker level estimates obtained from the STS method over time to the corresponding true individual-level and population-level biomarker levels.

After the investigator obtains population-level parameter estimates, the null hypothesis $H_0:\beta^{(1)}_1=0$, or the zero amplitude test, is assessed to determine whether or not biomarker levels oscillate post-treatment \citep{Bingham1982, Tong1976}. In Appendix \ref{App:B2}, we compute that the numeric value of the corresponding Wald test statistic would be
\begin{align*}
    \tau &=\frac{3M^{(1)}\phi^2_{b^{(1)}_{i,2}}(1)}{4} = \frac{3M^{(1)}}{8}.
\end{align*}
If the investigator assumes that the test statistic follows a central chi-squared distribution with one degree of freedom under the null hypothesis, then the test statistic would need to be at least as large as $3.841$ to reject the null hypothesis at a significance threshold of $\rho=0.05$ \citep[Section 1.3.6.7.4]{Nist2002}. As a consequence, the investigator would need to recruit at least $M^{(1)}=12$ individuals to reject the null hypothesis and avoid committing a type II error \citep[Section 3.1]{Lehmann2022}. It is noted that if there were no individual-level differences in phase-shift parameters, then the Wald test statistic would instead equal
\begin{align*}
    \tau &= \frac{3M^{(1)}}{4}.
\end{align*}
In this scenario, the investigator would only need to recruit $M^{(1)}=6$ individuals to reject the null hypothesis $H_0: \beta^{(1)}_1=0$, assuming each individual-level parameter estimate equals their corresponding estimand.

\subsection{A Refined Two-Stage Method for Trigonometric Regression} \label{sec:2.3}

Proposition \ref{prop:1} demonstrates that the population-level parameter estimates produced by the standard two-stage (STS) method are attenuated when the individual-level phase-shift parameter estimates are distinct across individuals. The numeric example of Section \ref{sec:2.2.2} shows that these attenuated parameter estimates would attenuate the Wald test statistic computed when performing the zero amplitudes test. To address this attenuation bias in population-level parameter estimates and inaccurate hypothesis testing results, this section introduces a refined two-stage (RTS) method for trigonometric regression. 

The first refinement transforms each individual-level parameter estimate $\hat{\gamma}^{(c)}_i$ from the linear model in (\ref{eq:lin_mod}) into separate amplitude and phase-shift components before estimating the population-level parameters. Specifically, the $(2K+1)\times 1$ vector $\hat{\gamma}^{(c)}_i$ would be transformed into the $(3K+1)\times 1$ vector $\tilde{\theta}^{(c)}_i$, where each element of $\tilde{\theta}^{(c)}_i$ is defined as
\begin{align}
    \tilde{\theta}^{(c)}_{i,0} &= \hat{\theta}_{i,0}^{(c)}, \nonumber \\
    \tilde{\theta}^{(c)}_{i,3k-2} &= \sqrt{(\hat{\gamma}_{i,2k-1}^{(c)})^2+(\hat{\gamma}_{i,2k}^{(c)})^2}, \label{eq:amp} \\
    \tilde{\theta}^{(c)}_{i,3k-1} &= \sin\left\{\mathrm{atan2}(-\hat{\gamma}_{i,2k-1}^{(c)}, \hat{\gamma}_{i,2k}^{(c)})\right\}, \label{eq:circ_avg1} \\  
    \tilde{\theta}^{(c)}_{i,3k} &= \cos\left\{\mathrm{atan2}(-\hat{\gamma}_{i,2k-1}^{(c)}, \hat{\gamma}_{i,2k}^{(c)})\right\}, \label{eq:circ_avg2}
\end{align}
with $k \in \{1,\ldots, K\}$. The amplitude quantity in (\ref{eq:amp}) is due to direct application of the amplitude identity from (\ref{eq:alt_to_orig}). The quantities in (\ref{eq:circ_avg1}) and (\ref{eq:circ_avg2}) similarly leverages the phase-shift identity from (\ref{eq:alt_to_orig}), as
\begin{align*}
    \hat{\theta}^{(c)}_{i,2k} = \mathrm{atan2}(-\hat{\gamma}_{i,2k-1}^{(c)}, \hat{\gamma}_{i,2k}^{(c)}),
\end{align*}
where $\hat{\theta}^{(c)}_{i,2k}$ is the $k$-th phase-shift estimate for the $i$-th individual in the $c$-th cohort. However, we define $\tilde{\theta}^{(c)}_{i,3k-1} = \sin(\hat{\theta}_{i,2k})$ and $\tilde{\theta}^{(c)}_{i,3k} = \cos(\hat{\theta}_{i,2k})$ to avoid direct averaging of the individual-level phase-shift parameters when computing population-level estimates. To clarify, phase-shift parameters are defined as circular quantities, or angles on the interval $[-\pi, \pi)$ \citep{Mardia1976}. If a population-level phase-shift parameter is defined on the boundary of this interval, then direct averaging of the corresponding individual-level phase-shift parameters could bias the population-level estimate depending on how the individual-level phase-shifts are dispersed. The transform of $\hat{\theta}^{(c)}_{i,2k}$ to $[\sin(\hat{\theta}^{(c)}_{i,2k}) \ \ \cos(\hat{\theta}^{(c)}_{i,2k})]$ is a mapping of the phase-shift parameter onto the unit circle, which ensures that $\mathrm{atan2}(\tilde{\theta}_{i,3k-1}^{(c)}, \tilde{\theta}_{i,3k}^{(c)})$ is an unbiased population-level phase-shift estimate \citep[Section 2.2]{Mardia1999}. It is noted that the quantities in (\ref{eq:circ_avg1}) and (\ref{eq:circ_avg2}) follow a normal distribution asymptotically \citep[Section 4.8]{Mardia1999}. The corresponding population-level parameter vector would then be computed as
\begin{align*}
    \tilde{\beta}^{(c)} &= \frac{1}{M^{(c)}}\sum_{i=1}^{M^{(c)}}\tilde{\theta}^{(c)}_{i}.
\end{align*}

The second refinement involves adjusting calculation of the empirical covariance matrix $\mathrm{Var}\{g_1(\tilde{\beta}^{(c)})\}$ following (\ref{eq:delt_var}). This adjustment involves first identifying a function $g_2(\hat{\gamma}_i^{(c)})$ where $g_1(\tilde{\theta}_i^{(c)})=g_2(\hat{\gamma}_i^{(c)})$. The empirical covariance matrix $\mathrm{Var}\{g_1(\tilde{\beta}^{(c)})\}$ can then be computed as
\begin{align*}
\mathrm{Var}\{g_1(\tilde{\beta}^{(c)})\} = \frac{1}{M^{(c)}}\left\{G_1(\tilde{\beta}^{(c)})\tilde{D}^{(c)}G_1(\tilde{\beta}^{(c)}) + \frac{1}{M^{(c)}}\sum_{i=1}^{M^{(c)}}G_2(\hat{\gamma}_i^{(c)})\hat{\Sigma}^{(c)}_iG_2(\hat{\gamma}_i^{(c)})\right\}.
\end{align*}
Here,
\begin{align*}
    \tilde{D}^{(c)} &= \frac{1}{M^{(c)}-1}\sum_{i=1}^{M^{(c)}} (\tilde{\theta}^{(c)}_{i}-\tilde{\beta}^{(c)})(\tilde{\theta}^{(c)}_{i}-\tilde{\beta}^{(c)})^T,
\end{align*}
with $G_1(\kappa)$ representing the derivative of $g_1(\kappa)$ with respect to the $(3K+1)\times 1$ vector $\kappa$, and $G_2(\xi)$ the derivative of $g_2(\xi)$ with respect to the $(2K+1)\times 1$ vector $\xi$.

\section{Simulation Study} \label{sec:3}

\subsection{Parameter Estimation and Inference with a Single Cohort} \label{sec:3.1}

\subsubsection{Simulation Setup for a Single Cohort} \label{sec:3.1.1}

A simulation study is conducted to compare the standard two-stage (STS) method for trigonometric regression from Section \ref{sec:2.1} to our refined two-stage (RTS) method from Section \ref{sec:2.3}. This study considers different simulation settings obtained from varying the following four design factors for a control cohort ($c=0$):
\begin{description}
    \item[Number of Harmonics.] $K=1$ (with $\beta^{(0)}_2 = \pi/4$); or $K=3$ (with $\beta^{(0)}_2 = \pi/8$, $\beta^{(0)}_4 = \pi/4$, and $\beta^{(0)}_6 = 3\pi/8$).
    \item[Phase-Shift Variability.] $b_{2k}^{(0)}\sim \mathrm{VM}(0, 2)$ for all $k \in \{1,\ldots, K\}$ (high phase-shift variability); and $b_{2k}^{(0)}\sim \mathrm{VM}(0, 8)$ for all $k \in \{1,\ldots, K\}$ (low phase-shift variability).
    \item[Sample Size and Cohort Size.] $n_{i}^{(0)} = 12$ for all $i$ and $M^{(0)} = 10$ (small sample size and cohort size); and $n_{i}^{(0)} = 192$ for all $i$ and $M^{(0)} = 20$ (large sample size and cohort size).
    \item[Signal-to-Noise Ratio.] $\beta^{(0)}_{2k-1}=1.5$ with $b^{(0)}_{i,2k-1}\sim \mathrm{TN}(-0.75, 0.75, 0, 1/\sqrt{2})$ for all $k \in \{1,\ldots, K\}$ (a high signal-to-noise ratio); and $\beta^{(0)}_{2k-1}=0.5$ with $b^{(0)}_{i,2k-1}\sim \mathrm{TN}(-0.25, 0.25, 0, 1/\sqrt{2})$ for all $k \in \{1,\ldots, K\}$ (a low signal-to-noise ratio). 
 \end{description}
Here, $\mathrm{TN}(Z_1, Z_2, Z_3, Z_4)$ represents a truncated normal distribution with mean $Z_1$, variance $Z_2$, lower bound $Z_3$, and upper bound $Z_4$; and $\mathrm{VM}(\kappa, \xi)$ denotes a von Mises distribution with mean $\kappa$ and concentration $\xi$. The von Mises distribution is the circular analog of the normal distribution \citep{Lee2010}, and the phase-shift parameter is considered a circular quantity \citep{Mardia1976}. 

In total, 16 different simulation settings are evaluated by varying these design factors. For instance, one setting would involve configuring the number of harmonics to $K=3$; the phase-shift variability to generate each $b^{(0)}_{i,2k} \sim \mathrm{VM}(0, 2)$ for all $k \in \{1, 2, 3\}$; the sample size for each $i$-th individual and cohort size to $n_{i}^{(0)} = 12$ and $M^{(0)} = 10$, respectively; and the signal-to-noise ratio to specify each amplitude parameter $\beta^{(0)}_{2k-1} = 1.5$ for all $k\in\{1, 2, 3\}$. In all 16 simulation settings, the following additional quantities are generated in the same manner:
\begin{enumerate}
    \item The random noise $\epsilon^{(0)}_{i,j}\sim \mathrm{N}(0, 1)$, or is generated from a standard normal distribution.
    \item Each population-level intercept parameter $\beta^{(0)}_0=6$.
    \item Each individual-level intercept parameter $b^{(0)}_{i, 0}\sim \mathrm{N}(0, 1)$.
    \item Each individual-level amplitude parameter $b^{(0)}_{i,2k-1}\sim \mathrm{TN}(0, 1/2, -\beta^{(0)}_{2k-1}, \beta^{(0)}_{2k-1})$ for all $k$.
    \item Covariate data are obtained from an equispaced experimental design, where $X^{(0)}_{i,j} = 24(j-1)/n_i^{(0)}$ for all $i$.
\end{enumerate}

These additional quantities enable interpretation of each design factor. Specifically, the number of harmonics design factor reflects the range of order parameters that are typically specified for trigonometric regression in circadian biology studies \citep{Albert2005, Hughes2009}. The phase-shift variability design factor represents how individual-level phase-shift parameters are dispersed. In particular, the concentration parameter $\kappa$ can be interpreted as a reciprocal measure of dispersion, with the approximation $1/\kappa$ used for the variance parameter of a normal distribution when $\kappa$ is large \citep[Equation 3.5.22]{Mardia1999}. This approximation implies that when $\kappa = 8$, the standard deviation of each individual-level phase-shift parameter for the $k$-th harmonic is approximately $12/(\pi\sqrt{8}) \approx 1.350$ hours, which could represent a scenario where the investigator applied strict inclusion criteria when selecting the study population. In contrast, when $\kappa=2$, the standard deviation increases to approximately $12/(\pi\sqrt{2}) \approx 2.701$ hours, which could represent a scenario where the investigator exercised less stringent control when selecting the study population \citep{Kennaway2023}. The sample size and cohort size design factor reflects the type of biological phenomenon under study. For example, a small number of samples taken from each individual ($n_i^{(0)}=12$) and a small cohort ($M^{(0)}=10$) could reflect a scenario where an investigator is measuring biomarker levels that are expensive to obtain. Finally, the signal-to-noise ratio design factor represents the ratio of each population-level amplitude for the $k$-th harmonic relative to the variance of the individual-level random noise, or $\beta^{(0)}_{2k-1}/\{(\sigma^{(0)})^2\}$.

For each simulation setting, 1,000 simulation trials are performed. In each simulation trial, four datasets are generated: 
\begin{description}
    \item[Dataset 1.] Generated following the design factors specified for a simulation setting.
    \item[Dataset 2.] Generated following the design factors specified for a simulation setting, except $b^{(0)}_{i,2k}=0$ for all $i$ and $k \in \{1,\ldots, K\}$.
    \item[Dataset 3.] Generated following the design factors specified for a simulation setting, except $\beta^{(0)}_{2k-1}=0$ and $b^{(0)}_{i,2k-1}=0$ for all $i$ and $k \in \{1,\ldots, K\}$.
    \item[Dataset 4.] Generated following the design factors specified for a simulation setting, except $\beta^{(0)}_{2k-1}=0$, $b^{(0)}_{i,2k-1}=0$, and $b^{(0)}_{i,2k}=0$ for all $i$ and $k \in \{1,\ldots, K\}$.
\end{description}
Dataset 1 and Dataset 3 represent scenarios where individual-level phase-shift parameters are distinct across individuals, whereas Dataset 2 and Dataset 4 represent scenarios where individual-level phase-shift parameters equal their corresponding population-level parameters. This setup enables a comparison of parameter estimates and hypothesis test results obtained from each method when individual variability is present relative to when it is absent. Additionally, Dataset 1 and Dataset 2 represent scenarios where the biological phenomenon oscillates, whereas Dataset 3 and Dataset 4 represent scenarios where the biological phenomenon does not oscillate. This setup enables a comparison of each method when performing the zero amplitudes test, or assessing the null hypothesis $H_0: \beta^{(0)}_{2k-1} = 0$ for all $k \in \{1,\ldots, K\}$.

Once each dataset is generated, we estimate population-level parameters and perform the zero amplitudes test using both the STS method and our RTS method. We then record the following quantities:
\begin{description}
    \item[Quantity 1.] $\hat{\beta}^{(0)}_{0} - \beta^{(0)}_0$, or the difference between the estimated and true population-level intercept parameters, on Dataset 1 and Dataset 2.
    \item[Quantity 2.] $\hat{\beta}^{(0)}_{2k-1} - \beta^{(0)}_{2k-1}$ for each $k \in \{1,\ldots, K\}$, or the difference between the estimated and true population-level amplitude parameters for each $k$-th harmonic, on Dataset 1 and Dataset 2.
    \item[Quantity 3.] $\mathrm{atan2}\{\mathrm{sin}(\hat{\beta}^{(0)}_{2k} - \beta^{(0)}_{2k}), \mathrm{cos}(\hat{\beta}^{(0)}_{2k} - \beta^{(0)}_{2k})\}$ for each $k \in \{1,\ldots, K\}$, or the circular difference between the estimated and true population-level phase-shift parameters for each $k$-th harmonic, on Dataset 1 and Dataset 2.
    \item[Quantity 4.] The bootstrapped $p$-value computed when performing the zero amplitudes test on Dataset 1 and Dataset 2.
    \item[Quantity 5.] The bootstrapped $p$-value computed when performing the zero amplitudes test on Dataset 3 and Dataset 4.
\end{description} 
It is noted that Quantity 3 computes the circular difference between the estimated and true phase-shift parameters, as the phase-shift parameters are defined as circular quantities \citep{Mardia1976}. Further, to compute a bootstrapped $p$-value for hypothesis testing, we perform $R=1,000$ bootstrap replicates. 

Once all 1,000 simulation trials are performed, we report the mean and standard deviations for Quantities 1-3. For Quantity 4, rather than reporting summary statistics of the $p$-values directly, we evaluate the overall effectiveness of the hypothesis test by computing the area under the empirical statistical power curve ($\mathrm{AUC}_{\mathrm{SP}}$). Specifically, let $p_i$ denote the $p$-value obtained in the $i$-th simulation trial for a simulation setting. Given a significance threshold $\rho$, the empirical statistical power is defined as
\begin{align*}
\mathrm{Statistical \ Power}(\rho) = \frac{1}{1000} \sum_{i=1}^{1000} \mathbf{1}(p_i \leq \rho)
\end{align*}
when the null hypothesis is false \citep[Section 3.1]{Lehmann2022}, where $\mathbf{1}(p_i \leq \rho )$ denotes an indicator function that equals one when $p_i \leq \rho$ and zero otherwise. We then define the $\mathrm{AUC}_{\mathrm{SP}}$ as the quantity
\begin{align}
\mathrm{AUC}_{\mathrm{SP}} = \int_0^1 \left\{\mathrm{Statistical \ Power}(\rho)\right\} d\rho. \label{eq:aucsp}
\end{align}
The quantity in (\ref{eq:aucsp}) serves as a threshold-free evaluation of a method's statistical power for a simulation setting. For the $\mathrm{AUC}_{\mathrm{SP}}$, larger values imply greater statistical power when considering all possible significance thresholds.

For Quantity 5, we similarly compute an area under the empirical type I error curve ($\mathrm{AUC}_{\mathrm{T1E}}$). Specifically, the empirical type I error rate is defined as
\begin{align*}
\mathrm{Type\ I\ Error}(\rho) = \frac{1}{1000} \sum_{i=1}^{1000} \mathbf{1}(p_i \leq \rho)
\end{align*}
when the null hypothesis is true, which implies that
\begin{align}
\mathrm{AUC}_{\mathrm{T1E}} = \int_0^1 \left\{\mathrm{Type\ I\ Error}(\rho)\right\} d\rho. \label{eq:auct1e}
\end{align}
Given that Quantity 5 is computed from datasets where the null hypothesis is true, the type I error rate would equal the significance threshold specified if the hypothesis test is well-calibrated, which would result in $\mathrm{AUC}_{\mathrm{T1E}} = 0.5$. However, it is possible for the hypothesis test to be mis-calibrated, with values less than 0.5 indicating that the method is conservative (the method rejects the null hypothesis less often than expected), while values greater than 0.5 indicate the method is anti-conservative (the method rejects the null hypothesis more often than expected). It is noted that for the $\mathrm{AUC}_{\mathrm{SP}}$ and the $\mathrm{AUC}_{\mathrm{T1E}}$ we will compute bootstrap confidence intervals with $R=100,000$ bootstrap replicates. 

\subsubsection{Simulation Study Results for a Single Cohort} \label{sec:3.1.2}

Table \ref{tab:sim1_K1} and Table \ref{tab:sim1_K3} summarize results for all simulation settings when the number of harmonics is set to $K=1$ and set to $K=3$, respectively. Our refined two-stage (RTS) method consistently mitigates attenuation bias in amplitude estimates and achieves greater statistical power for the zero amplitudes test when compared to the standard two-stage (STS) method, which aligns with the theory results in Section \ref{sec:2.2.1} and the numeric example in Section \ref{sec:2.2.2}. The only exception arises in simulation settings with low phase-shift variability, small sample sizes, and low signal-to-noise ratios. In these settings, our RTS method produces biased population-level parameter estimates and displays lower statistical power for the zero amplitudes test when compared to the STS method. These tables also demonstrate that our RTS method produces the same parameter estimates, as well as statistical power and type I error control for the zero amplitudes test, regardless of whether or not there are individual-level differences in phase-shift parameters. In every simulation setting, these corresponding quantities change for the STS method depending on whether or not there are individual-level differences in phase-shift parameters.

Figure \ref{fig:power_31} and Figure \ref{fig:typeI_31} present the empirical power and type I error curves for the zero amplitudes test, which are used to compute 
$\mathrm{AUC}_{\mathrm{SP}}$ and $\mathrm{AUC}_{\mathrm{T1E}}$, respectively. 95\% confidence intervals for these curves are obtained using the Dvoretzky–Kiefer–Wolfowitz inequality \citep{Dvoretzky1956}. Figure \ref{fig:power_31} indicates that our RTS method outperforms the STS method at every significance threshold, except in simulation settings with low phase-shift variability, small sample sizes, and low signal-to-noise ratios. Further, Figure \ref{fig:typeI_31} shows that both our RTS method and the STS method display conservative type I error control for the zero amplitudes test. This finding is consistent with the observations of \citet{Robins2000}, who noted that bootstrapped $p$-values are conservative when assessing the null hypothesis that a subset of the model parameters are equal to zero.

\subsection{Inference with Multiple Cohorts} \label{sec:3.2}

\subsubsection{Simulation Study Setup for Multiple Cohorts} \label{sec:3.2.1}

A second simulation study is conducted to compare the STS method to our RTS method for assessing differences in oscillations between cohorts. This second study considers different simulation settings obtained from varying design factors for a case cohort ($c=1$) and a control cohort ($c=0$). Specifically, the control cohort's data were generated by varying four design factors, defined and parameterized exactly as presented in Section \ref{sec:3.1.1}. For each design factor specified for the control cohort, the same design factor is specified for the case cohort, where the corresponding design factor for the case cohort is defined as follows:
\begin{description}
    \item[Number of Harmonics.] $K=1$ (with $\beta^{(1)}_2 = \pi/2$); or $K=3$ (with $\beta^{(1)}_2 = \pi/4$, $\beta^{(1)}_4 = \pi/2$, and $\beta^{(1)}_6 = 3\pi/4$).
    \item[Phase-Shift Variability.] $b_{2k}^{(1)}\sim \mathrm{VM}(0, 4)$ for all $k \in \{1,\ldots, K\}$ (high phase-shift variability); and $b_{2k}^{(1)}\sim \mathrm{VM}(0, 16)$ for all $k \in \{1,\ldots, K\}$ (low phase-shift variability). 
    \item[Sample Size and Cohort Size.] $n_{i}^{(1)} = 12$ for all $i$ and $M^{(1)} = 10$ (small sample size and cohort size); and $n_{i}^{(1)} = 192$ for all $i$ and $M^{(1)} = 20$ (large sample size and cohort size).
    \item[Signal-to-Noise Ratio.] $\beta^{(1)}_{2k-1}=1$ with $b^{(1)}_{i,2k-1}\sim \mathrm{TN}(-0.5, 0.5, 0, 1/\sqrt{2})$ for all $k \in \{1,\ldots, K\}$ (a high signal-to-noise ratio); and $\beta^{(1)}_{2k-1}=0.25$ with $b^{(1)}_{i,2k-1}\sim \mathrm{TN}(-0.125, 0.375, 0, 1/\sqrt{2})$ for all $k \in \{1,\ldots, K\}$ (a low signal-to-noise ratio). 
 \end{description}
 In all 16 simulation settings, the following additional quantities are generated in the same manner:
\begin{enumerate}
    \item The random noise $\epsilon^{(c)}_{i,j}\sim \mathrm{N}(0, 1)$ for each $c \in \{0,1\}$.
    \item Each population-level intercept parameter for the control cohort $\beta^{(0)}_0=6$.
    \item Each individual-level intercept parameter $b^{(c)}_{i, 0}\sim \mathrm{N}(0, 1)$ for each $c \in \{0, 1\}$.
    \item Each individual-level amplitude parameter $b^{(c)}_{i,2k-1}\sim \mathrm{TN}(0, 1/2, -\beta^{(c)}_{2k-1}, \beta^{(c)}_{2k-1})$ for all $k$ and each $c \in \{0,1\}$.
    \item Covariate data are obtained from an equispaced experimental design, where $X^{(c)}_{i,j} = 24(j-1)/n_i^{(c)}$ for all $i$ and $c \in \{0,1\}$.
\end{enumerate}
It is emphasized that for this second simulation study, we will set $\beta_{0}^{(1)} = 5$ when the signal-to-noise ratio design factor is set to ``high''; and we set $\beta_{0}^{(1)} = 4$ when this design factor is set to ``low''. 

For each simulation setting, 1,000 simulation trials are performed. In each simulation trial, four datasets are generated:
\begin{description}
    \item[Dataset 1.] Generated following the design factors specified for a simulation setting.
    \item[Dataset 2.] Generated following the design factors specified for a simulation setting, except $b^{(c)}_{i,2k}=0$ for all $i$, $c$, and $k \in \{1,\ldots, K\}$.
    \item[Dataset 3.] Generated following the design factors specified for a simulation setting, except $\beta^{(0)}_{k}$ is set to the same value as $\beta^{(1)}_{k}$ for all $k \in \{0,\ldots, 2K\}$.
    \item[Dataset 4.] Generated following the design factors specified for a simulation setting, except $\beta^{(0)}_{k}$ is set to the same value as $\beta^{(1)}_{k}$ for all $k \in \{0,\ldots, 2K\}$ and $b^{(c)}_{i,2k}=0$ for all $i$, $c$, and $k \in \{1,\ldots, K\}$.
\end{description}
Dataset 1 and Dataset 3 again represent scenarios where individual-level phase-shift parameters are distinct across individuals, whereas Dataset 2 and Dataset 4 represent scenarios where individual-level phase-shift parameters equal their corresponding population-level parameters. However, Dataset 1 and Dataset 2 now represent scenarios where each cohort displays different population-level oscillations, whereas Dataset 3 and Dataset 4 represent scenarios where each cohort displays the same population-level oscillations. 

For each dataset generated, we estimate population-level parameters using the STS method and our RTS method. We then assess two null hypotheses: $H_0: \beta_0^{(1)}-\beta_0^{(0)}=0$, which is known as the ``equal midlines test'' \citep{Bingham1982}; and $H_0: \beta_{2k-1}^{(1)} - \beta_{2k-1}^{(0)} = \mathrm{atan2}\{\sin(\beta_{2k}^{(1)} - \beta_{2k}^{(0)}), \cos(\beta_{2k}^{(1)} - \beta_{2k}^{(0)})\} = 0$ for all $k \in \{1, \ldots, K\}$, which we will define as the ``equal rhythms test.'' To clarify, the equal rhythms test assesses whether or not every amplitude and phase-shift parameter estimated for the control cohort is equal to the corresponding quantity estimated for the case cohort. It is emphasized that we assess $\mathrm{atan2}\{\sin(\beta_{2k}^{(1)} - \beta_{2k}^{(0)}), \cos(\beta_{2k}^{(1)} - \beta_{2k}^{(0)})\}=0$ as part of the equal rhythms test to transform the phase-shift parameters, which are circular quantities, to linear quantities. After we assess these null hypotheses, we record the following quantities:
\begin{description}
    \item[Quantity 1.] The bootstrapped $p$-value computed for the equal midlines test on Dataset 1 and Dataset 2.
    \item[Quantity 2.] The bootstrapped $p$-value computed for the equal midlines test on Dataset 3 and Dataset 4.
    \item[Quantity 3.] The bootstrapped $p$-value computed for the equal rhythms test on Dataset 1 and Dataset 2.
    \item[Quantity 4.] The bootstrapped $p$-value computed for the equal rhythms test on Dataset 3 and Dataset 4.
\end{description} 
Once 1,000 simulation trials have been performed, we report the $\mathrm{AUC}_{\mathrm{SP}}$ in (\ref{eq:aucsp}) and its bootstrapped standard deviation computed with $R=1,000$ replicates for Quantities 1 and 3. We also report the $\mathrm{AUC}_{\mathrm{T1E}}$ in (\ref{eq:auct1e}) and its bootstrapped standard deviation computed with $R=1,000$ replicates for Quantities 2 and 4.

\subsubsection{Simulation Study Results for Multiple Cohorts}

Table \ref{tab:sim2} summarizes the results for every simulation setting. Overall, our RTS method achieves greater statistical power and maintains type I error control for the equal rhythms test when compared to the STS method. The only exception occurs in settings with small sample sizes and low signal-to-noise ratios, where the RTS method exhibits reduced power relative to the STS method. Additionally, the RTS method maintains similar power and type I error rates regardless of whether individual-level phase-shift parameters are distinct across individuals. In contrast, the STS method’s performance varies depending on the presence of this variability. For the equal midlines test, both methods yield comparable power and type I error control across all settings. This finding for the equal midlines test is consistent with Proposition \ref{prop:1}, which states that the midline parameter is unbiased.

Figure \ref{fig:power_32R} and Figure \ref{fig:typeI_32R} display the empirical power and type I error curves used to compute $\mathrm{AUC}_{\mathrm{SP}}$ and $\mathrm{AUC}_{\mathrm{T1E}}$, respectively, for the equal rhythms test. The empirical power curves show that our RTS method has greater statistical power than the STS method at each significance threshold, unless the simulation setting is specified to have a low signal-to-noise ratio as well as a small sample size and cohort size. Notably, the STS method displays anti-conservative type I error control for this test when the null hypothesis is true. This mis-calibration arises from differences in the distributions used to generate individual-level phase-shift parameters across cohorts. For instance, when the phase-shift variability design factor is set to ``low,'' the control cohort generates $b_{i,2k}^{(0)} \sim \mathrm{VM}(0, 8)$, while the case cohort generates $b_{i,2k}^{(1)} \sim \mathrm{VM}(0, 16)$. According to Corollary \ref{cor:1}, these differences in dispersion lead to unequal attenuation of population-level amplitudes across cohorts, which could lead to an incorrect study conclusion that a biological phenomenon oscillates differently for each cohort. Figure \ref{fig:power_32M} and Figure \ref{fig:typeI_32M} show the corresponding power and type I error curves for the equal midlines test. Both our RTS method and the STS method obtain the same performance at each significance threshold.

\section{Real Data Illustrations} \label{sec:4}

\subsection{Setup for Illustrations} \label{sec:4.1}

In this section, we compare parameter estimates and hypothesis test results produced by our RTS method to those produced by the STS method on two circadian biology datasets. For this comparison, we do not assume that the order parameter $K$ for trigonometric regression is known a priori.  Instead, we select $K$ based on a forward selection procedure that is used in practice \citep{Mokon2020}. This procedure can be summarized as follows:
\begin{enumerate}
    \item Starting with $k = 1$, sequentially fit models of increasing order.
    \item At each iteration $k$, we assess the null hypothesis $H_0 : \beta^{(c)}_{2k-1} = 0$ using the bootstrap procedure from Section \ref{sec:2.1.2}. We record the $p$-value output from this test, which we denote as $p_k$.
    \item We retain the $k$-th harmonic for parameter estimation if $p_k < 0.05$. The selection process stops at the first $k$ where $p_k \geq 0.05$, and the selected order parameter would be $K = k - 1$.
\end{enumerate}
We will apply this procedure separately for each method (our RTS method and the STS method) and independently on each cohort in a dataset. To facilitate comparison of population-level parameter estimates and hypothesis test results across methods and cohorts, we will define a common order parameter $K$ as the largest order parameter selected across methods and cohorts.

\subsection{Illustration with Cortisol Levels Derived from Blood Tissue Samples} \label{sec:4.2}

We first analyze cortisol level data previously studied by \citet{Albert2005} and \citet{Wang2003}. The data came from an experiment where blood samples were drawn every two hours over a 24-hour period from three cohorts: nine healthy individuals with no known illnesses (the control cohort with $c=0$), eleven individuals diagnosed with major depressive disorder (MDD, the case cohort with $c=1$), and sixteen individuals with Cushing’s syndrome \citep{Wang1996}. Cortisol levels measured from each blood sample were transformed onto a logarithmic scale. This illustration focuses on the control and MDD cohorts because previous studies have shown that cortisol levels in individuals with Cushing’s syndrome typically do not oscillate \citep{BOYAR1979, LIU1987}.

The order parameter selection framework of Section \ref{sec:4.1} identified an order parameter of $K=3$. Table \ref{tab:ind} provides the population-level parameter estimates and the test statistics computed for the zero amplitudes test, or assessing the null hypothesis $H_0:\beta^{(c)}_{2k-1}=0$ for all $k \in \{1,2,3\}$, on each cohort separately. Each amplitude estimate produced by our RTS method is larger than the corresponding quantity produced by the STS method, which is consistent with the theoretical results in Section \ref{sec:2.2.1} and the numerical example of Section \ref{sec:2.2.2}. The $p$-value produced by the zero amplitudes test is also smaller for our RTS method when compared to the corresponding $p$-value produced by the STS method on the case cohort (both methods produce a bootstrapped $p$-value equal to zero for the control cohort). 

Figure \ref{fig:amps} visualizes the fitted population-level models for each cohort and method. Our RTS method produces multiple peaks in cortisol levels, whereas the STS method produces a single peak. The multiple peaks produced by our RTS method align with previous analyses of cortisol levels, which have been in part attributed to meal-induced cortisol stimulation \citep{Debono2009, LEGLER1982, Stimson2014}. When we perform the equal rhythms test for each method, we find that our RTS method produced a smaller $p$-value ($p=0.560$) when compared to the $p$-value produced by the STS method ($p=0.784$). Both methods produced the same $p$-value for the equal midlines test ($p=0.195$). 

\subsection{Illustration with Heart Rates Obtained from Wearable Devices} \label{sec:4.3}

The MMASH (Multilevel Monitoring of Activity and Sleep in Healthy people) dataset contains psycho-physiological data collected from 22 healthy adult males, including measures of sleep quality, physical activity, and anxiety. For this illustration, we focus on heart rate data recorded continuously over a 24-hour period with wearable heart rate monitors, which capture beat-to-beat intervals. We process these data to extract heart beats-per-minute (BPM) values following the protocol described in the MMASH study \citep{Rossi2020}.

For illustration, we create two cohorts based on each participating individual's reported stress levels based on their responses to the Daily Stress Inventory (DSI), which was taken towards the end of the study. To clarify, the DSI is a 58-item self-report questionnaire in which an individual answers questions about events that occurred during the previous 24-hour period and their perceived impact. The overall DSI score for this questionnaire ranges from 0 to 406, with higher scores reflecting both a greater number and intensity of stressful experiences \citep{Brantley1987}. We assign the 11 individuals with the lowest DSI scores into a ``low-stress cohort,'' and the 11 individuals with the highest DSI scores into a ``high-stress cohort.''

Application of the order selection procedure from Section \ref{sec:4.1} identified an order parameter of $K=13$. Table \ref{tab:ind2} presents population-level parameter estimates and $p$-values computed for the zero amplitudes test, which assesses the null hypothesis $H_0: \beta^{(c)}_{2k-1} = 0$ for all $k \in \{1, \ldots, 13\}$, on each cohort separately. Each amplitude produced by our RTS method is again larger than the corresponding quantity produced by the STS method. Each method produced a bootstrapped $p$-value equal to zero for the zero amplitudes test. 

Figure \ref{fig:amps2} presents the corresponding population-level fits produced by each method for each cohort. The STS method again produces a fit with fewer peaks that are more attenuated when compared to the fit produced by our RTS method. Notably, our RTS method produces three distinct heart rate peaks near 90 BPM for the high-stress cohort, which occur around hours 10, 16, and 20. In contrast, the RTS method produces a single peak at this frequency around hour 10 for the low-stress cohort, followed by relatively stable oscillations around 80 BPM until approximately hour 22. Both methods produced the same $p$-value for the equal rhythms test ($p=0.000$) and for the equal midlines test ($p=0.957$). 

\section{Discussion} \label{sec:5}

In this article, we propose a refined two-stage (RTS) method for analyzing circadian biology data with trigonometric regression. The development of this method is motivated by Proposition \ref{prop:1} and Corollary \ref{cor:1}, which show that individual-level differences in phase-shift parameters can bias population-level parameter estimates produced by the STS method, which could lead to inaccurate study conclusions. The presence of this bias is numerically validated by our simulation studies in Section \ref{sec:3}, which demonstrate that the STS method produces attenuated population-level parameter estimates and has lower statistical power for hypothesis tests. Notably, the STS method can also lose type I error control when comparing oscillations across cohorts. Our RTS method, on the other hand, consistently maintained type I error control for hypothesis testing and does not make assumptions about how the individual-level parameters are generated, which enhances its applicability to a wide array of circadian biology study data.

This study presents opportunities for future methodological research. First, our RTS method could be improved for scenarios in which the population-level amplitudes are small relative to the variance of the random noise (the signal-to-noise ratio) as well as both the sample size and cohort size are small. One approach that could mitigate these issues would involve incorporating biological assumptions about the distribution of individual-level phase-shift parameters \citep{Gorczycaa2024, Gorczycad2024}. Second, the method could be extended to incorporate additional covariates that influence individual-level oscillations, which could enable control of potential confounders in clinical studies.

\section*{Acknowledgments}
This research was supported by the Brazilian National Council for Scientific and Technological Development (CNPq; Proc. 446340/2024-3, 314878/2025-4, 442650/2025-6).

\section*{Conflict of interest}

The authors declare no potential conflict of interests.

\section*{DATA AVAILABILITY STATEMENT}
The authors have made code scripts for reproducing the results from Section \ref{sec:4} available \textcolor{blue}{\href{https://bitbucket.org/michaelgorczyca/two_stage_trigonometric_regression/src/main/}{here}}.


\clearpage
\newpage

\begin{figure}[!h]
\center
\includegraphics[scale=0.25]{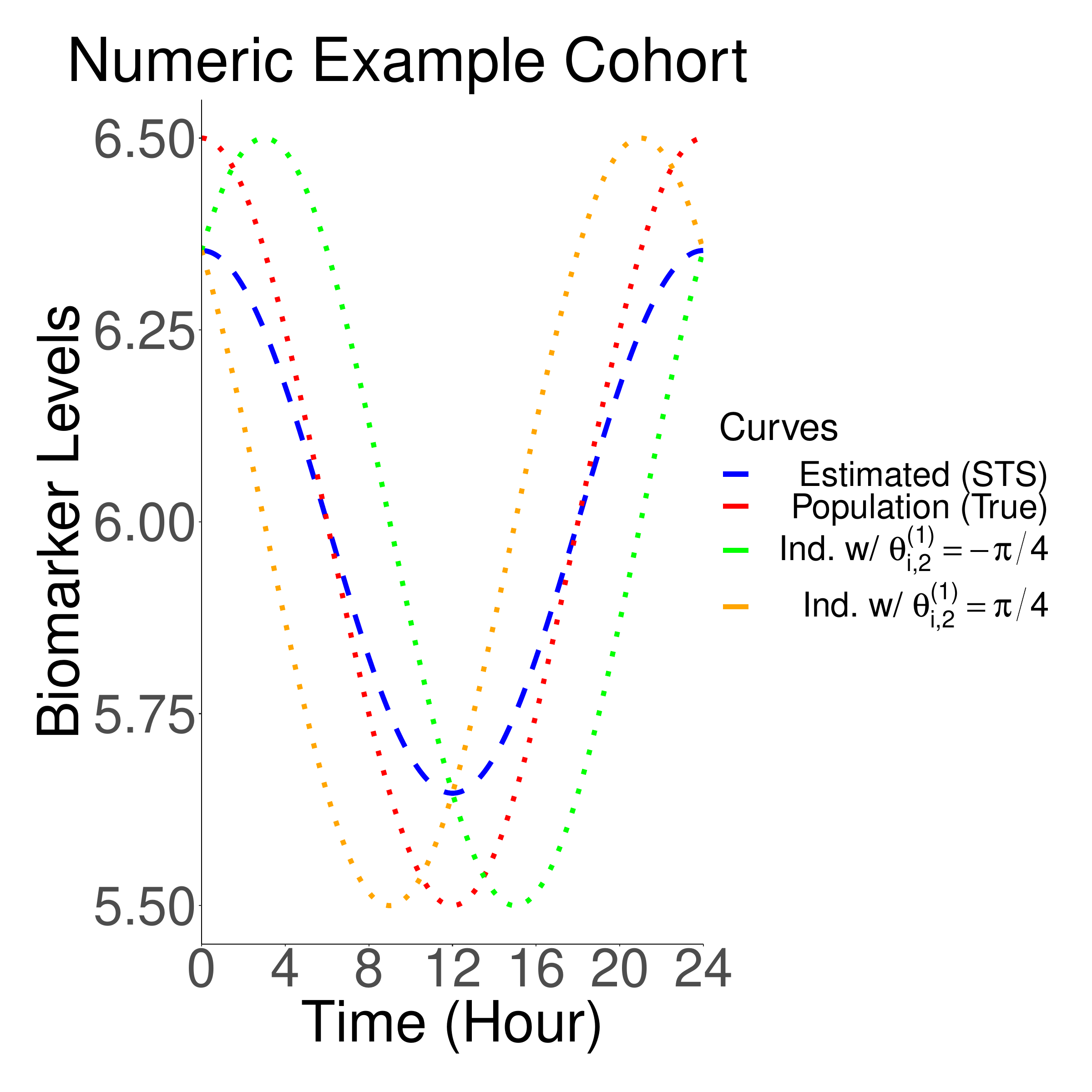}
    \caption{Illustration of using the standard two-stage (STS) method for first-order trigonometric regression. Half the individuals have a phase-shift parameter $\theta_{i,2}^{(1)}=-\pi/4$ (``Ind. w/ $\theta_{i,2}^{(1)}=-\pi/4$''), while the other half have $\theta_{i,2}^{(1)}=\pi/4$ (``Ind. w/ $\theta_{i,2}^{(1)}=\pi/4$''). The dashed blue line represents true population-level biomarker levels over time, while the dotted red line represents population-level biomarker level estimates produced by the STS method. The estimated population-level amplitude is attenuated relative to the true amplitude.} \label{fig:example}
\end{figure}

\clearpage
\newpage
 
\begin{sidewaystable*}
	\caption{Results from the simulation study in Section \ref{sec:3.1} when the number of harmonics design factor is set to $K=1$. Each estimated quantity is shown with its standard deviation in parentheses. Our refined two-stage (RTS) method consistently mitigates attenuation bias and yields higher power for the zero amplitudes test than the standard two-stage (STS) method, except in simulation settings with low phase-shift variability, small sample sizes, and low signal-to-noise ratios. For the ``Dataset'' column, For the ``Dataset'' column, ``DPS'' denotes a dataset generated to have ``different phase-shift'' parameters across individuals, and ``SPS'' denotes a dataset generated to have ``same phase-shift'' parameters for every individual. Bolded entries for the DPS datasets indicate better performance.} \label{tab:sim1_K1}
 \centering
\resizebox{1.0\textwidth}{!}{
  \begin{tabular}{|c|c|c|c|c|c|c|c|c|c|}
			\hline
   Signal-to-Noise Ratio & Sample Size and Cohort Size & Phase-Shift Variability & 
    Dataset& Method  & $\hat{\beta}^{(0)}_0-\beta^{(0)}_0$ & $\hat{\beta}^{(0)}_1-\beta^{(0)}_1$ & $\mathrm{atan2}\{\sin(\hat{\beta}^{(0)}_2-\beta^{(0)}_2), \cos(\hat{\beta}^{(0)}_2-\beta^{(0)}_2)\}$ & $\mathrm{AUC}_{\mathrm{SP}}$ & $\mathrm{AUC}_{\mathrm{T1E}}$  \\
   \hline
   
\multirow{16}{*}{High} & \multirow{8}{*}{Large} & \multirow{4}{*}{High} & \multirow{2}{*}{DPS} & RTS & $\mathbf{4.388\times 10^{-3} \ (0.225)}$ & $\mathbf{2.177\times 10^{-3} \ (5.721\times 10^{-2})}$ & $\mathbf{4.626\times 10^{-3} \ (0.924)}$ & \textbf{1.000 (0.000)} & \textbf{0.346 (0)} \\
& & & & STS& $\mathbf{4.388\times 10^{-3} \ (0.225)}$ & -0.439 (0.140) & $4.940\times 10^{-3}$ (0.195) & 0.882 (0.000) & 0.361 (0.000) \\
\cline{4-10}
& & & \multirow{2}{*}{SPS} & RTS& $4.388\times 10^{-3}$ (0.225) & $1.753\times 10^{-3}$ ($5.766\times 10^{-3}$) & $-4.264\times 10^{-4}$ ($1.548\times 10^{-2}$) & 1.000 (0.000) & 0.346 (0.000) \\
& & & & STS& $4.388\times 10^{-3}$ (0.225) & $-1.669\times 10^{-3}$ ($5.773\times 10^{-2}$) & $-4.831\times 10^{-4}$ ($1.481\times 10^{-2}$) & 1.000 (0.000) & 0.361 (0.000) \\
\cline{3-10}
& & \multirow{4}{*}{Small} & \multirow{2}{*}{DPS} & RTS & $\mathbf{6.177\times 10^{-3} \ (0.219)}$ & $\mathbf{1.275\times 10^{-3} \ (5.630 \times 10^{-3})}$ & $\mathbf{-6.608\times10^{-4} \ (8.427\times 10^{-2})}$ & \textbf{1.000 (0.000)} & \textbf{0.343 (0.000)} \\
& & & & STS& $\mathbf{6.177\times 10^{-3} \ (0.219)}$ & $-9.488\times 10^{-2} \ (6.090\times 10^{-2})$ & $9.018 \times 10^{-4} \ (8.510\times 10^{-2})$ & 1.000 (0.000) & 0.351 (0.000) \\
\cline{4-10}
& & & \multirow{2}{*}{SPS} & RTS & $6.177\times 10^{-3} \ (0.219)$ & $1.345\times 10^{-3}$ ($5.658\times 10^{-2}$) & $-6.225\times 10^{-4}$ ($1.518\times 10^{-2}$) & 1.000 (0.000) & 0.343 (0.000) \\
& & & & STS& $6.177\times 10^{-3} \ (0.219)$ & $-2.052\times 10^{-3}$ ($5.668\times 10^{-2}$) & $-6.115\times 10^{-4}$ ($1.453\times 10^{-2}$) & 1.000 (0.000) & 0.351 (0.000) \\
\cline{2-10}
& \multirow{8}{*}{Small} & \multirow{4}{*}{High} & \multirow{2}{*}{DPS} & RTS & $\mathbf{5.997 \times 10^{-3} \ (0.338)}$ & $\mathbf{6.509 \times 10^{-2} \ (0.148)}$ & $\mathbf{-1.851\times 10^{-2} \ (0.293)}$ & \textbf{0.998 (0.000)} & 0.439 (0.000) \\
& & & & STS& $\mathbf{5.997 \times 10^{-3} \ (0.338)}$ & -0.404 (0.233) & $-1.633\times 10^{-2}$ (0.308) & 0.922 (0.000) & \textbf{0.375 (0.000)} \\
\cline{4-10}
& & & \multirow{2}{*}{SPS} & RTS& $5.997 \times 10^{-3} \ (0.338)$) & $5.854\times 10^{-2} \ (0.149)$ & $3.984\times 10^{-5}$ ($9.120\times 10^{-2}$) & 0.998 (0.000) & 0.439 (0.000) \\
& & & & STS& $5.997\times 10^{-3}$ (0.338) & $5.026\times 10^{-3}$ (0.149) & $3.420\times 10^{-4}$ ($8.413\times 10^{-2}$) & 1.000 (0.000) & 0.375 (0.000) \\
\cline{3-10}
& & \multirow{4}{*}{Small} & \multirow{2}{*}{DPS} & RTS & $\mathbf{1.086\times 10^{-3} \ (0.328)}$ & $\mathbf{6.547\times 10^{-2} \ (0.144)}$ & $\mathbf{7.555\times 10^{-4} \ (0.147)}$ & 0.998 (0.000) & 0.436 (0.000) \\
& & & & STS& $\mathbf{1.086\times 10^{-3} \ (0.328)}$ & $-7.744\times 10^{-2}$ ($0.151$) & $-1.423\times 10^{-3}$ (0.147) & \textbf{1.000 (0.000)} & \textbf{0.347 (0.000)} \\
\cline{4-10}
& & & \multirow{2}{*}{SPS} & RTS & $1.086\times 10^{-3} \ (0.328)$ & $6.503\times 10^{-2}$ (0.145) & $3.935\times 10^{-3}$ ($8.914\times 10^{-2}$) & 0.997 (0.000) & 0.436 (0.000) \\
& & & & STS& $1.086\times 10^{-3} \ (0.328)$ & $1.119\times 10^{-2}$ (0.149) & $2.339\times 10^{-3}$ ($8.260\times 10^{-2}$) & 1.000 (0.000) & 0.347 (0.000) \\
\hline
\multirow{16}{*}{Low} & \multirow{8}{*}{Large} & \multirow{4}{*}{High} & \multirow{2}{*}{DPS} & RTS & $\mathbf{9.226\times 10^{-3} \ (0.223)}$ & $\mathbf{1.051\times 10^{-2} \ (3.012\times 10^{-2})}$ & $\mathbf{-9.929\times 10^{-3} \ (0.197)}$ & \textbf{1.000 (0.000)} & \textbf{0.349 (0.000)} \\
& & & & STS& $\mathbf{9.226\times 10^{-3} \ (0.223)}$ & -0.143 ($5.453\times 10^{-2}$) & $-1.176\times 10^{-2}$ (0.200) & 0.971 (0.000) & 0.353 (0.000) \\
\cline{4-10}
& & & \multirow{2}{*}{SPS} & RTS& $9.226\times 10^{-3} \ (0.223)$ & $1.087\times 10^{-2}$ ($3.077\times 10^{-2}$) & $-2.499\times 10^{-3}$ ($4.739\times 10^{-2}$) & 1.000 (0.000) & 0.349 (0.000) \\
& & & & STS& $9.226\times 10^{-3} \ (0.223)$ & $5.456\times 10^{-4}$ ($3.121\times 10^{-2}$) & $-2.185\times 10^{-3}$ ($4.427\times 10^{-2}$) & 1.000 (0.000) & 0.353 (0.000) \\
\cline{3-10}
& & \multirow{4}{*}{Small} & \multirow{2}{*}{DPS} & RTS & $\mathbf{-6.611\times 10^{-4} \ (0.222)}$ & $\mathbf{1.110\times 10^{-2} (2.815\times 10^{-2})}$ & $\mathbf{-4.282\times 10^{-3} \ (9.734\times 10^{-2})}$ & \textbf{1.000 (0.000)} & \textbf{0.330 (0.000)} \\
& & & & STS& $\mathbf{-6.611\times 10^{-4} \ (0.222)}$ & $-3.034\times 10^{-2}$ ($2.977\times 10^{-2}$) & $-4.403\times 10^{-3}$ ($9.741\times 10^{-2}$) & \textbf{1.000 (0.000)} & 0.358 (0.000) \\
\cline{4-10}
& & & \multirow{2}{*}{SPS} & RTS & $-6.611\times 10^{-4} \ (0.222)$ & $1.109\times 10^{-2}$ ($2.831\times 10^{-2}$) & $-3.673\times 10^{-3}$ ($4.725\times 10^{-2}$) & 1.000 (0.000) & 0.330 (0.000) \\
& & & & STS& $-6.611\times 10^{-4} \ (0.222)$ & $6.686\times 10^{-4}$ ($2.892\times 10^{-2}$) & $-3.637\times 10^{-3}$ ($4.421\times 10^{-2}$) & 1.000 (0.000) & 0.358 (0.000) \\
\cline{2-10}
& \multirow{8}{*}{Small} & \multirow{4}{*}{High} & \multirow{2}{*}{DPS} & RTS & $\mathbf{-2.063\times 10^{-2} \ (0.329)}$ & 0.197 (0.109) & $\mathbf{-3.654\times 10^{-3} \ (0.530)}$ & \textbf{0.728 (0.000)} & 0.431 (0.000) \\
& & & & STS& $\mathbf{-2.063\times 10^{-2} \ (0.329)}$ & $\mathbf{-0.116 (0.138)}$ & $6.777\times 10^{-3}$ (0.508) & 0.720 (0.000) & \textbf{0.351 (0.000)} \\
\cline{4-10}
& & & \multirow{2}{*}{SPS} & RTS &$-2.063\times 10^{-2} \ (0.329)$ & 0.191 (0.104) & $7.801\times 10^{-4}$ (0.293) & 0.715 (0.000) & 0.431 (0.000) \\
& & & & STS& $-2.063\times 10^{-2} \ (0.329)$ & $1.616\times 10^{-2}$ (0.127) & $-1.001\times 10^{-3}$ (0.258) & 0.921 (0.000) & 0.351 (0.000) \\
\cline{3-10}
& & \multirow{4}{*}{Small} & \multirow{2}{*}{DPS} & RTS & $\mathbf{-2.607\times 10^{-3} \ (0.317)}$ & 0.191 (0.108) & $\mathbf{-1.364\times 10^{-3} \ (0.351)}$ & 0.733 (0.000) & 0.433 (0.000) \\
& & & & STS& $\mathbf{-2.607\times 10^{-3} \ (0.317)}$ & $\mathbf{-1.191\times 10^{-2} \ (0.131)}$ & $-8.110\times 10^{-3} \ (0.307)$ & $\mathbf{0.889 \ (0.000)}$ & $\mathbf{0.342 \ (0.000)}$ \\
\cline{4-10}
& & & \multirow{2}{*}{SPS} & RTS & $-2.607\times 10^{-3} \ (0.317)$ & 0.192 (0.109) & $-7.094\times 10^{-3}$ (0.296) & 0.729 (0.000) & 0.433 (0.000) \\
& & & & STS& $-2.607\times 10^{-3} \ (0.317)$ & $1.612\times 10^{-2}$ (0.132) & $-8.981\times 10^{-3}$ (0.262) & 0.915 (0.000) & 0.342 (0.000) \\
\hline
\end{tabular}
}
\end{sidewaystable*}

\clearpage
\newpage

\begin{sidewaystable*}
	\caption{Results from the simulation study in Section \ref{sec:3.1} when the number of harmonics design factor is set to $K=3$. Each estimated quantity is shown with its standard deviation in parentheses. Our refined two-stage (RTS) method consistently mitigates attenuation bias and yields higher power for the zero amplitudes test than the standard two-stage (STS) method, except under low phase-shift variability, small sample sizes, and low signal-to-noise ratios. For the ``Dataset'' column, For the ``Dataset'' column, ``DPS'' denotes a dataset generated to have ``different phase-shift'' parameters across individuals, and ``SPS'' denotes a dataset generated to have ``same phase-shift'' parameters for every individual. Bolded entries for the DPS datasets indicate better performance.} \label{tab:sim1_K3}
 \centering
\resizebox{1.0\textwidth}{!}{
  \begin{tabular}{|c|c|c|c|c|c|c|c|c|c|c|c|c|c|}
			\hline
   Signal-to-Noise Ratio & Sample Size & Phase-Shift Variability & 
    Dataset& Method  &  $\hat{\beta}^{(0)}_0-\beta^{(0)}_0$ & $\hat{\beta}^{(0)}_1-\beta^{(0)}_1$ & $\mathrm{atan2}\{\sin(\hat{\beta}^{(0)}_2-\beta^{(0)}_2), \cos(\hat{\beta}^{(0)}_2-\beta^{(0)}_2)\}$ & $\hat{\beta}^{(0)}_3-\beta^{(0)}_3$ & $\mathrm{atan2}\{\sin(\hat{\beta}^{(0)}_4-\beta^{(0)}_4), \cos(\hat{\beta}^{(0)}_4-\beta^{(0)}_4)\}$ & $\hat{\beta}^{(0)}_5-\beta^{(0)}_5$ & $\mathrm{atan2}\{\sin(\hat{\beta}^{(0)}_6-\beta^{(0)}_6), \cos(\hat{\beta}^{(0)}_6-\beta^{(0)}_6)\}$ & $\mathrm{AUC}_{\mathrm{SP}}$ & $\mathrm{AUC}_{\mathrm{T1E}}$ \\
   \hline
   
\multirow{16}{*}{High} & \multirow{8}{*}{Large} & \multirow{4}{*}{High} & \multirow{2}{*}{DPS} & RTS & $\mathbf{-1.513\times 10^{-2} \ (0.229)}$ & $\mathbf{5.184\times 10^{-3} \ (5.571\times 10^{-2})}$ & $3.435\times 10^{-3} \ (0.185)$ & $\mathbf{6.452\times 10^{-3} \ (5.728\times 10^{-2})}$ & $-4.359\times 10^{-3} \ (0.190)$ & $\mathbf{5.017\times 10^{-3} \ (5.878\times 10^{-2})}$ & $\mathbf{-1.437\times 10^{-3} \ (0.196)}$ & \textbf{1.000 (0.000)} & \textbf{0.311 (0.000)} \\
& & & & STS & $\mathbf{-1.513\times 10^{-2} \ (0.229)}$ & -0.431 (0.143) & $\mathbf{3.321\times 10^{-3} \ (0.188)}$ & -0.432 (0.145) & $\mathbf{-3.376\times 10^{-3} \ (0.192)}$ & -0.426 (0.143) & $-2.498\times 10^{-3}$ (0.200) & 0.973 (0.000) & 0.329 (0.000) \\
\cline{5-14}
& & & \multirow{2}{*}{SPS} & RTS& $-1.513\times 10^{-2} \ (0.229)$ & $5.483\times 10^{-3}$ ($5.627\times 10^{-2}$) & $-9.555\times 10^{-5}$ ($1.560\times 10^{-2}$) & $5.647\times 10^{-3}$ ($5.607\times 10^{-2}$) & $3.467\times 10^{-4}$ ($1.560\times 10^{-2}$) & $3.646\times 10^{-3}$ ($5.862\times 10^{-2}$) & $7.262\times 10^{-5}$ ($1.647\times 10^{-2}$) & 1.000 (0.000) & 0.311 (0.000) \\
& & & & STS & $-1.513\times 10^{-2} \ (0.229)$ & $2.149\times 10^{-3}$ ($5.641\times 10^{-2}$) & $-1.908\times 10^{-4}$ ($1.505\times 10^{-2}$) & $2.259\times 10^{-3}$ ($5.620\times 10^{-2}$) & $3.342\times 10^{-4}$ ($1.499\times 10^{-2}$) & $2.190\times 10^{-4}$ ($5.874\times 10^{-2}$) & $6.785\times 10^{-5}$ ($1.582\times 10^{-2}$) & 1.000 (0.000) & 0.329 (0.000) \\
\cline{3-14}

& & \multirow{4}{*}{Low} & \multirow{2}{*}{DPS} & RTS & $\mathbf{-3.075\times 10^{-3} \ (0.221)}$ & $\mathbf{3.556\times 10^{-3} \ (5.897\times 10^{-2})}$ & $\mathbf{-1.157\times 10^{-3} \ (8.124\times 10^{-2})}$ & $\mathbf{4.653\times 10^{-3} \ (5.538\times 10^{-2})}$ & $3.895\times 10^{-3} \ (8.290\times 10^{-2})$ & $\mathbf{3.536\times 10^{-3} \ (5.456\times 10^{-2})}$ & $\mathbf{-4.876\times 10^{-3} \ (8.687\times 10^{-2})}$ & \textbf{1.000 (0.000)} & \textbf{0.293 (0.000)} \\
& & & & STS & $\mathbf{-3.075\times 10^{-3} \ (0.221)}$ & $-9.156\times 10^{-2} \ (6.396\times 10^{-2})$ & $-1.308\times 10^{-3} \ (8.216\times 10^{-2})$ & $-9.169\times 10^{-2} \ (6.026\times 10^{-2})$ & $\mathbf{3.711\times 10^{-3} \ (8.361\times 10^{-2})}$ & $-9.280\times 10^{-2} \ (5.811\times 10^{-2})$ & $-5.047\times 10^{-3} \ (8.848\times 10^{-2})$ & \textbf{1.000 (0.000)} & 0.339 (0.000) \\
\cline{5-14}
& & & \multirow{2}{*}{SPS} & RTS& $-3.075\times 10^{-3} \ (0.221)$ & $3.168\times 10^{-3} \ (5.915\times 10^{-2})$ & $-7.078\times 10^{-4} \ (1.572\times 10^{-2})$ & $4.398\times 10^{-3} \ (5.560\times 10^{-2})$ & $3.829\times 10^{-4} \ (1.509\times 10^{-2})$ & $3.398\times 10^{-3} \ (5.481\times 10^{-2})$ & $.2034\times 10^{-4} \ (1.588\times 10^{-2})$ & 1.000 (0.000) & 0.293 (0.000) \\
& & & & STS & $-3.075\times 10^{-3} \ (0.221)$ & $-2.348\times 10^{-4} \ (5.926\times 10^{-2})$ & $-7.366\times 10^{-4}$ ($1.531\times 10^{-2}$) & $9.870\times 10^{-4}$ ($5.573\times 10^{-2}$) & $4.450\times 10^{-4}$ ($1.459\times 10^{-2}$) & $3.888\times 10^{-6}$ ($5.489\times 10^{-2}$) & $7.729\times 10^{-5}$ ($1.505\times 10^{-2}$) & 1.000 (0.000) & 0.339 (0.000) \\
\cline{2-14}

 & \multirow{8}{*}{Small} & \multirow{4}{*}{High} & \multirow{2}{*}{DPS} & RTS & $\mathbf{1.478 \times 10^{-2} \ (0.326)}$ & $\mathbf{5.896 \times 10^{-2} \ (0.151)}$ & $\mathbf{-1.605 \times 10^{-2} \ (0.300)}$ & $\mathbf{6.399 \times 10^{-2} \ (0.146)}$ & $\mathbf{-2.790 \times 10^{-3} \ (.302)}$ & $\mathbf{6.434 \times 10^{-2} \ (0.152)}$ & $1.098 \times 10^{-2} \ (0.298)$ & \textbf{0.999 (0.000)} & 0.452 (0.000) \\
& & & & STS & $\mathbf{1.478 \times 10^{-2} \ (0.326)}$ & $-0.421 \ (0.232)$ & $-1.889 \times 10^{-2}$ ($0.0312$) & -0.403 (0.233) & $-4.035 \times 10^{-3}$ (0.313) & -0.404 (0.231) & $\mathbf{9.741 \times 10^{-3} \ (0.309)}$ & 0.949 (0.000) & \textbf{0.309 (0.000)} \\
\cline{5-14}
& & & \multirow{2}{*}{SPS} & RTS& $1.478 \times 10^{-2} \ (0.326)$ & $6.131 \times 10^{-2}$ (0.148) & $-1.077 \times 10^{-3}$ (0.107) & $6.942 \times 10^{-2}$ (0.149) & $4.230 \times 10^{-4}$ (0.101) & $6.605 \times 10^{-2}$ (0.151) & $1.952 \times 10^{-3}$ (0.105) & 0.999 (0.000) & 0.452 (0.000) \\
& & & & STS & $1.478 \times 10^{-2} \ (0.326)$ & $4.985 \times 10^{-4}$ ($0.154$) & $-2.286 \times 10^{-3}$ ($9.873 \times 10^{-2}$) & $7.914 \times 10^{-3}$ (0.154) & $-5.456 \times 10^{-4}$ ($9.522 \times 10^{-2}$) & $3.637 \times 10^{-3}$ (0.155) & $1.615 \times 10^{-3}$ ($9.636 \times 10^{-2}$) & 1.000 (0.000) & 0.309 (0.000) \\

\cline{3-14}
& & \multirow{4}{*}{Low} & \multirow{2}{*}{DPS} & RTS & $\mathbf{1.185 \times 10^{-3} \ (0.327)}$ & $\mathbf{6.770 \times 10^{-2} \ (0.146)}$ & $-2.967 \times 10^{-3} \ (0.155)$ & $\mathbf{6.991 \times 10^{-2} \ (0.154)}$ & $2.627 \times 10^{-3} \ (0.151)$ & $\mathbf{6.290 \times 10^{-2} \ (0.150)}$ & $\mathbf{1.194 \times 10^{-4} \ (0.157)}$ & 0.999 (0.000) & 0.444 (0.000) \\
& & & & STS & $\mathbf{1.185 \times 10^{-3} \ (0.327)}$ & $-8.120 \times 10^{-2}$ (0.160) & $\mathbf{-2.873 \times 10^{-3} \ (0.156)}$ & $-7.900 \times 10^{-2}$ ($0.160$) & $\mathbf{2.185 \times 10^{-3} \ (0.149)}$ & $-8.420 \times 10^{-2}$ ($0.158$) & $7.770 \times 10^{-4}$ ($0.155$) & \textbf{1.000 (0.000)} & \textbf{0.332 (0.000)} \\
\cline{5-14}
& & & \multirow{2}{*}{SPS} & RTS& $1.185 \times 10^{-3}$ ($0.327$) & $6.743 \times 10^{-2}$ ($0.149$) & $1.697 \times 10^{-3}$ ($9.737 \times 10^{-2}$) & $6.827 \times 10^{-2}$ ($0.155$) & $2.144 \times 10^{-3}$ ($9.793 \times 10^{-2}$) & $6.203 \times 10^{-2}$ ($0.152$) & $-4.773 \times 10^{-4}$ ($0.102$) & $0.999$ (0.000) & $0.444$ (0.000) \\
& & & & STS & 1.185 $\times 10^{-3}$ ($0.327$) & $5.613 \times 10^{-3}$ ($0.154$) & $1.631 \times 10^{-3}$ ($9.083 \times 10^{-2}$) & $7.352 \times 10^{-3}$ ($0.158$) & $2.073 \times 10^{-3}$ ($9.103 \times 10^{-2}$) & $1.468 \times 10^{-3}$ ($0.156$) & $-1.43 \times 10^{-3}$ ($9.26 \times 10^{-2}$) & $1.000$ (0.000) & $0.332$ (0.000) \\

\hline
\multirow{16}{*}{Low} & \multirow{8}{*}{Large} & \multirow{4}{*}{High} & \multirow{2}{*}{DPS} & RTS & $\mathbf{-1.032 \times 10^{-3} \ (0.226)}$ & $\mathbf{9.954 \times 10^{-3} \ (3.023 \times 10^{-2})}$ & $\mathbf{-8.437 \times 10^{-3} \ (0.2)}$ & $\mathbf{1.130 \times 10^{-2} \ (2.739 \times 10^{-2})}$ & $\mathbf{-7.765 \times 10^{-3} \ (0.197)}$ & $\mathbf{1.108 \times 10^{-2} \ (2.837 \times 10^{-2})}$ & $\mathbf{1.028 \times 10^{-3} \ (0.199)}$ & \textbf{1.000 (0.000)} & \textbf{0.286 (0.000)} \\
& & & & STS & $\mathbf{-1.032 \times 10^{-3} \ (0.226)}$ & $-0.144$ ($5.156 \times 10^{-2}$) & $-1.042 \times 10^{-2}$ ($0.205$) & $-0.146$ ($5.056 \times 10^{-2}$) & $-9.871 \times 10^{-3}$ ($0.202$) & $-0.144$ ($5.24 \times 10^{-2}$) & $2.907 \times 10^{-3}$ ($0.202$) & $0.994$ (0.000) & $0.335$ (0.000) \\
\cline{5-14}
& & & \multirow{2}{*}{SPS} & RTS& $-1.032 \times 10^{-3}$ ($0.226$) & $1.045 \times 10^{-2}$ ($3.012 \times 10^{-2}$) & $-2.315 \times 10^{-3}$ ($4.812 \times 10^{-2}$) & $1.166 \times 10^{-2}$ ($2.763 \times 10^{-2}$) & $1.295 \times 10^{-3}$ ($4.979 \times 10^{-2}$) & $1.076 \times 10^{-2}$ ($2.867 \times 10^{-2}$) & $2.458 \times 10^{-3}$ ($5.034 \times 10^{-2}$) & $1.000$ (0.000) & $0.286$ (0.000) \\
& & & & STS & $-1.032 \times 10^{-3}$ ($0.226$) & $3.656 \times 10^{-5}$ ($3.05 \times 10^{-2}$) & $-2.138 \times 10^{-3}$ ($4.477 \times 10^{-2}$) & $1.327 \times 10^{-3}$ ($2.808 \times 10^{-2}$) & $1.182 \times 10^{-3}$ ($4.587 \times 10^{-2}$) & $5.013 \times 10^{-4}$ ($2.912 \times 10^{-2}$) & $2.019 \times 10^{-3}$ ($4.703 \times 10^{-2}$) & $1.000$ (0.000) & $0.335$ (0.000) \\

\cline{3-14}
& & \multirow{4}{*}{Low} & \multirow{2}{*}{DPS} & RTS & $\mathbf{1.629 \times 10^{-2} (0.226)}$ & $\mathbf{9.889 \times 10^{-3} (2.9 \times 10^{-2})}$ & $7.313 \times 10^{-4} \ (9.628 \times 10^{-2})$ & $\mathbf{1.181 \times 10^{-2} \ (2.833 \times 10^{-2})}$ & $-1.808 \times 10^{-3} \ (9.71 \times 10^{-2})$ & $\mathbf{9.748 \times 10^{-3} \ (2.921 \times 10^{-2})}$ & $1.024 \times 10^{-3} \ (9.560 \times 10^{-2})$ & \textbf{1.000 (0.000)} & \textbf{0.309 (0.000)} \\
& & & & STS & $\mathbf{1.629 \times 10^{-2} \ (0.226)}$ & $-3.026 \times 10^{-2}$ ($3.029 \times 10^{-2}$) & $\mathbf{-6.150 \times 10^{-4} \ (9.623 \times 10^{-2})}$ & $-2.987 \times 10^{-2}$ ($2.944 \times 10^{-2}$) & $\mathbf{-7.974 \times 10^{-4} \ (9.857 \times 10^{-2})}$ & $-3.115 \times 10^{-2}$ ($3.051 \times 10^{-2}$) & $\mathbf{8.018 \times 10^{-4} \ (9.577 \times 10^{-2})}$ & \textbf{1.000 (0.000)} & $0.341$ (0.000) \\
\cline{5-14}
& & & \multirow{2}{*}{SPS} & RTS& $1.629 \times 10^{-2}$ ($0.226$) & $1.035 \times 10^{-2}$ ($2.918 \times 10^{-2}$) & $1.988 \times 10^{-3}$ ($4.874 \times 10^{-2}$) & $1.135 \times 10^{-2}$ ($2.806 \times 10^{-2}$) & $-1.072 \times 10^{-3}$ ($4.861 \times 10^{-2}$) & $9.881 \times 10^{-3}$ ($2.906 \times 10^{-2}$) & $-1.402 \times 10^{-3}$ ($4.81 \times 10^{-2}$) & $1.000$ (0.000) & $0.309$ (0.000) \\
& & & & STS & $1.629 \times 10^{-2}$ ($0.226$) & $1.774 \times 10^{-4}$ ($2.984 \times 10^{-2}$) & $1.449 \times 10^{-3}$ ($4.514 \times 10^{-2}$) & $1.059 \times 10^{-3}$ ($2.84 \times 10^{-2}$) & $-7.82 \times 10^{-4}$ ($4.512 \times 10^{-2}$) & $-5.479 \times 10^{-4}$ ($2.968 \times 10^{-2}$) & $-8.682 \times 10^{-4}$ ($4.535 \times 10^{-2}$) & $1.000$ (0.000) & $0.341$ (0.000) \\

\cline{2-14}
& \multirow{8}{*}{Small} & \multirow{4}{*}{High} & \multirow{2}{*}{DPS} & RTS & $\mathbf{2.565 \times 10^{-3} \ (0.0329)}$ & $0.196 \ (0.109)$ & $4.111 \times 10^{-2} \ (0.546)$ & $0.195 \ (0.106)$ & $\mathbf{2.274 \times 10^{-2} \ (0.551)}$ & $0.192 \ (0.109)$ & $2.041 \times 10^{-3} \ (0.527)$ & $\mathbf{0.781 \ (0.000)}$ & $0.461 \ (0.000)$ \\
& & & & STS & $\mathbf{2.565 \times 10^{-3} \ (0.329)}$ & $\mathbf{-0.105 \ (0.141)}$ & $\mathbf{2.554 \times 10^{-2} \ (0.526)}$ & $\mathbf{-0.112 \ (0.136)}$ & $2.373 \times 10^{-2}$ ($0.542$) & $\mathbf{-0.111 \ (0.136)}$ & $\mathbf{1.618 \times 10^{-3} \ (0.499)}$ & $0.683$ (0.000) & \textbf{0.313 (0.000)} \\
\cline{5-14}
& & & \multirow{2}{*}{SPS} & RTS& $2.565 \times 10^{-3}$ ($0.329$) & $0.199$ ($0.113$) & $1.539 \times 10^{-2}$ ($0.31$) & $0.193$ ($0.107$) & $1.126 \times 10^{-2}$ ($0.311$) & $0.194$ ($0.107$) & $-9.906 \times 10^{-3}$ ($0.311$) & $0.783$ (0.000) & $0.461$ (0.000) \\
& & & & STS & $2.565 \times 10^{-3}$ ($0.329$) & $2.035 \times 10^{-2}$ ($0.136$) & $1.312 \times 10^{-2}$ ($0.272$) & $1.779 \times 10^{-2}$ ($0.13$) & $9.64 \times 10^{-3}$ ($0.277$) & $1.88 \times 10^{-2}$ ($0.129$) & $-6.937 \times 10^{-3}$ ($0.275$) & $0.89$ (0.000) & $0.313$ (0.000) \\

\cline{3-14}
& & \multirow{4}{*}{Low} & \multirow{2}{*}{DPS} & RTS & $\mathbf{-1.452 \times 10^{-2} \ (0.329)}$ & $0.195 \ (0.110)$ & $0.221 \times 10^{-2} \ (0.358)$ & $0.195 \ (0.109)$ & $\mathbf{1.140 \times 10^{-2} \ (0.366)}$ & $0.190 \ (0.106)$ & $8.885 \times 10^{-3} \ (0.363)$ & $0.772 \ (0.000)$ & $0.453 \ (0.000)$ \\
& & & & STS & $\mathbf{-1.452 \times 10^{-2} \ (0.329)}$ & $\mathbf{-6.978 \times 10^{-3} \ (0.131)}$ & $\mathbf{1.250 \times 10^{-2} \ (0.306)}$ & $\mathbf{-7.975 \times 10^{-3} \ (0.130)}$ & $1.203 \times 10^{-2}$ ($0.322$) & $\mathbf{-1.099 \times 10^{-2} \ (0.128)}$ & $\mathbf{2.028 \times 10^{-3} \ (0.320)}$ & \textbf{0.853 (0.000)} & \textbf{0.317 (0.000)} \\
\cline{5-14}
& & & \multirow{2}{*}{SPS} & RTS& $-1.452 \times 10^{-2}$ ($0.329$) & $0.196$ ($0.11$) & $1.359 \times 10^{-2}$ ($0.307$) & $0.195$ ($0.109$) & $1.502 \times 10^{-2}$ ($0.325$) & $0.192$ ($0.107$) & $9.254 \times 10^{-3}$ ($0.323$) & $0.78$ (0.000) & $0.453$ (0.000) \\
& & & & STS & $-1.452 \times 10^{-2}$ ($0.329$) & $2.067 \times 10^{-2}$ ($0.131$) & $1.100 \times 10^{-2}$ ($0.265$) & $1.994 \times 10^{-2}$ ($0.131$) & $1.400 \times 10^{-2}$ ($0.286$) & $1.776 \times 10^{-2}$ ($0.129$) & $7.321 \times 10^{-3}$ ($0.279$) & $0.894$ (0.000) & $0.317$ (0.000) \\
\hline
\end{tabular}
}
\end{sidewaystable*}

\clearpage
\newpage

\begin{figure}[!h]
\center
\includegraphics[scale=0.2]{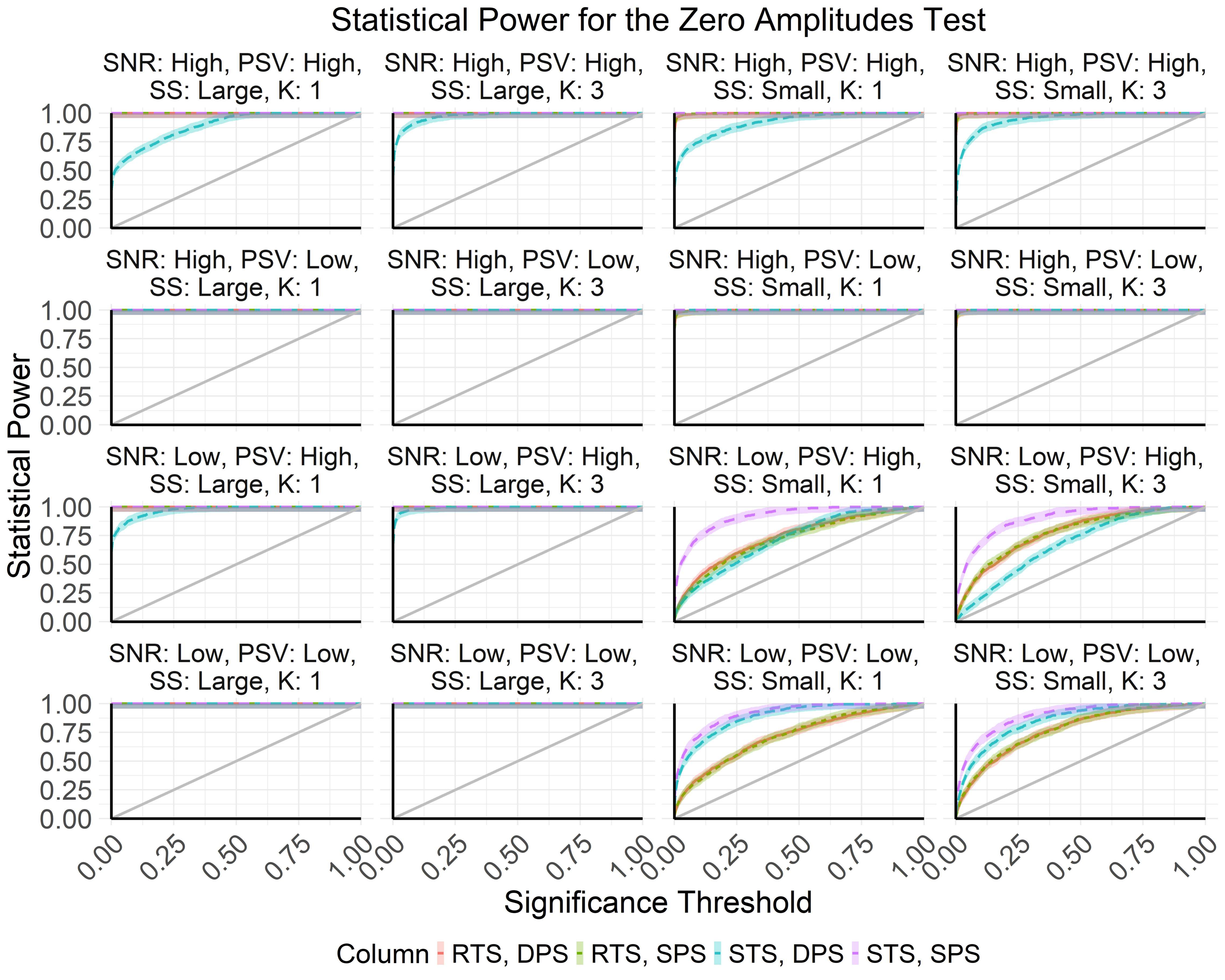}
    \caption{Empirical power curves used to compute $\mathrm{AUC}_{\mathrm{SP}}$ for the zero amplitudes test. Our RTS method generally outperforms STS across all significance thresholds, except in settings with low phase-shift variability, small sample sizes, and low signal-to-noise ratios. Here, ``DPS'' denotes curves computed from datasets generated to have ``different phase-shift'' parameters across individuals, and ``SPS'' denotes curves computed from datasets generated to have ``same phase-shift'' parameters for every individual.} \label{fig:power_31}
\end{figure}

\clearpage
\newpage

\begin{figure}[!h]
\center
\includegraphics[scale=0.2]{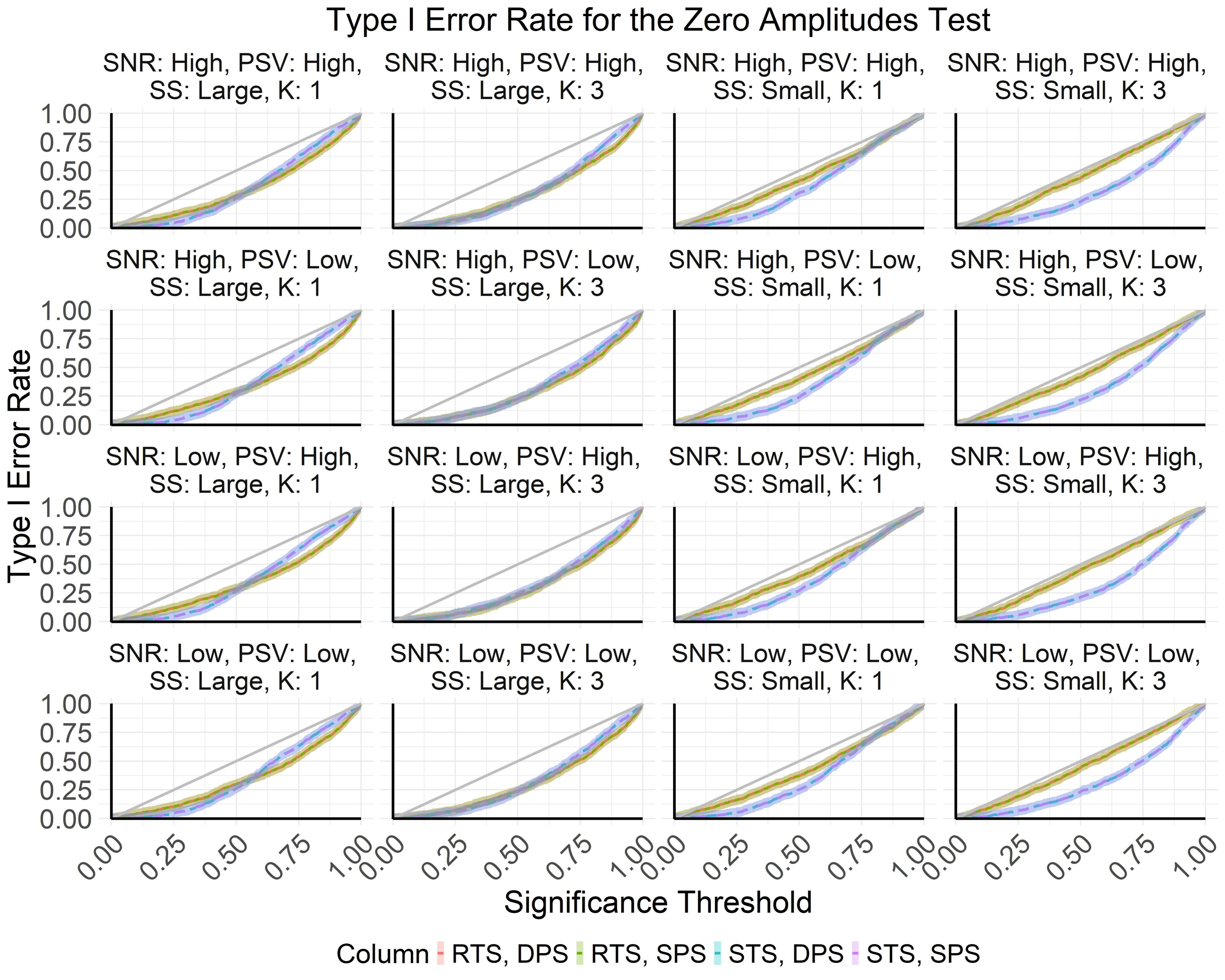}
    \caption{Empirical type I error curves used to compute $\mathrm{AUC}_{\mathrm{T1E}}$ for the zero amplitudes test. Both RTS and STS methods are conservative in rejecting the null hypothesis when it is true. Here, ``DPS'' denotes curves computed from datasets generated to have ``different phase-shift'' parameters across individuals, and ``SPS'' denotes curves computed from datasets generated to have ``same phase-shift'' parameters for every individual.} \label{fig:typeI_31}
\end{figure}

\clearpage
\newpage

\begin{table}
	\caption{Results from the simulation study in Section \ref{sec:3.2}. Our RTS method consistently yields higher statistical power for the equal midlines test and the equal rhythms test when compared to the STS method, except in simulation settings with small sample sizes and low signal-to-noise ratios. For the ``Dataset'' column, ``DPS'' denotes a dataset generated to have ``different phase-shift'' parameters across individuals, and ``SPS'' denotes a dataset generated to have ``same phase-shift'' parameters for every individual. Bolded entries for the DPS datasets indicate better performance.} \label{tab:sim2}
 \centering
\resizebox{1.0\textwidth}{!}{
  \begin{tabular}{|c|c|c|c|c|c|c|c|c|c|}
			\hline
   Signal-to-Noise Ratio & Sample Size & Phase-Shift Variability & Number of Harmonics &
    Dataset& Method  & Statistical Power (Equal Midlines) & Statistical Power (Equal Rhythms) & Type I Error Rate (Equal Midlines) & Type I Error Rate (Equal Rhythms) \\
   \hline
   
\multirow{32}{*}{High} & \multirow{16}{*}{Large} & \multirow{8}{*}{High} & \multirow{4}{*}{$K=1$} & \multirow{2}{*}{DPS} & RTS & $\mathbf{0.969 \ (2.840 \times 10^{-3})}$ & $\mathbf{1.000 \ (3.553 \times 10^{-6})}$ & $\mathbf{0.488 \ (9.227 \times 10^{-3})}$ & $\mathbf{0.496 \ (9.084 \times 10^{-3})}$ \\
& & & & & STS & $\mathbf{0.969 \ (2.840 \times 10^{-3})}$ & $0.968$ ($2.934 \times 10^{-3}$) & $\mathbf{0.488 \ (9.227 \times 10^{-3})}$ & $0.724$ ($8.326 \times 10^{-3}$) \\
\cline{5-10}
& & & & \multirow{2}{*}{SPS} & RTS& $0.969$ ($2.839 \times 10^{-3}$) & $1.000$ (0.000) & $0.488$ ($9.227 \times 10^{-3}$) & $0.449$ ($8.707 \times 10^{-3}$) \\
& & & & & STS & $0.969$ ($2.839\times 10^{-3}$) & $1.000$ (0.000) & $0.488$ ($9.227\times 10^{-3}$) & $0.446$ ($8.680\times 10^{-3}$) \\
\cline{4-10}

& & & \multirow{4}{*}{$K=3$} & \multirow{2}{*}{DPS} & RTS & $\mathbf{0.969 \ (2.676 \times 10^{-3})}$ & $\mathbf{1.000 \ (0.000)}$ & $\mathbf{0.508 \ (9.261 \times 10^{-3})}$ & $\mathbf{0.503 \ (8.848 \times 10^{-3})}$ \\
& & & & & STS & $\mathbf{0.969 \ (2.676 \times 10^{-3})}$ & $\mathbf{1.000 \ (1.520 \times 10^{-5})}$ & $\mathbf{0.508 \ (9.261 \times 10^{-3})}$ & $0.930 \ (4.088 \times 10^{-3})$ \\
\cline{5-10}
& & & & \multirow{2}{*}{SPS} & RTS& $0.968$ ($2.682 \times 10^{-3}$) & $1.000$ (0.000) & $0.508$ ($9.261 \times 10^{-3}$) & $0.364$ ($7.822 \times 10^{-3}$) \\
& & & & & STS & $0.968$ ($2.682 \times 10^{-3}$) & $1.000$ (0.000) & $0.508$ ($9.261 \times 10^{-3}$) & $0.523$ ($9.21 \times 10^{-3}$) \\

\cline{3-10}
& & \multirow{8}{*}{Low} & \multirow{4}{*}{$K=1$} & \multirow{2}{*}{DPS} & RTS & $\mathbf{0.974 \ (2.101 \times 10^{-3})}$ & \textbf{1.000 (0.000)} & $\mathbf{0.508 \ (9.294 \times 10^{-3})}$ & $\mathbf{0513 \ (9.078 \times 10^{-3})}$ \\
& & & & & STS & $\mathbf{0.974 \ (2.101 \times 10^{-3})}$ & $\mathbf{1.000 \ (0.000)}$ & $\mathbf{0.508 \ (9.294 \times 10^{-3})}$ & $0.559$ ($9.233 \times 10^{-3}$) \\
\cline{5-10}
& & & & \multirow{2}{*}{SPS} & RTS& $0.974$ ($2.103 \times 10^{-3}$) & $1.000$ (0.000) & $0.508 \ (9.289 \times 10^{-3})$ & $0.464$ ($9.065 \times 10^{-3}$) \\
& & & & & STS & $0.9741$ ($2.103 \times 10^{-3}$) & $1.000$ (0.000) & $0.508$ ($9.289 \times 10^{-3}$) & $0.463$ ($9.091 \times 10^{-3}$) \\

\cline{4-10}
& & & \multirow{4}{*}{$K=3$} & \multirow{2}{*}{DPS} & RTS & $\mathbf{0.971 \ (2.732 \times 10^{-3})}$ & $\mathbf{1.000 \ (0.000)}$ & $\mathbf{0.494 \ (9.244 \times 10^{-3})}$ & $\mathbf{0.493 \ (8.709 \times 10^{-3})}$ \\
& & & & & STS & $\mathbf{0.971 \ (2.732 \times 10^{-3})}$ & $\mathbf{1.000 \ (0.000)}$ & $\mathbf{0.494 \ (9.244 \times 10^{-3})}$ & $0.727$ ($8.412 \times 10^{-3}$) \\
\cline{5-10}
& & & & \multirow{2}{*}{SPS} & RTS& $0.971$ ($2.731 \times 10^{-3}$) & $1.000$ (0.000) & $0.494$ ($9.244 \times 10^{-3}$) & $0.352$ ($7.75 \times 10^{-3}$) \\
& & & & & STS & $0.971$ ($2.731 \times 10^{-3}$) & $1.000$ (0.000) & $0.494$ ($9.244 \times 10^{-3}$) & $0.511$ ($9.229 \times 10^{-3}$) \\
\cline{2-10}

& \multirow{16}{*}{Small} & \multirow{8}{*}{High} & \multirow{4}{*}{$K=1$} & \multirow{2}{*}{DPS} & RTS & $\mathbf{0.882 \ (5.929 \times 10^{-3})}$ & $\mathbf{0.938 \ (3.569 \times 10^{-3})}$ & $\mathbf{0.487 \ (9.154 \times 10^{-3})}$ & $\mathbf{0.474 \ (9.078 \times 10^{-3})}$ \\
& & & & & STS & $\mathbf{0.882 \ (5.929 \times 10^{-3})}$ & $0.818$ ($7.375 \times 10^{-3}$) & $\mathbf{0.487 \ (9.154 \times 10^{-3})}$ & $0.550$ ($9.047 \times 10^{-3}$) \\
\cline{5-10}
& & & & \multirow{2}{*}{SPS} & RTS& $0.882$ ($5.926 \times 10^{-3}$) & $0.991$ ($9.315 \times 10^{-4}$) & $0.486$ ($9.154 \times 10^{-3}$) & $0.415$ ($8.523 \times 10^{-3}$) \\
& & & & & STS & $0.882$ ($5.926 \times 10^{-3}$) & $0.997$ ($3.401 \times 10^{-4}$) & $0.486$ ($9.154 \times 10^{-3}$) & $0.407$ ($8.636 \times 10^{-3}$) \\
\cline{4-10}

& & & \multirow{4}{*}{$K=3$} & \multirow{2}{*}{DPS} & RTS & $\mathbf{0.876 \ (5.744 \times 10^{-3})}$ & $\mathbf{0.979 \ (1.274 \times 10^{-3})}$ & $\mathbf{0.489 \ (8.908 \times 10^{-3})}$ & $\mathbf{0.389 \ (8.378 \times 10^{-3})}$ \\
& & & & & STS & $\mathbf{0.876 \ (5.744 \times 10^{-3})}$ & $0.967$ ($2.449 \times 10^{-3}$) & $\mathbf{0.489 \ (8.908 \times 10^{-3})}$ & $0.704$ ($8.097 \times 10^{-3}$) \\
\cline{5-10}
&  & & & \multirow{2}{*}{SPS} & RTS& $0.876$ ($5.737 \times 10^{-3}$) & $0.998$ ($2.367 \times 10^{-4}$) & $0.489$ ($8.904 \times 10^{-3}$) & $0.256$ ($6.909 \times 10^{-3}$) \\
&  & & & & STS & $0.876$ ($5.737 \times 10^{-3}$) & $1.000$ ($2.155 \times 10^{-5}$) & $0.489$ ($8.904 \times 10^{-3}$) & $0.407$ ($8.803 \times 10^{-3}$) \\
\cline{3-10}

& & \multirow{8}{*}{Low} & \multirow{4}{*}{$K=1$} & \multirow{2}{*}{DPS} & RTS & $\mathbf{0.873 \ (5.895 \times 10^{-3})}$ & $0.981 \ (1.695 \times 10^{-3})$ & $\mathbf{0.496 \ (8.985 \times 10^{-3})}$ & $0.451 \ (8.649 \times 10^{-3})$ \\
& & & & & STS & $\mathbf{0.873 \ (5.895 \times 10^{-3})}$  & $\mathbf{0.984 \ (1.243 \times 10^{-3})}$ & $\mathbf{0.496 \ (8.985 \times 10^{-3})}$ & $\mathbf{0.450 \ (8.712 \times 10^{-3})}$ \\
\cline{5-10}
& & & & \multirow{2}{*}{SPS} & RTS& $0.873$ ($5.893 \times 10^{-3}$) & $0.992$ ($1.044 \times 10^{-3}$) & $0.496$ ($8.987 \times 10^{-3}$) & $0.419$ ($8.532 \times 10^{-3}$) \\
& & & & & STS & $0.873$ ($5.893 \times 10^{-3}$) & $0.997$ ($3.069 \times 10^{-4}$) & $0.496$ ($8.987 \times 10^{-3}$) & $0.411$ ($8.403 \times 10^{-3}$) \\
\cline{4-10}

& & & \multirow{4}{*}{$K=3$} & \multirow{2}{*}{DPS} & RTS & $\mathbf{0.864 \ (6.386 \times 10^{-3})}$ & $0.996 \ (3.445 \times 10^{-4})$ & $\mathbf{0.495 \ (8.995 \times 10^{-3})}$ & $\mathbf{0.306 \ (7.21 \times 10^{-3})}$ \\
& & & & & STS & $\mathbf{0.864 \ (6.386 \times 10^{-3})}$ & $\mathbf{0.999 \ (9.745 \times 10^{-5})}$ & $\mathbf{0.495 \ (8.995 \times 10^{-3})}$ & $0.489$ ($8.855 \times 10^{-3}$) \\
\cline{5-10}
& & & & \multirow{2}{*}{SPS} & RTS& $0.864$ ($6.384 \times 10^{-3}$) & $0.998$ ($1.989 \times 10^{-4}$) & $0.495$ ($8.985 \times 10^{-3}$) & $0.249$ ($6.155 \times 10^{-3}$) \\
& & & & & STS & $0.864$ ($6.384 \times 10^{-3}$) & $1.000$ ($2.436 \times 10^{-5}$) & $0.495$ ($8.985 \times 10^{-3}$) & $0.399$ ($8.103 \times 10^{-3}$) \\
\hline

\multirow{32}{*}{Low} & \multirow{16}{*}{Large} & \multirow{8}{*}{High} & \multirow{4}{*}{$K=1$} & \multirow{2}{*}{DPS} & RTS & $\mathbf{1.000 \ (2.511 \times 10^{-5})}$ & $\mathbf{1.000 \ (1.270 \times 10^{-5})}$ & $\mathbf{0.503 \ (9.160 \times 10^{-3})}$ & $\mathbf{0.489 \ (9.037 \times 10^{-3})}$ \\
& & & & & STS & $\mathbf{1.000 \ (2.511 \times 10^{-5})}$ & $0.970$ ($2.797 \times 10^{-3}$) & $\mathbf{0.503 \ (9.160 \times 10^{-3})}$ & $0.665$ ($8.795 \times 10^{-3}$) \\
\cline{5-10}
& & & & \multirow{2}{*}{SPS} & RTS& $1.000$ ($2.511 \times 10^{-5}$) & $1.000$ ($2.508 \times 10^{-6}$) & $0.503$ ($9.164 \times 10^{-3}$) & $0.456$ ($8.935 \times 10^{-3}$) \\
& & & & & STS & $1.000$ ($2.511 \times 10^{-5}$) & $1.000$ (0.000) & $0.503$ ($9.164 \times 10^{-3}$) & $0.452$ ($8.95 \times 10^{-3}$) \\
\cline{4-10}

& & & \multirow{4}{*}{$K=3$} & \multirow{2}{*}{DPS} & RTS & $\mathbf{1.000 \ (1.654 \times 10^{-5})}$ & \textbf{1.000 (0.000)} & $\mathbf{0.511 \ (9.138 \times 10^{-3})}$ & $\mathbf{0.425 \ (8.396 \times 10^{-3})}$ \\
& & & & & STS & $\mathbf{1.000 \ (1.654 \times 10^{-5})}$ & $\mathbf{1.000 \ (3.215 \times 10^{-5})}$ & $\mathbf{0.511 \ (9.138 \times 10^{-3})}$ & $0.867$ ($5.772 \times 10^{-3}$) \\
\cline{5-10}
& & & & \multirow{2}{*}{SPS} & RTS& $1.000$ ($1.636 \times 10^{-5}$) & $1.000$ (0.000) & $0.511$ ($9.134 \times 10^{-3}$) & $0.278$ ($6.589 \times 10^{-3}$) \\
& & & & & STS & $1.000$ ($1.636 \times 10^{-5}$) & $1.000$ (0.000) & $0.511$ ($9.134 \times 10^{-3}$) & $0.435$ ($8.683 \times 10^{-3}$) \\

\cline{3-10}
& & \multirow{8}{*}{Low} & \multirow{4}{*}{$K=1$} & \multirow{2}{*}{DPS} & RTS & $\mathbf{1.000 \ (2.060 \times 10^{-5})}$ & $\mathbf{1.000 \ (2.505 \times 10^{-6}})$ & $\mathbf{0.479 \ (9.003 \times 10^{-3})}$ & $\mathbf{0.475 \ (8.837 \times 10^{-3})}$ \\
& & & & & STS & $\mathbf{1.000 \ (2.060 \times 10^{-5})}$ & $\mathbf{1.000 \ (2.508 \times 10^{-6})}$ & $\mathbf{0.479 \ (9.003 \times 10^{-3})}$ & $0.487$ ($9.032 \times 10^{-3}$) \\
\cline{5-10}
& & & & \multirow{2}{*}{SPS} & RTS& $1.000$ ($2.177 \times 10^{-5}$) & $1.000$ (0.000) & $0.479$ ($9.003 \times 10^{-3}$) & $0.446$ ($8.732 \times 10^{-3}$) \\
& & & & & STS & $1.000$ ($2.177 \times 10^{-5}$) & $1.000$ (0.000) & $0.479$ ($9.003 \times 10^{-3}$) & $0.44$ ($8.68 \times 10^{-3}$) \\

\cline{4-10}
& & & \multirow{4}{*}{$K=3$} & \multirow{2}{*}{DPS} & RTS & $\mathbf{1.000 \ (1.113 \times 10^{-5})}$ & $\mathbf{1.000 \ (0.000)}$ & $\mathbf{0.513 \ (9.257 \times 10^{-3})}$ & $\mathbf{0.371 \ (7.778 \times 10^{-3})}$ \\
& & & & & STS & $\mathbf{1.000 \ (1.113 \times 10^{-5})}$ & $\mathbf{1.000 \ (0.000)}$ & $\mathbf{0.513 \ (9.257 \times 10^{-3})}$ & $0.576$ ($9.129 \times 10^{-3}$) \\
\cline{5-10}
& & & & \multirow{2}{*}{SPS} & RTS& $1.000$ ($1.113\times 10^{-5}$) & $1.000$ (0.000) & $0.513$ ($9.258\times 10^{-3}$) & $0.292$ ($6.772\times 10^{-3}$) \\
& & & & & STS & $1.000$ ($1.113 \times 10^{-5}$) & $1.000$ (0.000) & $0.513$ ($9.258 \times 10^{-3}$) & $0.454$ ($8.812 \times 10^{-3}$) \\
\cline{2-10}

& \multirow{16}{*}{Small} & \multirow{8}{*}{High} & \multirow{4}{*}{$K=1$} & \multirow{2}{*}{DPS} & RTS & $\mathbf{0.995 \ (4.843 \times 10^{-4})}$ & $\mathbf{0.581 \ (8.616 \times 10^{-3})}$ & $\mathbf{0.492 \ (9.005 \times 10^{-3})}$ & $\mathbf{0.415 \ (8.356 \times 10^{-3})}$ \\
& & & & & STS & $\mathbf{0.995 \ (4.843 \times 10^{-4})}$ & $0.575$ ($9.163 \times 10^{-3}$) & $\mathbf{0.492 \ (9.005 \times 10^{-3})}$ & $0.449$ ($8.705 \times 10^{-3}$) \\
\cline{5-10}
& & & & \multirow{2}{*}{SPS} & RTS& $0.995$ ($4.817 \times 10^{-4}$) & $0.618$ ($8.657 \times 10^{-3}$) & $0.492$ ($9.007 \times 10^{-3}$) & $0.392$ ($8.33 \times 10^{-3}$) \\
& & & & & STS & $0.995$ ($4.817 \times 10^{-4}$) & $0.733$ ($7.981 \times 10^{-3}$) & $0.492$ ($9.007 \times 10^{-3}$) & $0.407$ ($8.417 \times 10^{-3}$) \\
\cline{4-10}

& & & \multirow{4}{*}{$K=3$} & \multirow{2}{*}{DPS} & RTS & $\mathbf{0.995 \ (5.683 \times 10^{-4})}$ & $0.574 \ (8.582 \times 10^{-3})$ & $\mathbf{0.498 \ (9.078 \times 10^{-3})}$ & $\mathbf{0.274 \ (7.016 \times 10^{-3})}$ \\
& & & & & STS & $\mathbf{0.995 \ (5.683 \times 10^{-4})}$ & $\mathbf{0.705 \ (8.065 \times 10^{-3})}$ & $\mathbf{0.498 \ (9.078 \times 10^{-3})}$ & $0.479$ ($8.971 \times 10^{-3}$) \\
\cline{5-10}
& & & & \multirow{2}{*}{SPS} & RTS& $0.995$ ($5.661 \times 10^{-4}$) & $0.611$ ($8.533 \times 10^{-3}$) & $0.498$ ($9.075 \times 10^{-3}$) & $0.21$ ($6.059 \times 10^{-3}$) \\
& & & & & STS & $0.995$ ($5.661 \times 10^{-4}$) & $0.874$ ($5.103 \times 10^{-3}$) & $0.498$ ($9.075 \times 10^{-3}$) & $0.396$ ($8.601 \times 10^{-3}$) \\
\cline{3-10}

& & \multirow{8}{*}{Low} & \multirow{4}{*}{$K=1$} & \multirow{2}{*}{DPS} & RTS & $\mathbf{0.995 \ (5.318 \times 10^{-4})}$ & $0.607 \ (9.097 \times 10^{-3})$ & $\mathbf{0.496 \ (9.277 \times 10^{-3})}$ & $0.400 \ (8.13 \times 10^{-3})$ \\
& & & & & STS & $\mathbf{0.995 \ (5.318 \times 10^{-4})}$ & $\mathbf{0.699 \ (8.249 \times 10^{-3})}$ & $\mathbf{0.496 \ (9.277 \times 10^{-3})}$ & $\mathbf{0.422 \ (8.368 \times 10^{-3})}$ \\
\cline{5-10}
& & & & \multirow{2}{*}{SPS} & RTS& $0.995$ ($5.313 \times 10^{-4}$) & $0.626$ ($8.963 \times 10^{-3}$) & $0.496$ ($9.27 \times 10^{-3}$) & $0.401$ ($8.042 \times 10^{-3}$) \\
& & & & & STS & $0.995$ ($5.313 \times 10^{-4}$) & $0.733$ ($7.884 \times 10^{-3}$) & $0.496$ ($9.27 \times 10^{-3}$) & $0.419$ ($8.284 \times 10^{-3}$) \\
\cline{4-10}

& & & \multirow{4}{*}{$K=3$} & \multirow{2}{*}{DPS} & RTS & $\mathbf{0.995 \ (5.261 \times 10^{-4})}$ & $0.592 \ (8.566 \times 10^{-3})$ & $\mathbf{0.486 \ (9.152 \times 10^{-3})}$ & $\mathbf{0.220 \ (6.202 \times 10^{-3})}$ \\
& & & & & STS & $\mathbf{0.995 \ (5.261 \times 10^{-4})}$ & $\mathbf{0.832 \ (6.147 \times 10^{-3})}$ & $\mathbf{0.486 \ (9.152 \times 10^{-3})}$ & $0.408$ ($8.321 \times 10^{-3}$) \\
\cline{5-10}
& & & & \multirow{2}{*}{SPS} & RTS& $0.995$ ($5.221 \times 10^{-4}$) & $0.601$ ($8.531 \times 10^{-3}$) & $0.486$ ($9.152 \times 10^{-3}$) & $0.202$ ($5.803 \times 10^{-3}$) \\
& & & & & STS & $0.995$ ($5.221 \times 10^{-4}$) & $0.865$ ($5.338 \times 10^{-3}$) & $0.486$ ($9.152 \times 10^{-3}$) & $0.392$ ($8.261 \times 10^{-3}$) \\
\hline
\end{tabular}
}
\end{table}

\clearpage
\newpage

\begin{figure}[!h]
\center
\includegraphics[scale=0.2]{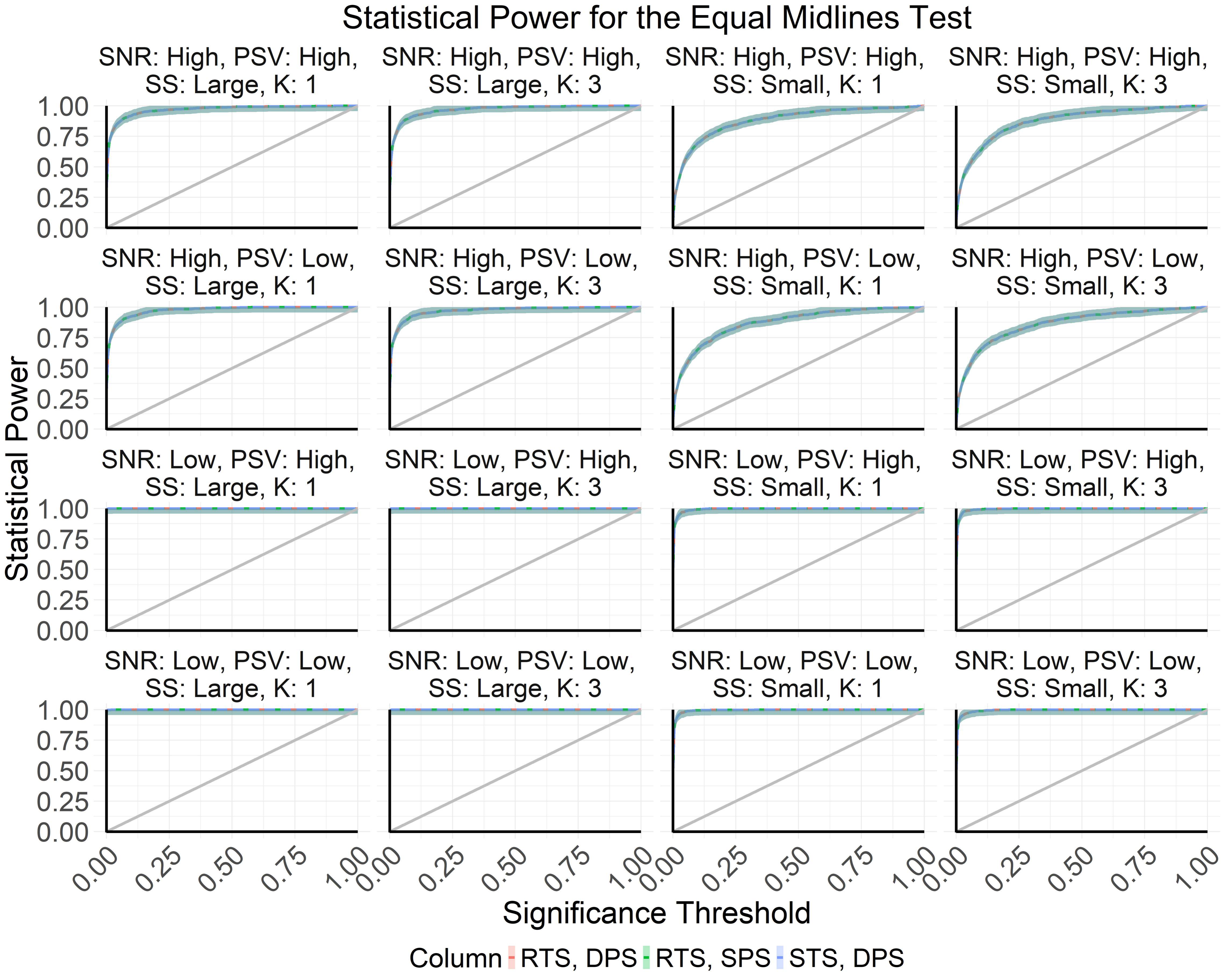}
    \caption{Empirical power curves used to compute $\mathrm{AUC}_{\mathrm{SP}}$ for the equal midlines test. Our RTS method and the STS method obtained the same performance in each simulation setting. Here, ``DPS'' denotes curves computed from datasets generated to have ``different phase-shift'' parameters across individuals, and ``SPS'' denotes curves computed from datasets generated to have ``same phase-shift'' parameters for every individual.} \label{fig:power_32M}
\end{figure}

\clearpage
\newpage

\begin{figure}[!h]
\center
\includegraphics[scale=0.2]{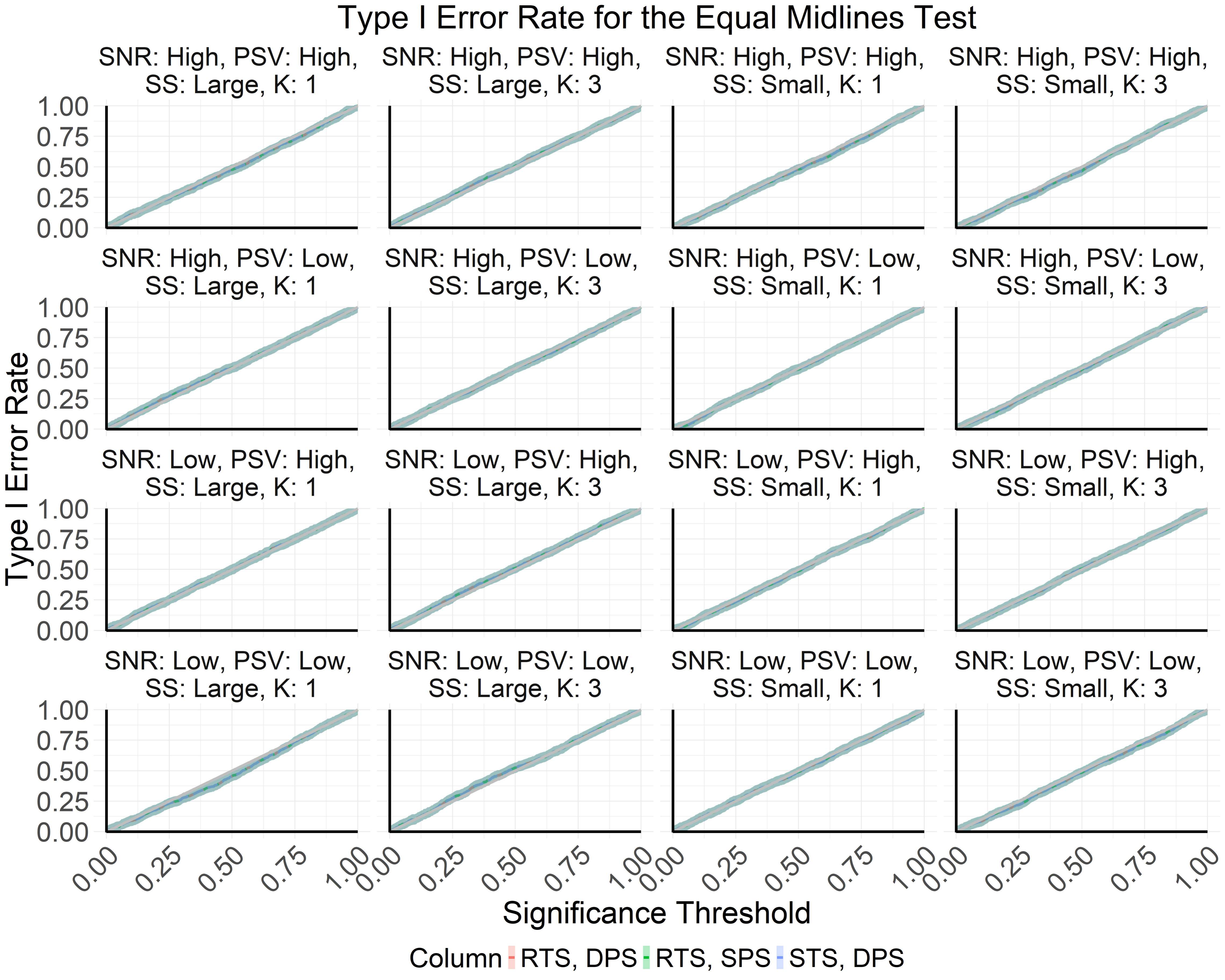}
    \caption{Empirical type I error curves used to compute $\mathrm{AUC}_{\mathrm{T1E}}$ for the equal midlines test. Our RTS method and the STS method obtained the same performance in each simulation setting. Here, ``DPS'' denotes curves computed from datasets generated to have ``different phase-shift'' parameters across individuals, and ``SPS'' denotes curves computed from datasets generated to have ``same phase-shift'' parameters for every individual.} \label{fig:typeI_32M}
\end{figure}

\clearpage
\newpage

\begin{figure}[!h]
\center
\includegraphics[scale=0.2]{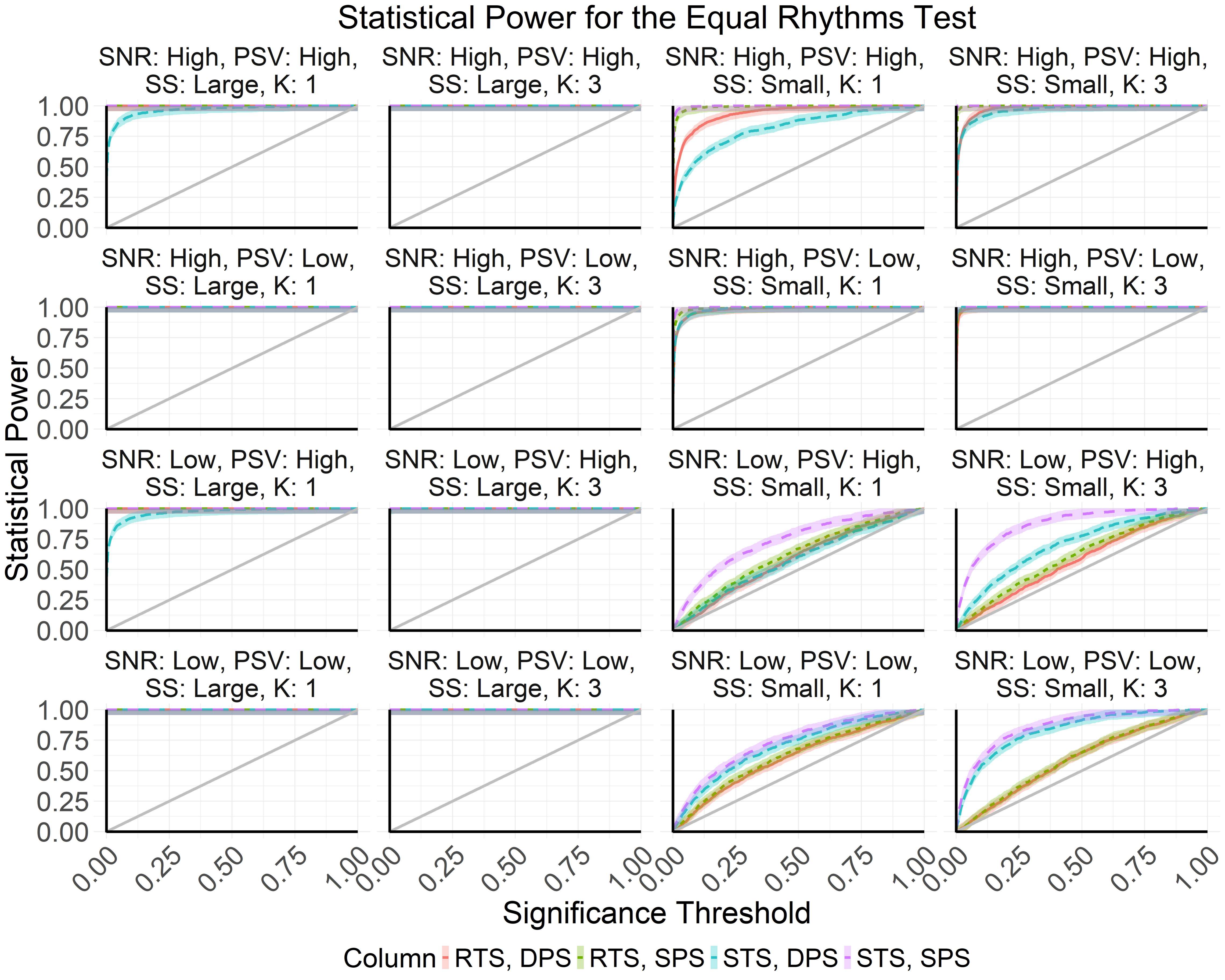}
    \caption{Empirical power curves used to compute $\mathrm{AUC}_{\mathrm{SP}}$ for the equal rhythms test. Our RTS method generally outperforms the STS method in every simulation setting, except in scenarios with low small sample sizes and low signal-to-noise ratios.} \label{fig:power_32R}
\end{figure}

\clearpage
\newpage

\begin{figure}[!h]
\center
\includegraphics[scale=0.2]{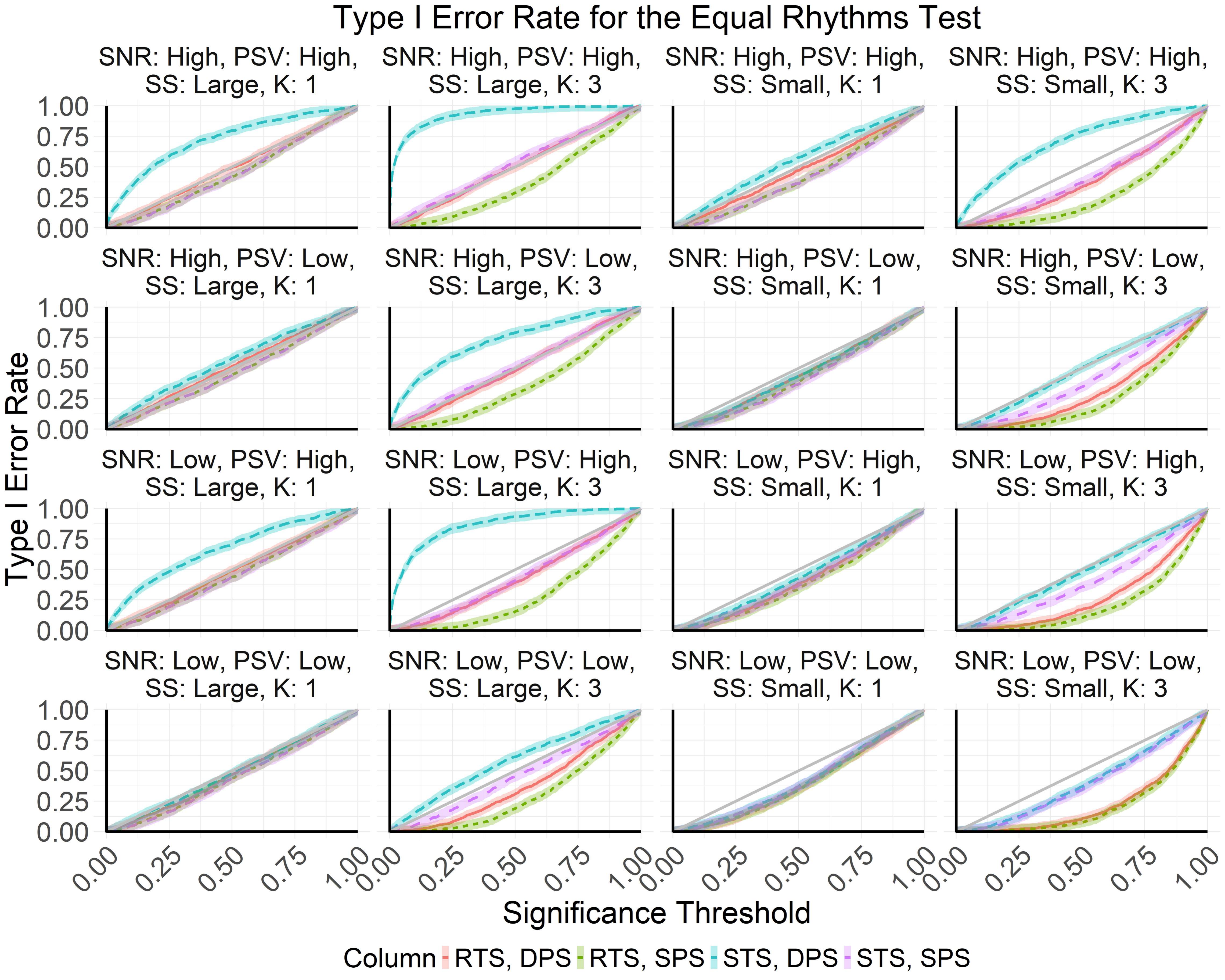}
    \caption{Empirical type I error curves used to compute $\mathrm{AUC}_{\mathrm{T1E}}$ for the equal rhythms test. The STS method appears to lose type I error control, particularly when the variability in individual-level phase-shifts across cohorts is relatively large. Our RTS method maintains type I error control in every simulation setting.} \label{fig:typeI_32R}
\end{figure}

\clearpage
\newpage

\begin{table*}[!h]
	\caption{Comparison of population-level parameter estimates ($\hat{\beta}^{(c)}$) and $p$-values from the zero amplitudes test obtained using our RTS method versus the STS method on cortisol level data. The amplitude estimates produced by our RTS method are consistently larger than those produced by the STS method, and the $p$-values produced by the RTS method are smaller for the cohort with major depressive disorder (MDD). The $p$-values obtained from the control cohort are equal.} \label{tab:ind}
 \centering
 \resizebox{1.05\textwidth}{!}{
		\begin{tabular}{|l|c|c|c|c|c|c|c|c|c|c|c|}
			\hline
   \multirow{1}{*}{Cohort} & \multirow{1}{*}{Method} & \multicolumn{1}{c|}{$\hat{\beta}_0^{(c)}$} & \multicolumn{1}{c|}{$\hat{\beta}_1^{(c)}$} & \multicolumn{1}{c|}{$\hat{\beta}_2^{(c)}$} & \multicolumn{1}{c|}{$\hat{\beta}_3^{(c)}$} & \multicolumn{1}{c|}{$\hat{\beta}_4^{(c)}$} & \multicolumn{1}{c|}{$\hat{\beta}_5^{(c)}$} & \multicolumn{1}{c|}{$\hat{\beta}_6^{(c)}$} & \multicolumn{1}{c|}{$p$-value} \\
   \hline
         \multirow{2}{*}{Control} & RTS & 1.635 & 0.893  & 3.073 & 0.563  & 1.987 & 0.190  & -2.087 &  0.000 \\
         & STS & 1.635 & 0.836 &  3.081 & 0.360 &  2.034 & 0.132 & -2.201  & 0.000 \\ 
         \hline 
         \multirow{2}{*}{MDD} & RTS & 1.846 & 0.885 &  3.126 & 0.383 &  1.862 & 0.248 &   0.131 &  0.004 \\
         & STS & 1.846 & 0.724  & 3.088 & 0.244 &   1.779 & 0.067 & -1.154  & 0.012 \\
         \hline 
\end{tabular} }
\end{table*}

\clearpage
\newpage

\begin{figure*}[!h]
\centering
\includegraphics[scale=0.25]{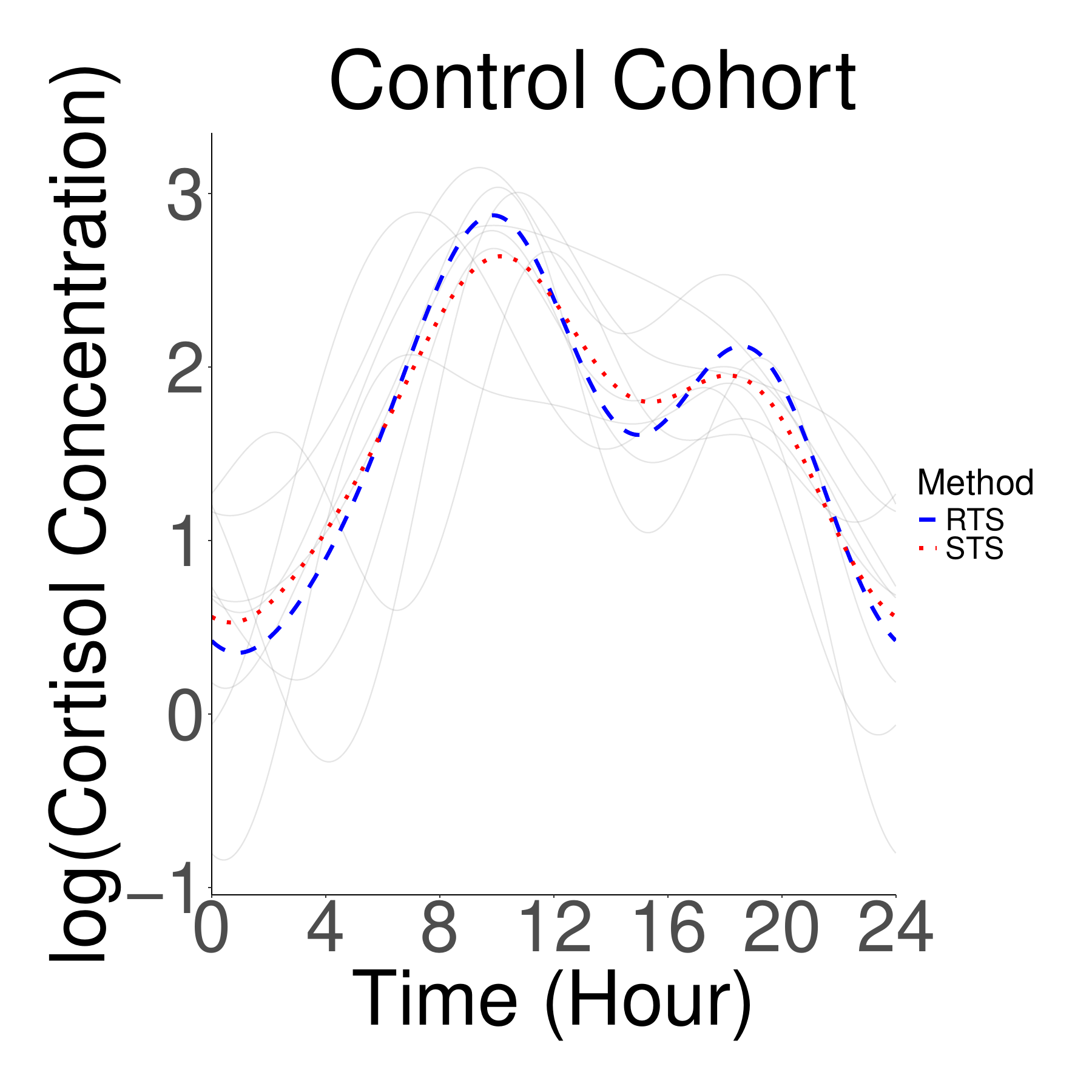}
\includegraphics[scale=0.25]{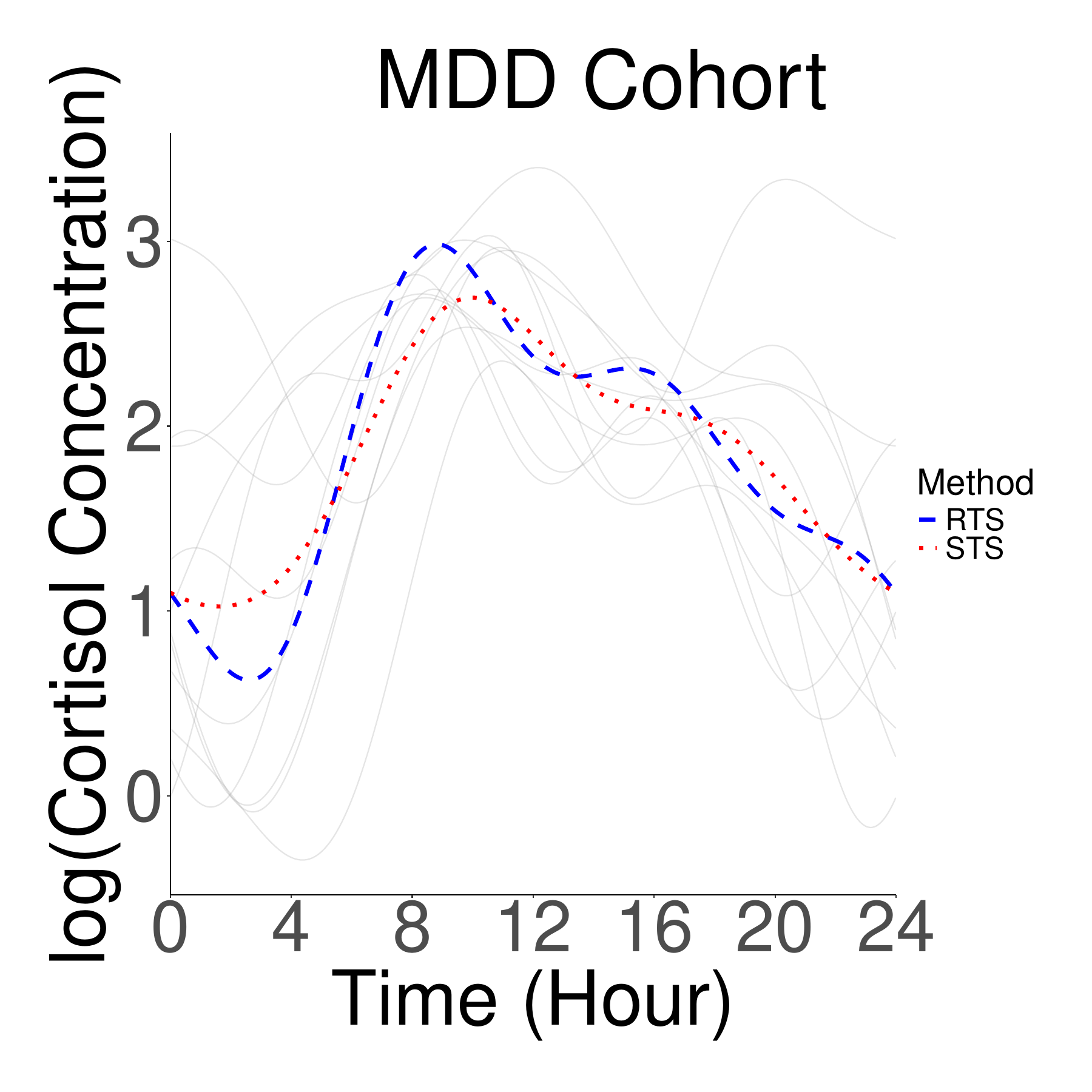}
\caption{Population-level curves computed from each method on cortisol level data. The solid blue line represents the population-level curve fit with our RTS method, while the dashed red line represents the population-level curve fit with the STS method. Each gray line represents corresponding individual-level fits.}
    \label{fig:amps}
\end{figure*}

\clearpage
\newpage

\begin{table*}[!h]
	\caption{Comparison of population-level parameter estimates ($\hat{\beta}^{(c)}$) and $p$-values from the zero amplitudes test obtained using our RTS method versus the STS method on heart rate data. The amplitude estimates produced by our RTS method are consistently larger than those produced by the STS method, and the $p$-values produced by the RTS method are smaller for the cohort with major depressive disorder (MDD). The $p$-values obtained from the control cohort are equal.} \label{tab:ind2}
 \centering
 \resizebox{1.05\textwidth}{!}{
		\begin{tabular}{|l|c|c|c|c|c|c|c|c|c|c|c|}
			\hline
Quantity   & RTS Method, Low-Stress Cohort & STS Method, Low-Stress Cohort & RTS Method, High-Stress Cohort & STS  Method, High-Stress Cohort \\
   \hline 
$\beta^{(c)}_0$     &    76.422   &   76.422       &     76.338       &     76.338 \\
$\beta^{(c)}_1$     &    11.470   &   10.152       &     12.410       &     11.764 \\
$\beta^{(c)}_2$   &     2.528   &    2.488       &      2.152       &     2.254 \\
$\beta^{(c)}_3$     &     8.119   &    6.530       &      6.642       &     5.822 \\
$\beta^{(c)}_4$   &     1.291   &    1.382       &      1.130       &     1.263 \\
$\beta^{(c)}_5$     &     3.999   &    2.693       &      4.730       &     3.105 \\
$\beta^{(c)}_6$   &    -0.589   &   -0.492       &     -1.713       &     -1.283 \\
$\beta^{(c)}_7$   &     4.212   &    2.479       &      4.862       &     3.715 \\
$\beta^{(c)}_8$      &    -2.764   &   -3.034       &      2.462       &     2.510 \\
$\beta^{(c)}_9$       &     2.878  &     1.949       &      3.791       &     0.760 \\
$\beta^{(c)}_{10}$       &     0.820  &     0.692       &     -0.598       &     -0.275 \\
$\beta^{(c)}_{11}$       &     2.245  &     0.801       &      2.890       &     1.217 \\
$\beta^{(c)}_{12}$       &   -0.896   &    0.003       &     -0.506       &     -0.572 \\
$\beta^{(c)}_{13}$       &     2.949  &     0.965       &      2.049       &     0.774 \\
$\beta^{(c)}_{14}$       &   -1.911   &   -1.944       &      2.413       &     2.307 \\
$\beta^{(c)}_{15}$       &     2.503  &     0.999       &      2.423       &     0.524 \\
$\beta^{(c)}_{16}$       &     1.681  &     1.808       &      1.511       &     1.940 \\
$\beta^{(c)}_{17}$       &     2.471  &     0.446       &      2.481       &     0.695 \\
$\beta^{(c)}_{18}$       &   -0.657   &   -1.461       &      1.683       &     1.686 \\
$\beta^{(c)}_{19}$       &     2.165  &     0.363       &      1.202       &     0.325 \\
$\beta^{(c)}_{20}$       &     0.553  &     0.120       &     -0.196       &     -1.052 \\
$\beta^{(c)}_{21}$       &     2.271  &     0.889       &      1.589       &     0.658 \\
$\beta^{(c)}_{22}$       &   -0.038   &   -0.102       &      0.110       &     -0.020 \\
$\beta^{(c)}_{23}$       &     1.601  &     0.309       &      1.792       &     0.354 \\
$\beta^{(c)}_{24}$       &     1.267  &    -0.141       &      1.651       &     1.952 \\
$\beta^{(c)}_{25}$       &     2.006  &     0.373       &      2.082       &     0.847 \\
$\beta^{(c)}_{26}$       &   -1.137   &   -1.114       &      0.069       &     -0.044 \\
$p$-value       &     0.000  &     0.000       &      0.000       &     0.000 \\
         \hline 
\end{tabular} }
\end{table*}

\clearpage
\newpage

\begin{figure*}[!h]
\centering
\includegraphics[scale=0.25]{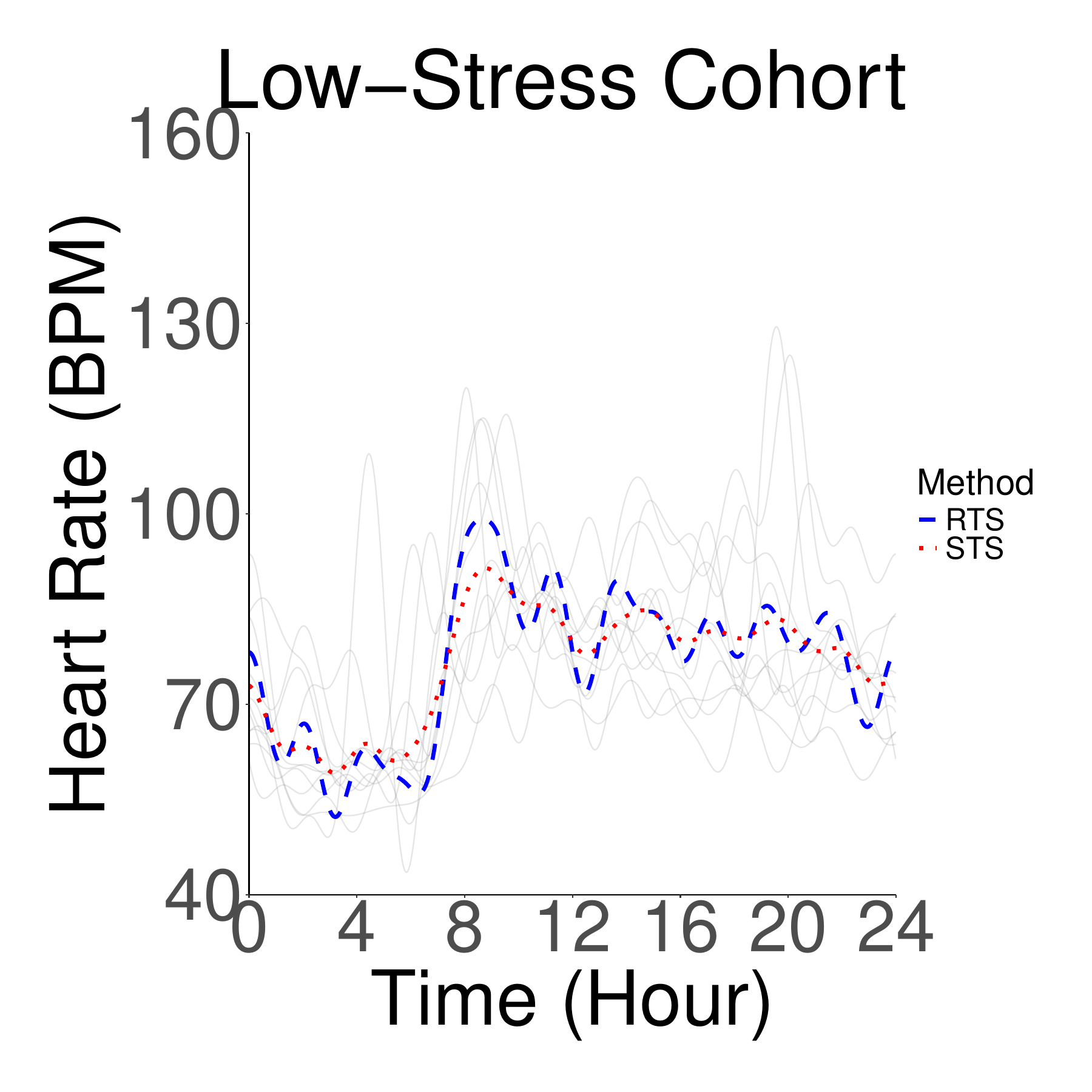}
\includegraphics[scale=0.25]{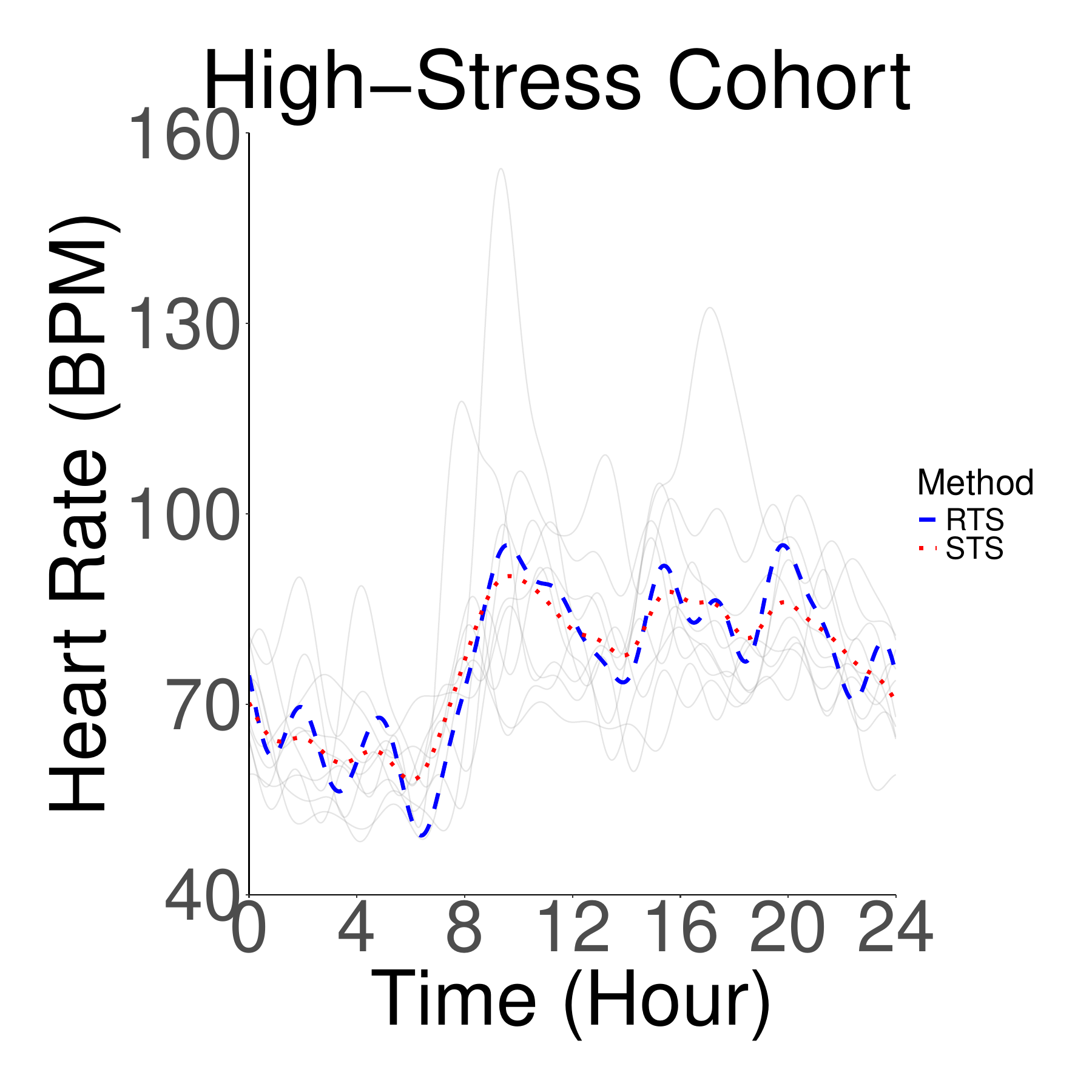}
\caption{Population-level curves computed from each method on heart rate data. The solid blue line represents the population-level curve fit with our RTS method, while the dashed red line represents the population-level curve fit with the STS method. Each gray line represents corresponding individual-level fits.}
    \label{fig:amps2}
\end{figure*}

\clearpage 
\newpage

\appendix

\section{Derivations for Parameter Estimation} \label{App:A}

\subsection{Derivation for Proposition \ref{prop:1}} \label{App:A1}

\begin{proof}
To simplify presentation, we omit the superscript $(c)$. We first note that the parameter estimates $\hat{\gamma}_i\sim \mathrm{N}(\gamma_i, \Sigma_i/n_i)$ by definition of the central limit theorem. Further, the identities in (\ref{eq:alt_to_orig}) imply
\begin{align}
    \gamma_{i,0} &= \beta_0+b_{i,0}, \nonumber \\
    \gamma_{i,2k-1} &= -(\beta_{2k-1}+b_{i,2k-1})\sin(\beta_{2k}+b_{i,2k}), \label{eq:a1_ind} \\
    \gamma_{i,2k} &= (\beta_{2k-1}+b_{i,2k-1})\cos(\beta_{2k}+b_{i,2k}), \nonumber
\end{align}
where the $(2K+1)\times 1$ vector $b_i$ is generated by a probability distribution with mean zero. It is trivial that $\mathbb{E}(\hat{\gamma}_{i,0})=\beta_0$. For $\hat{\alpha}_{2k-1}$, we find
\begin{align}
\mathbb{E}(\hat{\alpha}_{2k-1}) &= \mathbb{E}\left(\frac{1}{M}\sum_{i=1}^{M}\hat{\gamma}_{i,2k-1}\right) \nonumber \\
&= \frac{1}{M}\sum_{i=1}^{M}\mathbb{E}(\hat{\gamma}_{i, 2k-1}) \nonumber \\
&= \frac{1}{M}\sum_{i=1}^{M}\mathbb{E}(\gamma_{i, 2k-1}) \nonumber  \\
&= \mathbb{E}\{-(\beta_{2k-1}+b_{i,2k-1})\sin(\beta_{2k} + b_{i,2k})\} \label{eq:lema1_1}  \\
    &= \mathbb{E}[-(\beta_{2k-1} + b_{i,2k-1})\{\sin(\beta_{2k})\cos(b_{i,2k}) + \cos(\beta_{2k})\sin(b_{i,2k})\}] \label{eq:lema1_2} \\
    &=-\beta_{2k-1}\mathbb{E}\{\sin(\beta_{2k})\cos(b_{i,2k}) + \cos(\beta_{2k})\sin(b_{i,2k})\} \nonumber \\
    &\quad \quad -\mathbb{E}(b_{i,2k-1})\mathbb{E}\{\sin(\beta_{2k})\cos(b_{i,2k}) + \cos(\beta_{2k})\sin(b_{i,2k})\} \label{eq:lema1_3} \\
    &=-\beta_{2k-1}[\sin(\beta_{2k})\mathbb{E}\{\cos(b_{i,2k})\} +\cos(\beta_{2k})\mathbb{E}\{\sin(b_{i,2k})\}] \nonumber \\
    &=\alpha_{2k-1}\mathbb{E}\{\cos(b_{i,2k})\} -\alpha_{2k}\mathbb{E}\{\sin(b_{i,2k})\}. \label{eq:lema1_4}
\end{align}
Here, (\ref{eq:lema1_1}) is by application of the individual-level identities in (\ref{eq:a1_ind}), (\ref{eq:lema1_2}) results from the identity $\sin(Z_1+Z_2) = \sin(Z_1)\cos(Z_2) + \cos(Z_1)\sin(Z_2)$, (\ref{eq:lema1_3}) is due to the assumption that each $b_{i,j}$ is independent of $b_{i,k}$ for all $j\neq k$, and (\ref{eq:lema1_4}) is by application of the population-level identities in (\ref{eq:alt_to_orig}).

Similarly, for $\hat{\alpha}_{2k}$ we find
\begin{align*}
    \mathbb{E}(\hat{\alpha}_{2k}) 
&= \frac{1}{M}\sum_{i=1}^{M}\mathbb{E}(\gamma_{i, 2k}) \\
&= \beta_{2k-1}\mathbb{E}\{\cos(\beta_{2k} + b_{i,2k})\} \\
    &=\beta_{2k-1}[\cos(\beta_{2k})\mathbb{E}\{\cos(b_{i,2k})\} -\sin(\beta_{2k})\mathbb{E}\{\sin(b_{i,2k})\}] \\
    &=\alpha_{2k}\mathbb{E}\{\cos(b_{i,2k})\} +\alpha_{2k-1}\mathbb{E}\{\sin(b_{i,2k})\}.
\end{align*}
\end{proof}

\subsection{Derivation for Corollary \ref{cor:1}} \label{App:A2}

We omit the superscript $(c)$ to simplify presentation. We first compute the $k$-th population-level amplitude given the quantities obtained in Proposition \ref{prop:1}. By application of the identities in (\ref{eq:alt_to_orig}), we find
\begin{align*}
\mathbb{E}(\hat{\alpha}_{2k-1}) &= -\beta_{2k-1}[\sin(\beta_{2k})\mathbb{E}\{\cos(b_{i,2k})\} +\cos(\beta_{2k})\mathbb{E}\{\sin(b_{i,2k})\}] \\ 
&=-\beta_{2k-1}\mathbb{E}\{\sin(\beta_{2k}+b_{i,2k})\}, \\
\mathbb{E}(\hat{\alpha}_{2k}) &= \beta_{2k-1}[\cos(\beta_{2k})\mathbb{E}\{\cos(b_{i,2k})\} -\sin(\beta_{2k})\mathbb{E}\{\sin(b_{i,2k})\}] \\ 
&=\beta_{2k-1}\mathbb{E}\{\cos(\beta_{2k}+b_{i,2k})\}.
\end{align*}
As a consequence, the expression
\begin{align}
    \sqrt{\mathbb{E}(\hat{\alpha}_{2k-1})^2+\mathbb{E}(\hat{\alpha}_{2k})^2} &= \beta_{2k-1}\sqrt{\mathbb{E}\{\sin(\beta_{2k}+b_{i,2k})\}^2 + \mathbb{E}\{\cos(\beta_{2k}+b_{i,2k})\}^2} \nonumber \\
    &= \beta_{2k-1}|\mathbb{E}[\exp\{z(\beta_{2k}+b_{i,2k})\}]| \label{eq:a2_euler} \\
    &= \beta_{2k-1}|\exp(z\beta_{2k})| |\mathbb{E}\{\exp(zb_{i,2k})\}| \nonumber \\
    &=\beta_{2k-1}|\phi_{b_{i,2k}}(1)|, \nonumber
\end{align}
where $z=\sqrt{-1}$, with (\ref{eq:a2_euler}) due to Euler's identity or $\exp(zW) = \cos(W)+z\sin(W)$ given an argument $W$. Now, for the $k$-th phase-shift, we find
\begin{align*}
    \mathrm{atan2}\left\{-\mathbb{E}(\hat{\alpha}_{2k-1}), \mathbb{E}(\hat{\alpha}_{2k})\right\} &=\mathrm{atan2}[\beta_{2k-1}\mathbb{E}\{\sin(\beta_{2k}+b_{i,2k})\}, \beta_{2k-1}\mathbb{E}\{\cos(\beta_{2k}+b_{i,2k})\}] \\
    & = \mathrm{atan2}\bigg[\sin(\beta_{2k})\mathbb{E}\{\cos( b_{i,2k})\} + \cos(\beta_{2k})\mathbb{E}\{\sin(b_{i,2k})\}, \\
    & \quad \quad \quad \quad \ \cos(\beta_{2k}) \mathbb{E}\{\cos(b_{i,2k})\} - \sin(\beta_{2k}) \mathbb{E}\{\sin(b_{i,2k})\} \bigg].
\end{align*}

\section{Derivations for Test Statistic Calculation} \label{App:B}
\subsection{Supporting Lemmas} \label{App:B1}

\begin{lem}[Lemma 1.7, \citealt{Tsybakov2009}] \label{lem:1}  Suppose each $n_i^{(c)} = n$ and each sample $X^{(c)}_{i,j} = 12(j-1)/(n\pi)$ for all $i$. Then
\begin{align*}
(W^{(c)}_i)^T (W^{(c)}_i) &= \mathrm{diag}_{2K+1}\left(n, \frac{n}{2},\ldots, \frac{n}{2}\right),
\end{align*}
where $\mathrm{diag}_p(m_1,\ldots, m_p)$ denotes a $p\times p$ diagonal matrix with $m_k$ representing the $k$-th element along the diagonal, and
\begin{align*}
W^{(c)}_i = \begin{bmatrix}
1 & \sin\left(\frac{\pi X^{(c)}_{i,1}}{12}\right) & \cos\left(\frac{\pi X^{(c)}_{i,1}}{12}\right) & \cdots & \sin\left(\frac{K\pi X^{(c)}_{i,1}}{12}\right) & \cos\left(\frac{K\pi X^{(c)}_{i,1}}{12}\right) \\
\vdots & \vdots & \vdots & \ddots & \vdots & \vdots \\
1 & \sin\left(\frac{\pi X^{(c)}_{i,n_i}}{12}\right) & \cos\left(\frac{\pi X^{(c)}_{i,n_i}}{12}\right) & \cdots & \sin\left(\frac{K\pi X^{(c)}_{i,n_i}}{12}\right) & \cos\left(\frac{K\pi X^{(c)}_{i,n_i}}{12}\right)
\end{bmatrix}.
\end{align*}    
\end{lem}

\begin{lem} \label{lem:2}
Suppose the assumptions of Lemma \ref{lem:1} and Proposition \ref{prop:1} are valid. If the parameters of a correctly specified first-th order trigonometric regression model are estimated, then the covariance matrix for a population-level parameter vector $\hat{\alpha}^{(c)}$ can be expressed as
\begin{align*}
 \mathrm{Var}(\hat{\alpha}^{(c)}) &= \frac{1}{M^{(c)}}\left(\bar{D}^{(c)} + \frac{1}{M^{(c)}}\sum_{i=1}^{M^{(c)}}\Sigma^{(c)}_i\right),
\end{align*}
where the distinct elements of $\bar{D}$ can be expressed as
\begin{align*}
    \bar{D}_{1,1} &= \mathrm{Var}(b_{i,0}), \\
    \bar{D}_{1,2} &= \mathbb{E}[\{(\gamma_{i,0} - \mathbb{E}(\gamma_{i,0})\}\{(\gamma_{i,1} - \mathbb{E}(\gamma_{i,1})\}] = \mathbb{E}(b_{i,0})\mathbb{E}\{(\gamma_{i,1} - \mathbb{E}(\gamma_{i,1})\}=0, \\
    \bar{D}_{1,3} &= \mathbb{E}[\{(\gamma_{i,0} - \mathbb{E}(\gamma_{i,0})\}\{(\gamma_{i,2} - \mathbb{E}(\gamma_{i,2})\}] = \mathbb{E}(b_{i,0})\mathbb{E}\{(\gamma_{i,2} - \mathbb{E}(\gamma_{i,2})\}=0, \\
    \bar{D}_{2,2} &= \{\beta_1^2 + \mathrm{Var}(b_{i,1})\}\left\{\frac{1-\cos(2\beta_2)\phi_{b_{i,2}}(2)}{2}\right\}-\beta_1^2\sin^2(\beta_2)\phi^2_{b_{i,2}}(1), \\
    \bar{D}_{2,3} &= -\sin(\beta_2)\cos(\beta_2)\phi_{b_{i,2}}(2)\{\beta^2_1+\mathrm{Var}(b_{i,1})\} + \beta_{1}^2\left\{\sin(\beta_2)\phi_{b_{i,2}}(1)\right\}\left\{\cos(\beta_2)\phi_{b_{i,2}}(1)\right\}, \\
    \bar{D}_{3,3} &= \left\{\frac{1+\cos(2\beta_2)\phi_{b_{i,2}}(2)}{2}\right\}\{\beta_1^2 + \mathrm{Var}(b_{i,1})\}-\beta_1^2\cos^2(\beta_2)\phi^2_{b_{i,2}}(1),
\end{align*}
and
\begin{align*}
    \Sigma^{(c)}_i &= \mathrm{diag}\left\{\frac{(\sigma_i^{(c)})^2}{n}, \frac{2(\sigma_i^{(c)})^2}{n}, \frac{2(\sigma_i^{(c)})^2}{n}\right\}.
\end{align*}
\end{lem}

\begin{proof}
We omit the superscript $(c)$ to simplify presentation. The derivation for $\Sigma_i$ follows from Lemma \ref{lem:1} given that $\gamma_i$ is estimated by minimizing squared loss. We compute the elements for the upper triangular of $\bar{D}$. For elements in the first row of $\bar{D}$, we find
\begin{align*}
    \bar{D}_{1,1} &= \mathrm{Var}(b_{i,0}), \\
    \bar{D}_{1,2} &= \mathbb{E}[\{(\gamma_{i,0} - \mathbb{E}(\gamma_{i,0})\}\{(\gamma_{i,1} - \mathbb{E}(\gamma_{i,1})\}] = \mathbb{E}(b_{i,0})\mathbb{E}\{(\gamma_{i,1} - \mathbb{E}(\gamma_{i,1})\}=0, \\
    \bar{D}_{1,3} &= \mathbb{E}[\{(\gamma_{i,0} - \mathbb{E}(\gamma_{i,0})\}\{(\gamma_{i,2} - \mathbb{E}(\gamma_{i,2})\}] = \mathbb{E}(b_{i,0})\mathbb{E}\{(\gamma_{i,2} - \mathbb{E}(\gamma_{i,2})\}=0.
\end{align*}

For the diagonal element $\bar{D}_{2,2}$ in the second row, note that
\begin{align*}
    \bar{D}_{2,2} &= \mathrm{Var}\{-(\beta_1+b_{i,1})\sin(\beta_2 + b_{i,2})\} \\
    &= \mathbb{E}\{(\beta_1+b_{i,1})^2\}\mathbb{E}\{\sin(\beta_2 + b_{i,2})^2\} - \mathbb{E}(\beta_1+b_{i,1})^2\mathbb{E}\{\sin(\beta_2 + b_{i,2})\}^2 \\
    &=\{\beta_1^2 + \mathrm{Var}(b_{i,1})\}\mathbb{E}\{\sin(\beta_2 + b_{i,2})^2\}-\beta_1^2\mathbb{E}\{\sin(\beta_2 + b_{i,2})\}^2.
\end{align*}
Here,
\begin{align*}
    \mathbb{E}\{\sin(\beta_2 + b_{i,2})^2\} &= \mathbb{E}\{\cos^2(b_{i,2})\sin^2(\beta_2) + 2\cos(\beta_2)\cos(b_{i,2})\sin(\beta_2)\sin(b_{i,2})+\cos^2(\beta_2)\sin^2(b_{i,2})\} \\
    &=\sin^2(\beta_2)\mathbb{E}\{\cos^2(b_{i,2})\} + \cos^2(\beta_2)\mathbb{E}\{\sin^2(b_{i,2})\} \\
    &=\frac{\sin^2(\beta_2)[1+\mathbb{E}\{\cos(2b_{i,2})\}]}{2} + \frac{\cos^2(\beta_2)[1-\mathbb{E}\{\cos(2b_{i,2})\}]}{2} \\
    &=\frac{\sin^2(\beta_2)\{1+\phi_{b_{i,2}}(2)\}}{2}  + \frac{\cos^2(\beta_2)\{1-\phi_{b_{i,2}}(2)\}}{2} \\
    &=\frac{1-\cos(2\beta_2)\phi_{b_{i,2}}(2)}{2}.
\end{align*}
and
\begin{align*}
    \mathbb{E}\{\sin(\beta_2 + b_{i,2})\}^2 &= \mathbb{E}\{\cos(\beta_2)\cos(b_{i,2}) + \cos(b_{i,2})\sin(\beta_2)\}^2 \\
    &=\mathbb{E}\{\cos(b_{i,2})\sin(\beta_2)\}^2 \\
    &=\sin^2(\beta_2)\phi^2_{b_{i,2}}(1),
\end{align*}
which yields
\begin{align*}
    \bar{D}_{2,2} &= \{\beta_1^2 + \mathrm{Var}(b_{i,1})\}\left\{\frac{1-\cos(2\beta_2)\phi_{b_{i,2}}(2)}{2}\right\}-\beta_1^2\sin^2(\beta_2)\phi^2_{b_{i,2}}(1).
\end{align*}

For the off diagonal element $\bar{D}_{2,3}$ in the second row, it follows that
\begin{align*}
    \bar{D}_{2,3} &= \mathbb{E}[\{\gamma_{i,1}-\mathbb{E}(\gamma_{i,1})\}\{\gamma_{i,2}-\mathbb{E}(\gamma_{i,2})\}] \\ 
    &= \mathbb{E}[\gamma_{i,1}\{\gamma_{i,2}-\mathbb{E}(\gamma_{i,2})\}] \\
    &= \mathbb{E}[-(\beta_1+b_{i,1})^2\sin(\beta_2+b_{i,2})\cos(\beta_2+b_{i,2})] \\
    & \quad \quad + \beta_1^2\mathbb{E}[\sin(\beta_2+b_{i,2})\}\mathbb{E}\{\cos(\beta_2+b_{i,2})\}],
\end{align*}
with
\begin{align*}
    \mathbb{E}\{-(\beta_1+b_{i,1})^2\sin(\beta_2+b_{i,2})\cos(\beta_2+b_{i,2})\} &= -\mathbb{E}\{(\beta_1+b_{i,1})^2\}\mathbb{E}\left[\sin(\beta_2)\cos(\beta_2)\{\cos^2(b_{i,2}) - \sin^2(b_{i,2})\} \right] \\
    &= -\{\beta^2_1+\mathrm{Var}(b_{i,1})\}\{\sin(\beta_2)\cos(\beta_2)\mathbb{E}(\cos(2b_{i,2})\} \\
    &= -\sin(\beta_2)\cos(\beta_2)\phi_{b_{i,2}}(2)\{\beta^2_1+\mathrm{Var}(b_{i,1})\}.
\end{align*}
and
\begin{align*}
    \beta_1^2\mathbb{E}[\sin(\beta_2+b_{i,2})\}\mathbb{E}\{\cos(\beta_2+b_{i,2})\}] &= \beta_{1}^2\left\{\sin(\beta_2)\phi_{b_{i,2}}(1)\right\}\left\{\cos(\beta_2)\phi_{b_{i,2}}(1)\right\},
\end{align*}
which yields
\begin{align*}
    \bar{D}_{2,3} &= -\sin(\beta_2)\cos(\beta_2)\phi_{b_{i,2}}(2)\{\beta^2_1+\mathrm{Var}(b_{i,1})\} + \beta_{1}^2\left\{\sin(\beta_2)\phi_{b_{i,2}}(1)\right\}\left\{\cos(\beta_2)\phi_{b_{i,2}}(1)\right\}.
\end{align*}

Finally, the derivation for $\bar{D}_{3,3}$ follows the derivation of $\bar{D}_{2,2}$, with
\begin{align*}
    \bar{D}_{3,3} &= \mathrm{Var}\{(\beta_1+b_{i,1})\cos(\beta_2 + b_{i,2})\} \\
    &= \mathbb{E}\{(\beta_1+b_{i,1})^2\}\mathbb{E}\{\cos(\beta_2 + b_{i,2})^2\} - \mathbb{E}(\beta_1+b_{i,1})^2\mathbb{E}\{\cos(\beta_2 + b_{i,2})\}^2 \\
    &=\{\beta_1^2 + \mathrm{Var}(b_{i,1})\}\mathbb{E}\{\cos(\beta_2 + b_{i,2})^2\}-\beta_1^2\mathbb{E}\{\cos(\beta_2 + b_{i,2})\}^2 \\
    &= \left\{\frac{1+\cos(2\beta_2)\phi_{b_{i,2}}(2)}{2}\right\}\{\beta_1^2 + \mathrm{Var}(b_{i,1})\}-\beta_1^2\cos^2(\beta_2)\phi^2_{b_{i,2}}(1).
\end{align*}
\end{proof}

\begin{lem} \label{lem:3}
Suppose the assumptions of Lemma \ref{lem:2} are valid. If the probability distribution generating individual-level phase-shifts is symmetric with mean zero, then the asymptotic covariance matrix for $\mathrm{Var}\{g(\hat{\alpha})\}$ can be expressed as
\begin{align*}
    \mathrm{Var}\{g(\hat{\alpha})\} &= \frac{D_{2,2}\alpha_1^2\phi^2_{b_{i,1}}(1)+D_{3,3}\alpha_2^2\phi^2_{b_{i,1}}(1) + 2D_{2,3}\alpha_1\alpha_2\phi^2_{b_{i,1}}(1)}{M\phi^2_{b_{i,1}}(1)\beta_1} \\
    & \quad \quad + \frac{2}{M^2n}\sum_{i=1}^{M}\sigma_i^2,
\end{align*}
where $g(\hat{\alpha}^{(c)}) = \sqrt{\alpha_1^2+\alpha_2^2}=\beta_1$.
\end{lem}

\begin{proof}
Recall that the Delta method \citep[Theorem 5.19]{Boos2013} indicates the asymptotic distribution of $g(\hat{\alpha})$ is given by
\begin{align*}
    g(\hat{\alpha}) \sim \mathrm{N}\left[g(\alpha), \frac{1}{M}\left\{G(\alpha)^TDG(\alpha) + \frac{1}{M}\sum_{i=1}^{M}G(\gamma_i)^T\Sigma_iG(\gamma_i)\right\}\right]
\end{align*}
when
\begin{align*}
\hat{\alpha}\sim \mathrm{N}\left\{\alpha, \frac{1}{M}\left(D + \frac{1}{M}\sum_{i=1}^{M}\Sigma_i\right)\right\}
\end{align*}
and there are no individual-level differences in phase-shift parameters. Here, the Jacobian $G(\alpha) = \delta g(\alpha)/\delta\alpha$. In the context of first-order trigonometric regression and the hypothesis test $H_0:\beta_1=0$, define the function $g(\alpha) = \sqrt{\alpha_1^2+\alpha_2^2}=\beta_1$, which implies
\begin{align*}
    G(\alpha) &= \begin{bmatrix} 0 & \frac{\alpha_1}{\sqrt{\alpha_1^2+\alpha_2^2}} & \frac{\alpha_2}{\sqrt{\alpha_1^2+\alpha_2^2}}
    \end{bmatrix}.
\end{align*}

We first derive $\sum_{i=1}^{M}\{G(\gamma_i)\Sigma_iG(\gamma_i)^T\}/M$, for which we find
\begin{align*}
    \frac{1}{M}\sum_{i=1}^{M}G(\gamma_i)\Sigma_iG(\gamma_i)^T&=\frac{1}{Mn}\sum_{i=1}^{M}\begin{bmatrix}0 \\ \frac{\gamma_{i,1}}{\sqrt{\gamma_{i,1}^2+\gamma_{i,2}^2}} \\ \frac{\gamma_{i,1}}{\sqrt{\gamma_{i,2}^2+\gamma_{i,2}^2}} \end{bmatrix}^T \begin{bmatrix} \sigma_i^2 & 0 & 0 \\
    0 & 2\sigma_i^2 & 0 \\ 0 & 0 & 2\sigma_i^2\end{bmatrix}\begin{bmatrix}0 \\ \frac{\gamma_{i,1}}{\sqrt{\gamma_{i,1}^2+\gamma_{i,2}^2}} \\ \frac{\gamma_{i,1}}{\sqrt{\gamma_{i,2}^2+\gamma_{i,2}^2}} \end{bmatrix} \\
    &= \frac{2}{Mn}\sum_{i=1}^{M}\sigma_i^2.
\end{align*}

We now derive $G(\alpha)^TDG(\alpha)$, for which we find
\begin{align}
    G(\alpha)DG(\alpha)^T &= \frac{D_{2,2}\alpha_1^2+D_{3,3}\alpha_2^2 + 2D_{2,3}\alpha_1\alpha_2}{\alpha_1^2+\alpha_2^2}. \label{eq:sym_Z}
\end{align}
Substitution of $\alpha$ with the expectation of the parameter estimates from Proposition \ref{prop:1} yields 
\begin{align*}
    \frac{1}{M}\left(D + \frac{1}{M}\sum_{i=1}^{M}\Sigma_i\right) &= \frac{D_{2,2}\alpha_1^2\phi^2_{b_{i,1}}(1)+D_{3,3}\alpha_2^2\phi^2_{b_{i,1}}(1) + 2D_{2,3}\alpha_1\alpha_2\phi^2_{b_{i,1}}(1)}{M\phi^2_{b_{i,1}}(1)\beta_1} \\
    & \quad \quad + \frac{2}{M^2n}\sum_{i=1}^{M}\sigma_i^2.
\end{align*}
\end{proof}

\subsection{Computation of Wald Test Statistic for Section \ref{sec:2.2.2}} \label{App:B2}

In Section \ref{sec:2.2.2}, we set $\beta_1=1/2$ and $\beta_2=0$, which imply that $\alpha_1 = 0$ and $\alpha_2 = 1/2$ by application of the identities in (\ref{eq:alt_to_orig}). We also set $\mathrm{Var}(b_{i,1})=0$ and $\phi_{b_{i,2}}(t) = \{\cos(t\pi/4) + \cos(-t\pi/4)\}/2$. Following Lemma \ref{lem:2}, we first compute 
\begin{align*}
    \bar{D}_{2,2} &= \{\beta_1^2 + \mathrm{Var}(b_{i,1})\}\left\{\frac{1-\cos(2\beta_2)\phi_{b_{i,2}}(2)}{2}\right\}-\beta_1^2\sin^2(\beta_2)\phi^2_{b_{i,2}}(1) \\
    &= \frac{1}{4}\left(\frac{1-1(0)}{2}\right) \\
    &=\frac{1}{8}, \\
    \bar{D}_{2,3} &=\sin(\beta_2)\cos(\beta_2)\phi_{b_{i,2}}(2)\{\beta^2_1+\mathrm{Var}(b_{i,1})\} + \beta_{1}^2\left\{\sin(\beta_2)\phi_{b_{i,2}}(1)\right\}\left\{\cos(\beta_2)\phi_{b_{i,2}}(1)\right\} \\
    &=0, \\
    \bar{D}_{3,3} &= \left\{\frac{1+\cos(2\beta_2)\phi_{b_{i,2}}(2)}{2}\right\}\{\beta_1^2 + \mathrm{Var}(b_{i,1})\}-\beta_1^2\cos^2(\beta_2)\phi^2_{b_{i,2}}(1) \\
    &= \frac{1}{8}-\frac{1}{8} \\
    &=0.
\end{align*}
It follows that
\begin{align*}
    \frac{1}{M}\left(D + \frac{1}{M}\sum_{i=1}^{M}\Sigma_i\right) &= \frac{D_{2,2}\alpha_1^2\phi^2_{b_{i,1}}(1)+D_{3,3}\alpha_2^2\phi^2_{b_{i,1}}(1) + 2D_{2,3}\alpha_1\alpha_2\phi^2_{b_{i,1}}(1)}{M\phi^2_{b_{i,1}}(1)\beta_1} \\
    & \quad \quad + \frac{2}{M^2n}\sum_{i=1}^{M}\sigma_i^2 \\
    &=\frac{D_{2,2}\alpha_1^2\phi^2_{b_{i,1}}(1)}{M\phi^2_{b_{i,1}}(1)\beta_1}  + \frac{1}{3M} \\
    &=\frac{(1/8)\{-\sin^2(0)\}}{M/2}+ \frac{1}{3M} \\
    &=\frac{1}{3M}.
\end{align*}
When it is instead the case that there is no individual-level phase-shift parameters, the test statistic further simplifies to
\begin{align*}
    \tau &=\frac{3M\phi^2_{b_{i,2}}(1)}{4} = \frac{3M}{8}.
\end{align*}

\newpage 

\bibliographystyle{apalike} 
\bibliography{bibliography}

\end{document}